\newif\iflong\longtrue   
\iflong
\newcommand{\longonly}[1]{#1}
\newcommand{\shortonly}[1]{}
\else
\newcommand{\longonly}[1]{}
\newcommand{\shortonly}[1]{#1}
\fi

\iflong
\documentclass[acmsmall,screen,prologue,table,nonacm]{acmart}
\setcopyright{none}
\else
\documentclass[acmsmall,screen,prologue,table]{acmart}
\fi

\newif\ifcomments\commentsfalse   

\ifcomments

\newcommand{\sz}[1]{\textcolor{olive!100!black!100}{{[SZ:~#1]}}}
\newcommand{\nr}[1]{\textcolor{blue!100!black!100}{{[NR:~#1]}}}
\else

\newcommand{\sz}[1]{}
\newcommand{\nr}[1]{}
\fi

\newcommand{\todo}[1]{}

\newcommand{\sep}{\ensuremath{\ \ | \ \ }}
\newcommand{\eg}{e.g.\xspace}
\newcommand{\ie}{i.e.\xspace}

\newcommand{\code}{\lstinline[language=Haskell]}

\newcommand{\figsize}{\small}

\newenvironment{stackAux}[2]{%
  \setlength{\arraycolsep}{0pt}
  \begin{array}[#1]{#2}}{
  \end{array}\ignorespacesafterend}

\theoremstyle{remark}
\newtheorem*{remark}{Remark}

\newcommand{\Lref}[1]{Lemma~\ref{#1}}
\newcommand{\Cref}[1]{Corollary~\ref{#1}}

\newcommand{\Tref}[1]{Theorem~\ref{#1}}
\newcommand{\Pref}[1]{Proposition~\ref{#1}}
\newcommand{\Fref}[1]{Figure~\ref{#1}}
\newcommand{\Sref}[1]{\S\ref{#1}}
\newcommand{\Aref}[1]{Appendix~\ref{#1}}

\newcommand{\gcalc}{\texorpdfstring{$\lambda_{\vee}$\xspace}{λ∨}}

\newcommand{\kw}[1]{\ensuremath{\mathbf{#1}}}
\newcommand{\kwLet}{\kw{let}}
\newcommand{\symName}[1]{\ensuremath{\textsf{\textbf{#1}}}}

\newcommand{\tmTopC}{\ensuremath{\top}}
\newcommand{\tmBotC}{\ensuremath{\bot}}
\newcommand{\tmJoinSym}{\vee}
\newcommand{\tmJoin}[2]{\ensuremath{#1 \tmJoinSym #2}}
\newcommand{\tmJoinMany}[2]{\ensuremath{#1 \vee \cdots \vee #2}}
\newcommand{\tmJoinManySet}[2]{\bigvee_{#1} #2}
\newcommand{\tmBigJoin}[3]{\bigvee_{#1 \in #2} #3}
\newcommand{\tmPair}[2]{(#1, #2)}
\newcommand{\tmApp}[2]{#1\ #2}
\newcommand{\tmLexPair}[2]{\langle #1, #2 \rangle}
\newcommand{\tmMonBind}[3]{#1 \leftarrow #2 ; #3}

\newcommand{\tmEmptySet}{\{\}}
\newcommand{\tmMeta}{\ottnt{e}}
\newcommand{\tmMetaAlt}{\ottnt{t}}
\newcommand{\valMeta}{\ottnt{v}}
\newcommand{\resMeta}{\ottnt{r}}
\newcommand{\tmRecord}[1]{\ensuremath{\{#1\}}}
\newcommand{\tmFV}[2]{\ensuremath{\fieldName{#1} = #2}}
\newcommand{\tmPrj}[2]{\ensuremath{#1.\kw{#2}}}
\newcommand{\tmLetIn}[3]{\ensuremath{\kwLet\; #1 = #2\ \kw{in}\ #3}}
\newcommand{\tmId}[1]{\ensuremath{\mathit{#1}}}

\newcommand{\fieldName}[1]{\mathtt{#1}}
\newcommand{\litString}[1]{\textrm{``#1''}}

\newcommand{\botv}{{\bot_\mathrm{v}}}
\newcommand{\valBotV}{\botv}
\newcommand{\valLam}[2]{\lambda#1.\:#2}

\newcommand{\ctxHole}{[\cdot]}

\newcommand{\stepsym}{\mapsto}
\newcommand{\step}[2]{#1 \stepsym #2}

\newcommand{\returns}[2]{#1 \Downarrow #2}
\newcommand{\returnsany}[1]{#1\, {\Downarrow}}

\newcommand{\substmeta}{\gamma}

\newcommand{\ctxeq}[2]{#1 \approx_{\mathrm{ctx}} #2}
\newcommand{\ctxapx}[2]{#1 \preceq_{\mathrm{ctx}} #2}

\newcommand{\logapxsym}{\preceq_{\mathrm{log}}}
\newcommand{\logapx}[2]{#1 \logapxsym #2}

\newcommand{\subst}[3]{#1[#2 / #3]}

\newcommand{\extmap}[3]{#1[#2 \mapsto #3]}

\newcommand{\textand}{\mathbin{\text{and}}}

\newcommand{\textiff}{\mathbin{\text{iff}}}
\newcommand{\mathsetcomp}[2]{\{ #1 \mid #2 \}}

\newcommand{\mathforall}[2]{\forall #1 . #2}
\newcommand{\mathinv}[1]{#1^{-1}}
\newcommand{\mathexists}[2]{\exists #1 . #2}
\newcommand{\mathin}[2]{#1 \in #2}
\newcommand{\mathemdash}{\text{\textemdash}}

\newcommand{\mathbind}[3]{#1 \leftarrow #2 ; #3}
\newcommand{\monlift}[1]{#1^*}

\newcommand{\mathnat}{\mathbb{N}}
\newcommand{\comp}[2]{#1 \circ #2}

\newcommand{\compact}[1]{K(#1)}
\newcommand{\ideals}[1]{\mathcal{I}(#1)}
\newcommand{\prinideal}[1]{\,{\downarrow}{#1}\!}
\newcommand{\downclosk}[1]{\,{\downarrow^{K}}{#1}\!}
\newcommand{\powdom}[2]{\mathcal{P}_{#1}(#2)}
\newcommand{\powdomh}{\powdom{H}}

\newcommand{\eclossym}{\mathbf{c}}
\newcommand{\eclos}[1]{#1_\eclossym}

\newcommand{\relname}[1]{\mathsf{#1}}
\newcommand{\varset}{\relname{{Var}}}
\newcommand{\symset}{\relname{{Sym}}}
\newcommand{\synval}{\relname{{Val}}}

\newcommand{\synvtype}{\relname{{VForm}}}
\newcommand{\synctype}{\relname{{CForm}}}

\newcommand{\synvtypepair}{\relname{{VForm}_\times}}
\newcommand{\synvtypeset}{\relname{{VForm}_{\{\}}}}
\newcommand{\synvtypefun}{\relname{{VForm}_{\to}}}

\newcommand{\synenv}{\relname{Env}}
\newcommand{\domval}{D}

\newcommand{\domcomp}{\eclos{\domval}}

\newcommand{\syncomp}{\relname{Exp}}
\newcommand{\synres}{\relname{Res}}
\newcommand{\synectx}{\relname{ECtx}}

\newcommand{\powerset}[1]{\mathbb{P}(#1)}

\newcommand{\denot}[2]{\llbracket #1 \rrbracket_{#2}}

\newcommand{\semapxsym}{\sqsubseteq}
\newcommand{\semapx}[2]{#1 \semapxsym #2}

\newcommand{\synjoin}[2]{#1 \sqcup #2}
\newcommand{\pairop}[2]{\eclos{(#1, #2)}}

\newcommand{\lrrel}[2]{\mathcal{#1}\llbracket#2\rrbracket}
\newcommand{\lrexp}[1]{\lrrel{E}{#1}}
\newcommand{\lrres}[1]{\lrrel{R}{#1}}
\newcommand{\lrval}[1]{\lrrel{V}{#1}}
\newcommand{\lrenv}[1]{\lrrel{G}{#1}}
\newcommand{\judglr}[3]{#1 \vDash #2 : #3}

\newcommand{\tyAc}{\phi}
\newcommand{\tyBc}{\psi}
\newcommand{\tyAv}{\tau}
\newcommand{\tyBv}{\sigma}

\newcommand{\envEmpty}{\cdot}
\newcommand{\tyFunOne}[2]{#1 \to #2}
\newcommand{\tyJoinSynSym}{\vee}
\newcommand{\tyJoinSyn}[2]{#1 \tyJoinSynSym #2}

\newcommand{\tyFun}[3]{\bigvee_{#1}{(\tyFunOne{#2}{#3})}}
\newcommand{\tyBigJoin}[2]{\bigsqcup_{#1}{#2}}
\newcommand{\tySng}[1]{\eclos{\{#1\}}}
\newcommand{\tyPairOp}[2]{\eclos{(#1, #2)}}
\newcommand{\tyJoinSym}{\sqcup}
\newcommand{\tyJoin}[2]{#1 \tyJoinSym #2}
\newcommand{\envDom}[1]{\mathrm{dom}(#1)}

\newcommand{\envJoin}[2]{#1 \tyJoinSym #2}

\newcommand{\tyapxsym}{\semapxsym}
\newcommand{\tyapx}[2]{\semapx{#1}{#2}}
\newcommand{\envapx}[2]{\tyapx{#1}{#2}}

\newcommand{\judg}[3]{#1 \vdash #2 : #3}

\newcommand{\contfun}[2]{#1 \to_{\textit{cont}} #2}
\newcommand{\apxmap}[2]{#1 \to_{\textit{apx}} #2}

\newcommand{\R}{\mathrel{R}}
\newcommand{\disosym}{\delta}

\newcommand{\rulename}[1]{\textsc{#1}}

\newcommand{\fbasis}[1]{#1}

\usepackage{wrapfig}
\usepackage{epigraph}
\usepackage{booktabs}   
\usepackage{subcaption} 

\usepackage{xcolor}
\usepackage{xspace}
\usepackage{graphicx}
\usepackage{mathpartir}
\usepackage{cmll}
\usepackage{tikz-cd}

\usepackage{listproc}
\usepackage[supertabular]{ottalt}

\newcommand{\ottdrule}[4][]{{\displaystyle\frac{\begin{array}{l}#2\end{array}}{#3}\quad\ottdrulename{#4}}}

\newcommand{\ottpremise}[1]{ #1 \\}
\newenvironment{ottdefnblock}[3][]{ \framebox{\mbox{#2}} \quad #3 \\[0pt]}{}

\newcommand{\ottnt}[1]{\mathit{#1}}
\newcommand{\ottmv}[1]{\mathit{#1}}
\newcommand{\ottkw}[1]{\mathbf{#1}}
\newcommand{\ottsym}[1]{#1}

\newcommand{\ottdrulename}[1]{\textsc{#1}}

\newcommand{\ottafterlastrule}{\\}

\newcommand{\ottdruleTApxSet}[1]{\ottdrule[#1]{%
\ottpremise{ \mathforall{  \mathin{ \mathit{i} }{ \mathit{I} }  }{  \mathexists{  \mathin{ \mathit{j} }{ \mathit{J} }  }{  \tyapx{ \tyAv_{\ottmv{i}} }{ \tyAv'_{\ottmv{j}} }  }  } }%
}{
 \tyapx{ \ottsym{\{}  \tyAv_{\ottmv{i}}  \ottsym{\mbox{$\mid$}}   \mathin{ \mathit{i} }{ \mathit{I} }   \ottsym{\}} }{ \ottsym{\{}  \tyAv'_{\ottmv{j}}  \ottsym{\mbox{$\mid$}}   \mathin{ \mathit{j} }{ \mathit{J} }   \ottsym{\}} } }{%
{\ottdrulename{TApxSet}}{}%
}}

  \renewottcommands[ott]

\usepackage{multirow}
\usepackage{amsmath}
\usepackage{stmaryrd}
\usepackage{supertabular}
\usepackage{ifthen}
\usepackage{alltt}
\usepackage{hyperref}
\usepackage{listings}
\usepackage{csquotes}
\usepackage{mathtools}
\usepackage[normalem]{ulem}

\begin{document}
\def\MathparLineskip{\lineskiplimit=0.8em\lineskip=0.8em plus 0.1em}

\shortonly{
\titlenote{Additional proofs are provided in this article's accompanying technical report. \todo{cite}}
}
\title{Functional Meaning for Parallel Streaming}
\longonly{\subtitle{Technical Report}}

\author{Nick Rioux}
\email{nrioux@cis.upenn.edu}
\orcid{0000-0001-5277-8920}
\author{Steve Zdancewic}
\email{stevez@cis.upenn.edu}
\orcid{0000-0002-3516-1512}
\affiliation{%
  \institution{University of Pennsylvania}
  \city{Philadelphia}
  \state{Pennsylvania}
  \country{USA}
}


\begin{abstract}
  Nondeterminism introduced by race conditions and message reorderings makes parallel and
distributed programming hard.
Nevertheless, promising approaches such as LVars and CRDTs address
this problem by introducing a partial order structure on shared state that
describes how the state evolves over time.
\emph{Monotone} programs that respect the order are deterministic.
Datalog-inspired languages incorporate this idea of monotonicity in a
first-class way but they are not general-purpose.
We would like parallel and distributed languages to be as natural to use as any
functional language, without sacrificing expressivity, and with a formal basis of
study as appealing as the lambda calculus.

This paper presents \gcalc{}, a core language for
deterministic parallelism that embodies the ideas above. In \gcalc{}, values
may increase over time according to a \emph{streaming order} and all computations
are monotone with respect to that order.
The streaming order coincides with the approximation order found in Scott semantics
and so unifies the foundations of functional programming
with the foundations of deterministic distributed computation.
The resulting lambda calculus has a computationally adequate model rooted in domain theory.
It integrates the compositionality and power of abstraction characteristic of
functional programming with the declarative nature of Datalog.

\longonly{This version of the paper includes extended exposition and appendices with proofs.}

\end{abstract}


\maketitle

\todo{Remember to disable TODOs and comments.}
\nr{Also remember to turn off comments.}

\section{Introduction}
It is no secret that parallel and distributed programming are hard.
To empower different threads of computation to work together, mechanisms such as
shared state and message passing are often adopted~\cite{lynch1996,hewitt73}.
Unfortunately, these common programming models are rife with nondeterminism thanks to races and message-reorderings.
This nondeterminism makes it harder to debug parallel programs and to replicate
them for fault tolerance~\cite{devietti12}.
Moreover, under such approaches, ensuring that a program computes at most one well-defined
answer is a manual task.
Programmers have to consider every possible interleaving of reads and writes or of
messages and add synchronization operations to guarantee that their programs are correct.

This state space explosion hinders efforts to verify the correctness of systems
through local equational reasoning.
On the other hand, compositional reasoning about software is a hallmark of
functional programming and one would hope that it would be possible even in the
parallel setting.
We believe that programmers will write more trustworthy code faster if they do
not need to account for all interleaving of reads and writes to shared
state or draw out Lamport diagrams~\cite{lamport78}.
Indeed, the ability to write parallel programs in a style in which determinism is the default
and nondeterminism, if present at all, is made explicit, has long been valued by
those studying parallel programming~\cite{lee06,bocchino09}.
To this end, we investigate a way for programmers to express parallel programs
that are, in a sense, \emph{deterministic by construction} in a functional
framework.

This idea is not new.  Over the years, there have been many approaches to dealing with nondeterminism.
In each, the structure of \textit{mathematical semilattices} plays a crucial role:
\begin{itemize}
  \item In functional programming, \emph{LVars}~\cite{kuper13} are an approach to
  deterministic-by-construction concurrency in which a memory cell
  contains a value that is an element of a semilattice.
  Writes to the cell combine the old and new values using the commutative least
  upper bound (or \emph{join}) operation, ensuring that writes are deterministic
  even if they race with each other.
  \item \emph{Conflict-free replicated data types}~(CRDTs)~\cite{shapiro11}
  replicate semilattice-structured data across nodes of a distributed system.
  Nodes share their state and use the join operation to merge their current state
  with other incoming states.
  Order-theoretic properties ensure that the system is guaranteed to be
  \emph{eventually consistent}.
  \item \emph{Datalog} programs specify inference rules (or \emph{Horn clauses})
  describing how to infer new facts from those that are already known.
  Application of these inference rules to sets of known facts may be performed in
  parallel;
  sets of learned facts may then be combined using the join operation of the
  powerset semilattice (\ie set union).
  This helps make Datalog appealing for highly parallel~\cite{gilray21}
  and distributed~\cite{loo09,alvaro11dedalus} computation.
\end{itemize}
Across all of these lines of work, shared data is endowed with semilattice
structure and, thus, forms a partial order.
The data may change over the course of a computation, but only so long as it
increases according to this order.

\paragraph{Streaming}

To capture the common essence behind all of these programming models and
integrate them into a functional setting, we introduce \gcalc, an untyped call-by-value
\emph{parallel streaming lambda calculus} in which semilattice structure is a
first-class language feature.
Every \gcalc value is an element of a partial order called the \emph{streaming order}.
Programs represent long-running computations that return (or ``stream'') a value
that may evolve over time according to this order.

This makes it convenient to program with possibly-infinite lists (often called
``streams''), but the use of partial orders means that in \gcalc,
\emph{streaming} is not just about streams.
Over time, trees may grow, sets may gain elements, and records may obtain new fields.
Because partial orders can be composed to form new ones, streaming data types
fit together just as well as functional programmers might expect: streams of streams and
even sets of higher-order functions are straightforward to work with in \gcalc.

The \gcalc function $\mathit{evens}$ below streams the infinite set
containing the even natural numbers.
{\small
\[
 \tmApp{ \mathit{evens} }{ \mathit{\_} }  =  \tmApp{  \tmJoin{ \ottsym{\{}  \ottsym{0}  \ottsym{\}} }{ \mathit{plus2all} }  }{ \ottsym{(}   \tmApp{ \mathit{evens} }{ \ottsym{()} }   \ottsym{)} }  \qquad\qquad
 \tmApp{ \mathit{plus2all} }{ \mathit{xs} }  =  \tmBigJoin{ \mathit{x} }{ \mathit{xs} }{ \ottsym{\{}  \mathit{x}  \ottsym{+}  \ottsym{2}  \ottsym{\}} } 
\]}

\noindent The join operator $\tmJoinSym$ in this definition runs its two arguments
simultaneously.
As this operator streams in two increasingly large sets, it streams out their set
union.
The function $\mathit{plus2all}$ produces a set containing $\mathit{x}  \ottsym{+}  \ottsym{2}$ for each
element $\mathit{x}$ of its input.
The expression $ \tmApp{ \mathit{evens} }{ \ottsym{()} } $ represents an iterative fixed point computation
that computes the least set containing $0$ and, for every element $\mathit{x}$ of
$ \tmApp{ \mathit{evens} }{ \ottsym{()} } $, the element $\mathit{x}  \ottsym{+}  \ottsym{2}$.
This exploits the parallel streaming nature of the join operator;
replacing it with a call to a function in a strict language like \gcalc
will result in a meaningless infinite loop.
Even in lazy Haskell, using the standard \code{Data.Set} type, this example diverges.

\paragraph{Monotonicity}
Not every function we can dream up behaves well in the streaming setting.
Suppose we could write a function $\mathit{f}$ in \gcalc according to the following specification.
{\small
\[
   \tmApp{ \mathit{f} }{ \ottsym{(}  \mathit{x}  \ottsym{)} }  = \begin{cases}
    \ottsym{\{}  \ottsym{1}  \ottsym{\}} & \text{if $x$ is a set containing the element 2 but \emph{not} 4} \\
      \tmEmptySet   & \text{otherwise}
  \end{cases}
  \]}\noindent The observers of a running \gcalc program might want to take some
action (such as sending a request to an external system) as soon as they observe
the element $\ottsym{1}$ in the output of $\mathit{f}$.
The table below shows the state of affairs we end up with when we stream the set
of even numbers, as defined above, to $\mathit{f}$.
{\small
\begin{center}
  \begin{tabular}{ r | c c c c c c}
    Time & $1$ & $2$ & $3$ & $\ldots$ \\
    \hline
    $ \tmApp{ \mathit{evens} }{ \ottsym{()} } $    & $\ottsym{\{}  \ottsym{0}  \ottsym{\}}$       & $\ottsym{\{}  \ottsym{0}  \ottsym{,}  \ottsym{2}  \ottsym{\}}$  & $\ottsym{\{}  \ottsym{0}  \ottsym{,}  \ottsym{2}  \ottsym{,}  \ottsym{4}  \ottsym{\}}$ & $\ldots$ \\
    $ \tmApp{ \mathit{f} }{ \ottsym{(}   \tmApp{ \mathit{evens} }{ \ottsym{()} }   \ottsym{)} } $ & $  \tmEmptySet  $   & $\ottsym{\{}  \ottsym{1}  \ottsym{\}}$ & $  \tmEmptySet  $ & $\ldots$ \\
    \hline\hline
    Action taken & none & request sent & ? & $\ldots$
  \end{tabular}
\end{center}
}
We have a problem: $\mathit{f}$ might \emph{retract} the element $\ottsym{1}$ from its
output, by which time the request may already be well on its way.
This is because $\mathit{f}$ is not \emph{monotone}.
Non-monotone functions break the \gcalc covenant that values
evolve over time according to the streaming order.
Consequently, an outside observer can never be sure it is safe to take an
action based on the output of such a function.

This problem can be seen as a form of nondeterminism. Modifying $\mathit{evens}$ to stream out the
element $\ottsym{4}$ before the element $\ottsym{2}$ should not change its meaning since
it still computes the same infinite set.
However, such a change does make a big difference in the example above:
if $\ottsym{4}$ is streamed in to $\mathit{f}$ before $\ottsym{2}$, the request is never
sent.

Compounding our predicament, in a distributed system, the \emph{Consistency as
Logical Monotonicity (CALM) Theorem}~\cite{hellerstein10,ameloot13,hellerstein20} suggests that breaking monotonicity would
require an implementation of \gcalc to use expensive coordination
mechanisms to preserve desirable consistency properties.
The design of \gcalc avoids these thorny issues.
By construction, as any function receives more input, it may only stream more
output.


\paragraph{Meaning \& Determinism}
In defining the meaning of \gcalc programs, we have two goals.
First, we want to capture how programs evolve over time.
Small-step reduction systems are useful for this.
However, a finite trace defined by such a system can only capture a finite amount
of output; such semantics cannot explicitly describe the infinite end
behavior of programs like $ \tmApp{ \mathit{evens} }{ \ottsym{()} } $.

Describing the behavior of such programs ``in the limit'' is our second goal.
Thus, we seek a denotational semantics for \gcalc capable of
describing infinite behaviors which also reflects that values in the language are
ordered and that all functions are monotone.
This style of semantics lets us define determinism as the property that every program has at most
one end behavior.


The reader might note that these desiderata are well-known
characteristics of Scott semantics~\cite{scott70}.
The Scott approach to the semantics of the lambda calculus describes
programs in terms of a class of partial order known as \emph{Scott domains}.
As opposed to the \gcalc streaming order, which (following \citet{kahn74})
describes how data evolves over
time, partial orders in Scott's models describe how ``well defined'' functions are.
Happily, these two notions are compatible: becoming ``more defined'' over time
is the streaming behavior for functions in \gcalc.


The point here is not only that domain theory, a classical tool in the study of
programming languages, may be useful in describing parallel and
distributed systems, but also that it naturally subsumes the mathematics
that designers of programming models for such systems are already using
to obtain desirable properties like determinism and eventual consistency.


\sz{Somewhere, maybe here? We should revisit the point about determinism--to tie
  it to some of our technical development as we discussed. (If at all possible.)}

The reader lacking familiarity with domain theory need not be dissuaded.
We prioritize operational explanations of the insights gleaned from
domain theory.
Furthermore, the meaning of \gcalc programs is described using a \emph{filter model}~\citep{barendregt83}.
This technique, long used in the literature on intersection types, is essentially a
way of giving a denotational semantics by defining a very fine-grained type system.
Thus, an understanding of lambda calculus, operational semantics, and type
inference rules is sufficient background to read most of this paper.

\paragraph{Contributions}
Our primary contribution is the design of the core language \gcalc.
A cousin of Datafun~\cite{arntzenius16}, it fuses the expressive functional
programming of the untyped lambda calculus with Datalog-style logic programming.
Accompanying this design are several technical contributions.
\begin{itemize}
  \item We demonstrate in \Sref{sec:examples} examples of the programming
  patterns enabled by the parallel streaming design of \gcalc not directly
  expressible in other languages.
  \item We describe the meaning of the language with a small-step reduction
  system~(\Sref{sec:opsem}) and a filter model~(\Sref{sec:logsem}).
  It is known that a Scott-style semantics (including a solution to a recursive
  domain equation) can be derived from a filter model in a straightforward way.
  \item We connect the two semantics via a computational adequacy result
  utilizing a novel logical relation in \Sref{sec:logsem:adequacy}.
\end{itemize}

Our focus is the parallel semantics for the present work; we leave
it to the future to formally address network nondeterminism and
fault tolerance.
Nonetheless, a key motivation of this work is to align the
use of partial orders and monotone functions found in
distributed computing with their use in programming language semantics \`a la
Scott.

\section{Language Design \& Main Ideas}
\label{sec:language}
The key ingredients of \gcalc are:
(1) data types endowed with a partial order---the \textit{streaming order}---that induces a semilattice structure on computations,
(2) primitive operations that respect the streaming order on data types (including pattern matching based on \textit{threshold queries}),
and (3) a general \textit{parallel join} operation.
This section introduces these ingredients and demonstrates their use via some
motivating examples.

\subsection{Syntax}
\label{sec:syntax}
\begin{figure}[t]
  \figsize
  \[
  \begin{array}{lrcrcl}
    \textit{expressions} &
           \syncomp & \ni &
           \tmMeta\longonly{, \tmMetaAlt} & ::= &  \tmBotC 
           \sep  \tmTopC 
           \sep  \valBotV 
           \sep \mathit{x}
           \sep  \valLam{ \mathit{x} }{ \tmMeta } 
           \sep  \tmPair{ \tmMeta_{{\mathrm{1}}} }{ \tmMeta_{{\mathrm{2}}} } 
           \sep \ottnt{s}
           \sep \{ \tmMeta_{{\mathrm{1}}}, \ldots, \tmMeta_{\ottmv{n}} \}
           \sep  \tmApp{ \tmMeta_{{\mathrm{1}}} }{ \tmMeta_{{\mathrm{2}}} } 
           \\
&&&&&       \sep \ottkw{let} \, \ottsym{(}  \mathit{x_{{\mathrm{1}}}}  \ottsym{,}  \mathit{x_{{\mathrm{2}}}}  \ottsym{)}  \ottsym{=}  \tmMeta \, \ottkw{in} \, \tmMeta'
           \sep \ottkw{let} \, \ottnt{s}  \ottsym{=}  \tmMeta \, \ottkw{in} \, \tmMeta'
           \sep  \tmBigJoin{ \mathit{x} }{ \tmMeta_{{\mathrm{1}}} }{ \tmMeta_{{\mathrm{2}}} } 
           \sep  \tmJoin{ \tmMeta_{{\mathrm{1}}} }{ \tmMeta_{{\mathrm{2}}} } 
           \\
           \textit{results} &
           \synres & \ni & \resMeta & ::= &  \tmBotC 
           \sep  \tmTopC 
           \sep \valMeta
                  \\
           \textit{values} &
           \synval & \ni & \valMeta & ::= & \mathit{x}
                  \sep  \valBotV 
                  \sep  \valLam{ \mathit{x} }{ \tmMeta } 
                  \sep  \tmPair{ \valMeta_{{\mathrm{1}}} }{ \valMeta_{{\mathrm{2}}} } 
                  \sep \ottnt{s}
                  \sep \{ \valMeta_{{\mathrm{1}}}, \ldots, \valMeta_{\ottmv{n}} \}
                  \\ \\
           \textit{eval.\ contexts} & \synectx & \ni &
                         \ottnt{E} & ::= &  \ctxHole 
                         \sep  \tmPair{ \ottnt{E} }{ \tmMeta } 
                         \sep  \tmPair{ \valMeta }{ \ottnt{E} } 
                         \sep \{\tmMeta_{{\mathrm{1}}}, \ldots, \tmMeta_{\ottmv{n}}, \ottnt{E}, \tmMeta'_{{\mathrm{1}}}, \ldots, \tmMeta'_{\ottmv{m}}\}
                         \sep  \tmApp{ \ottnt{E} }{ \tmMeta } 
                         \sep  \tmApp{ \valMeta }{ \ottnt{E} }  \\
&&&&&                      \sep \ottkw{let} \, \ottsym{(}  \mathit{x_{{\mathrm{1}}}}  \ottsym{,}  \mathit{x_{{\mathrm{2}}}}  \ottsym{)}  \ottsym{=}  \ottnt{E} \, \ottkw{in} \, \tmMeta
                         \sep \ottkw{let} \, \ottnt{s}  \ottsym{=}  \ottnt{E} \, \ottkw{in} \, \tmMeta
                         \sep  \tmBigJoin{ \mathit{x} }{ \ottnt{E} }{ \tmMeta } 
                         \sep  \tmJoin{ \ottnt{E} }{ \tmMeta } 
                         \sep  \tmJoin{ \tmMeta }{ \ottnt{E} } 
  \end{array}
  \]
  \caption{\gcalc Syntax}
  \label{fig:syntax}
\end{figure}

The syntax of \gcalc is given in \Fref{fig:syntax}.  All forms are finitary and
we consider them up to $\alpha$-equivalence.
The expression $ \tmBotC $ represents a ``meaningless'' computation that does not
produce any output.
During evaluation, it propagates throughout a program in a manner similar to an
error or diverging term.
Dually, $ \tmTopC $ is an error that represents an inconsistent result.
\nr{What does inconsistent mean? Should allude to this earlier}
In contrast to $ \tmBotC $, the value $ \valBotV $ represents the knowledge that a
computation has successfully produced \emph{something}---but nothing more about
what the result may be.
It can be passed around as any other value, but it will produce $ \tmBotC $ if inspected
in any way.
Functions are introduced with $\lambda$-abstractions and eliminated with
application.  Pairs are expressions $ \tmPair{ \tmMeta_{{\mathrm{1}}} }{ \tmMeta_{{\mathrm{2}}} } $ that can be destructured
with a single-case pattern-matching $\kw{let}$ expression.

The language is parameterized by a set $\varset$ of \emph{variables}
ranged over by the metavariable $\mathit{x}$ and a set $\symset$
of \emph{symbols} ranged over by the metavariable $\ottnt{s}$.
Variables are the usual notion from the study of the lambda calculus.
By convention, the variable $\mathit{\_}$ in a binding position indicates the binding is unused.
Symbols are base values (constants) that may have some order structure.
We assume a partial\longonly{\footnote{The partiality of this operation means that symbols
do not, in general, form a semi-lattice.}} computable
operation over symbols $ \ottnt{s_{{\mathrm{1}}}}  \sqcup  \ottnt{s_{{\mathrm{2}}}} $ which is associative, commutative, and
idempotent.
We also assume that equality of symbols is decidable.
The streaming order on symbols is defined as $\ottnt{s_{{\mathrm{1}}}}  \leq  \ottnt{s_{{\mathrm{2}}}}$ iff $ \ottnt{s_{{\mathrm{1}}}}  \sqcup  \ottnt{s_{{\mathrm{2}}}}   \ottsym{=}  \ottnt{s_{{\mathrm{2}}}}$.

The elimination form $\ottkw{let} \, \ottnt{s}  \ottsym{=}  \tmMeta_{{\mathrm{1}}} \, \ottkw{in} \, \tmMeta_{{\mathrm{2}}}$ is an example of a \emph{threshold
  query}~\cite{kuper13}.  It produces no output until the evaluation of
$\tmMeta_{{\mathrm{1}}}$ produces a symbol greater than or equal to the threshold $\ottnt{s}$.
If and when this happens, it produces the output of $\tmMeta_{{\mathrm{2}}}$.  We often assume
the existence of certain symbols; these are referred to as names (as
in $ \symName{true} $ and $ \symName{false} $), string literals, and the unit value $\ottsym{()}$.
Except where explicitly
stated otherwise, joins of distinct symbols (\eg
$  \symName{true}   \sqcup   \symName{false}  $) are assumed to be undefined.
It follows that such symbols are incomparable.

The language \gcalc includes a \textit{set} data type, constructed via
$\{ \tmMeta_{{\mathrm{1}}}, \ldots, \tmMeta_{\ottmv{n}} \}$.
Following Datafun~\cite{arntzenius16}, sets are eliminated with the ``big join'' form,
which maps an operation over the elements of a set and joins together all of the
results.
Any value may be an element of a set and equality of set elements is
\textit{not} required to be decidable.
As with all features in this language, the elimination form for sets is in a sense monotone.
Consequently, it is impossible to calculate the difference between two sets or
test for the absence of a particular value in a set.
This design ensures that actions taken based on the elements currently
in the set will remain valid in the future.
These caveats are familiar to users of LVars and, to a lesser extent, CRDTs; see \Sref{sec:disc:ext}.

The binary join operator is written $ \tmJoin{ \tmMeta_{{\mathrm{1}}} }{ \tmMeta_{{\mathrm{2}}} } $.
It is a parallel composition operator whose behavior is
overloaded depending on the type of data being joined.
By convention, the angled join symbol $\vee$ represents \emph{syntax} while the
square join symbol $\sqcup$ represents a \emph{metafunction}.

\Fref{fig:syntax} also defines \emph{evaluation contexts} which will be used by the semantics in \Sref{sec:op-sem:red}.
This definition ensures that evaluation proceeds sequentially left-to-right over
most forms in the language with the exception of set introduction and binary
join, whose subterms are evaluated in parallel.
The presence of these forms make it possible to decompose an expression into an evaluation
context and a redex in multiple different ways.

\subsection{Encodings}
We make use of a few derived syntactic forms. For example,
$\ottkw{let} \, \mathit{x}  \ottsym{=}  \tmMeta \, \ottkw{in} \, \tmMeta'$ can be encoded using abstraction and application in the
usual way as $ \tmApp{ \ottsym{(}   \valLam{ \mathit{x} }{ \tmMeta' }   \ottsym{)} }{ \tmMeta } $.
To avoid explicitly nested pattern matching constructs, we use compound patterns
like $\ottkw{let} \, \ottsym{(}  \ottnt{s}  \ottsym{,}  \mathit{x_{{\mathrm{2}}}}  \ottsym{)}  \ottsym{=}  \tmMeta \, \ottkw{in} \, \tmMeta'$ in place of the more verbose
$\ottkw{let} \, \ottsym{(}  \mathit{x_{{\mathrm{1}}}}  \ottsym{,}  \mathit{x_{{\mathrm{2}}}}  \ottsym{)}  \ottsym{=}  \tmMeta \, \ottkw{in} \, \ottkw{let} \, \ottnt{s}  \ottsym{=}  \mathit{x_{{\mathrm{1}}}} \, \ottkw{in} \, \tmMeta'$.
One may observe that patterns represent some minimum threshold that a
scrutinized value must reach in order to trigger some computation.
In other words, pattern matching is a form of threshold query.

The familiar expression $\ottkw{if} \, \tmMeta_{{\mathrm{1}}} \, \ottkw{then} \, \tmMeta_{{\mathrm{2}}} \, \ottkw{else} \, \tmMeta_{{\mathrm{3}}}$ is encoded in \gcalc as
\[  \tmJoin{ \ottkw{let} \, \mathit{x}  \ottsym{=}  \tmMeta_{{\mathrm{1}}} \, \ottkw{in} \, \ottsym{(}  \ottkw{let} \,  \symName{true}   \ottsym{=}  \mathit{x} \, \ottkw{in} \, \tmMeta_{{\mathrm{2}}}  \ottsym{)} }{ \ottsym{(}  \ottkw{let} \,  \symName{false}   \ottsym{=}  \mathit{x} \, \ottkw{in} \, \tmMeta_{{\mathrm{3}}}  \ottsym{)} }  \]
The idea here is to run two threads in parallel, one for each boolean value.
When the value of $\tmMeta_{{\mathrm{1}}}$ is observed to be $ \symName{true} $, the thread
containing the \kw{then} branch of the \kw{if} expression will execute.
On the other hand, the thread for the \kw{else} branch will always be observed as
$ \tmBotC $ since $\mathit{x}$ never meets its threshold of $ \symName{false} $.
This expression behaves as expected because, as previously noted, the
symbols $ \symName{true} $ and $ \symName{false} $ are incomparable with each other.
If we were to instead dictate (as Datafun does) that
$ \symName{true} $ is \emph{greater} than $ \symName{false} $, then preserving monotonicity
would require that the \kw{else} branch runs even when the condition evaluates to
$ \symName{true} $.

We can generalize this idea to support various forms of pattern matching.
Data constructors are represented as a pair of a symbol (a tag indicating which
data constructor is being applied) and an argument.
For example, we might represent the empty list $ \ottsym{[]} $ as the value $ \tmPair{  \symName{nil}  }{  \valBotV  } $
and a non-empty list $\valMeta_{{\mathrm{1}}}  \ottsym{::}  \valMeta_{{\mathrm{2}}}$ as $ \tmPair{  \symName{cons}  }{  \tmPair{ \valMeta_{{\mathrm{1}}} }{ \valMeta_{{\mathrm{2}}} }  } $ where $\valMeta_{{\mathrm{1}}}$ is the first
element of the list and $\valMeta_{{\mathrm{2}}}$ is the tail of the list.
We can then encode a pattern match
$\ottkw{case} \, \tmMeta_{{\mathrm{1}}} \, \ottkw{of} \, \ottsym{[]}  \to  \tmMeta_{{\mathrm{2}}}  \ottsym{\mbox{$\mid$}}  \mathit{y}  \ottsym{::}  \mathit{ys}  \to  \tmMeta_{{\mathrm{3}}}$ as:
\[  \tmJoin{ \ottkw{let} \, \mathit{x}  \ottsym{=}  \tmMeta_{{\mathrm{1}}} \, \ottkw{in} \, \ottsym{(}  \ottkw{let} \, \ottsym{(}   \symName{nil}   \ottsym{,}  \mathit{\_}  \ottsym{)}  \ottsym{=}  \mathit{x} \, \ottkw{in} \, \tmMeta_{{\mathrm{2}}}  \ottsym{)} }{ \ottsym{(}  \ottkw{let} \, \ottsym{(}   \symName{cons}   \ottsym{,}  \ottsym{(}  \mathit{y}  \ottsym{,}  \mathit{ys}  \ottsym{)}  \ottsym{)}  \ottsym{=}  \mathit{x} \, \ottkw{in} \, \tmMeta_{{\mathrm{3}}}  \ottsym{)} }  \]
We assume that natural numbers are encoded as
algebraic data types in a manner similar to lists.
A consequence is that the streaming order on these numbers is the discrete order
(\ie $\ottsym{1}$ is incomparable with $\ottsym{2}$), \emph{not} the standard order (in which $\ottsym{1}$ would be less than $\ottsym{2}$).
As with booleans, this choice is made because it reflects the behavior of the numeric data types
programmers are used to in Haskell and ML.
We will also assume common boolean, arithmetic, and comparison operations have been implemented using these encodings.

The parallel nature of the join operator has some
consequences of note for pattern matching.
First, commutativity of joins ensures that pattern matching is symmetric;
the order of the cases does not matter.
When there are multiple applicable branches, all of them run and are combined
with join.
Second, this parallelism is an increase in expressivity over sequential
languages: it allows one to write the parallel-or function
(see \Sref{sec:por}).


\longonly{
\begin{remark}
  As joins operate pointwise on functions, it is possible to define
  functions handling different cases of a data type and compose them together post hoc.
  This is essentially a means of encoding the overloading of functions.
\[
   \tmJoin{ \ottsym{(}   \valLam{ \mathit{x} }{ \ottkw{case} \, \mathit{x} \, \ottkw{of} \, \ottsym{[]}  \to  \tmMeta_{{\mathrm{1}}} }   \ottsym{)} }{ \ottsym{(}   \valLam{ \mathit{x} }{ \ottkw{case} \, \mathit{x} \, \ottkw{of} \, \mathit{y}  \ottsym{::}  \mathit{ys}  \to  \tmMeta_{{\mathrm{2}}} }   \ottsym{)} }   \ottsym{=}   \valLam{ \mathit{x} }{ \ottkw{case} \, \tmMeta \, \ottkw{of} \, \ottsym{[]}  \to  \tmMeta_{{\mathrm{1}}}  \ottsym{\mbox{$\mid$}}  \mathit{y}  \ottsym{::}  \mathit{ys}  \to  \tmMeta_{{\mathrm{2}}} } 
\]
  On one hand, this demonstrates the ability to stream higher-order data in \gcalc.
  A streamed function may gain the ability to handle more and more cases
  of a data type over time.
  On the other hand, even beyond its uses for streaming, this example shows that the
  join operator empowers the programmer to code in an especially modular style.
  This view of join is somewhat inspired by the \emph{merge operator} of
  \citet{dunfield14}. \citet{rioux23} study a related approach to overloading in a
  typed setting.
\end{remark}
}

A record \tmRecord{\tmFV{fld_1}{\valMeta_{{\mathrm{1}}}}, \tmFV{fld_2}{\valMeta_{{\mathrm{2}}}}} can be expressed as a
function from field identifiers, encoded as symbols, to values:
\(\lambda x.\; (\kw{let}\; \symName{fld}_1 = x\ \kw{in}\ \valMeta_{{\mathrm{1}}}) \vee
(\kw{let}\; \symName{fld}_2 = x\ \kw{in}\ \valMeta_{{\mathrm{2}}}) \).
Record projection $\tmPrj{e}{\fieldName{fld}}$ is
then just function application $(e\ \symName{fld})$.
In examples, it is convenient to introduce record field pattern matching, which
puns field identifiers with a variable of the same name and desugars to projection.
Under this encoding, the join of two records acts pointwise.


The call-by-value fixed point combinator
$Z =  \tmApp{  \valLam{ \mathit{f} }{ \ottsym{(}   \tmApp{  \valLam{ \mathit{x} }{ \mathit{f} }  }{ \ottsym{(}   \tmApp{  \tmApp{  \valLam{ \mathit{y} }{ \mathit{x} }  }{ \mathit{x} }  }{ \mathit{y} }   \ottsym{)} }   \ottsym{)} }  }{ \ottsym{(}   \tmApp{  \valLam{ \mathit{x} }{ \mathit{f} }  }{ \ottsym{(}   \tmApp{  \tmApp{  \valLam{ \mathit{y} }{ \mathit{x} }  }{ \mathit{x} }  }{ \mathit{y} }   \ottsym{)} }   \ottsym{)} } $ computes the
least fixed point of a \gcalc function.
In examples, we use recursive function notation and infix operators.

\subsection{Examples}
\label{sec:examples}
Let us now explore some examples of how the language's features can be used.

\subsubsection*{Streams}
\begin{figure}
  \[
    {
      \small
      \begin{array}{ll|l}
        & \text{Program} & \text{Observation} \\
        \hline
        &  \tmApp{ \mathit{fromN} }{ \ottsym{0} }  &  \tmBotC  \\
             \stepsym^{*}  &  \tmJoin{ \ottsym{(}   \tmApp{ \ottsym{0}  \ottsym{::}  \mathit{fromN} }{ \ottsym{1} }   \ottsym{)} }{  \valBotV  }  &  \valBotV  \\
             \stepsym^{*}  &  \tmJoin{ \ottsym{(}  \ottsym{0}  \ottsym{::}  \ottsym{(}   \tmJoin{ \ottsym{(}   \tmApp{ \ottsym{1}  \ottsym{::}  \mathit{fromN} }{ \ottsym{2} }   \ottsym{)} }{  \valBotV  }   \ottsym{)}  \ottsym{)} }{  \valBotV  }  & \ottsym{0}  \ottsym{::}   \valBotV  \\
             \stepsym^{*}  &  \tmJoin{ \ottsym{(}  \ottsym{0}  \ottsym{::}  \ottsym{(}   \tmJoin{ \ottsym{(}  \ottsym{1}  \ottsym{::}  \ottsym{(}   \tmJoin{ \ottsym{(}   \tmApp{ \ottsym{2}  \ottsym{::}  \mathit{fromN} }{ \ottsym{3} }   \ottsym{)} }{  \valBotV  }   \ottsym{)}  \ottsym{)} }{  \valBotV  }   \ottsym{)}  \ottsym{)} }{  \valBotV  }  &
            \ottsym{0}  \ottsym{::}  \ottsym{1}  \ottsym{::}   \valBotV  \\
            & \vdots & \vdots \\
      \end{array}
    }
    \]
    \caption{Behavior of the term $ \tmApp{ \mathit{fromN} }{ \ottsym{0} } $.}
    \label{fig:fromN}
\end{figure}
The function below computes the stream of natural numbers starting from $n$:
\[  \tmApp{ \mathit{fromN} }{ n }  =  \tmJoin{ \ottsym{(}   \tmApp{ n  \ottsym{::}  \mathit{fromN} }{ \ottsym{(}  n  \ottsym{+}  \ottsym{1}  \ottsym{)} }   \ottsym{)} }{  \valBotV  }  \]


\noindent
\Fref{fig:fromN} illustrates the runtime behavior of the program $ \tmApp{ \mathit{fromN} }{ \ottsym{0} } $.
The evaluation steps down the left-hand column correspond to an
``unrolling'' of the recursive definition as well as steps that simplify
arithmetic operations.
The sequence of observations down the right-hand column gives the finitary
partial results that can be generated by $ \tmApp{ \mathit{fromN} }{ \ottsym{0} } $.
Informally, an observation is the information that the computation has streamed
out so far.
As the system of reduction in \Sref{sec:opsem} will make clear,
observations are obtained by regarding running computations (\eg recursive calls
that have not yet been unrolled) as $ \tmBotC $ and simplifying the resulting
expression.
The intent of the streaming order is that from top to bottom, these values only
\textit{increase}.
\todo{We need to explain the streaming order here or earlier.}
From the sequence of observations, it is possible to see why the definition of $\mathit{fromN}$ includes a join with $ \valBotV $.
Since $ \tmBotC $ propagates like an error, programs like $\ottsym{0}  \ottsym{::}   \tmBotC $ are
equivalent to $ \tmBotC $.
Omitting the join would in essence replace each $ \valBotV $ with $ \tmBotC $ in our observations, so no nontrivial result would ever be produced by $\mathit{fromN}$ without it.
The denotational semantics in \Sref{sec:logsem} will give us the
meaning of a \gcalc{} program ``in the limit''.
For $ \tmApp{ \mathit{fromN} }{ \ottsym{0} } $, that would be the infinite stream of natural numbers.

\subsubsection*{Parallel Or}
\label{sec:por}
As in the work of \citet{boudol94} and \citet{dezani94}, the classic parallel-or
operator can be encoded.
The call-by-value parallel-or function is:
{\small
  \[ \tmApp{\tmApp{\tmId{por}}{\mathit{x}}}{\mathit{y}}
    \begin{array}[t]{cl}
    =&(\tmLetIn{ \symName{true} }{ \tmApp{ \mathit{x} }{ \ottsym{()} } }{ \symName{true} }) \ \tmJoinSym \ 
    (\tmLetIn{ \symName{true} }{ \tmApp{ \mathit{y} }{ \ottsym{()} } }{ \symName{true} }) \ \tmJoinSym
    \\ &
    (\tmLetIn{ \symName{false} }{ \tmApp{ \mathit{x} }{ \ottsym{()} } }{\tmLetIn{ \symName{false} }{ \tmApp{ \mathit{y} }{ \ottsym{()} } }{ \symName{false} }})
  \end{array}
\]
}

\noindent
The function \tmId{por} takes two thunks \tmId{x} and \tmId{y} as input, which,
if they run to completion, are expected to return booleans.
If forcing either \tmId{x} or \tmId{y} produces $ \symName{true} $, then so does
\tmId{por\; x\; y}, even if the other thunk loops and never produces an output.
When both thunks return $ \symName{false} $, so does \tmId{por}.



\subsubsection*{Datalog-style sets}
\label{sec:datalog}

In contrast to the past theoretical studies of join, which focus on its role as
a composition of parallel \emph{computations}, \gcalc emphasizes the
importance of joins of \emph{values}, including functions and sets.
For example, the expression $  \tyJoinSyn{ \ottsym{\{}  \ottsym{(}  \ottsym{1}  \ottsym{,}  \ottsym{2}  \ottsym{)}  \ottsym{\}} }{ \ottsym{\{}  \ottsym{(}  \ottsym{2}  \ottsym{,}  \ottsym{3}  \ottsym{)}  \ottsym{\}} }  $ computes the set
$\ottsym{\{}   \tmPair{ \ottsym{1} }{ \ottsym{2} }   \ottsym{,}   \tmPair{ \ottsym{2} }{ \ottsym{3} }   \ottsym{\}}$, which contains two tuples of integers.
Sets of tuples encode relations; together with the join operator, we can express
Datalog-style programs like the program below.\footnote{This lesson can
  already be seen in Datafun. In \gcalc, however, we do not need any special constructs for
  recursion since fixed point combinators are definable in the language.} %
\[  \tmApp{ \mathit{reaches} }{ \mathit{x} }  =  \tmApp{  \tmJoin{ \ottsym{\{}  \mathit{x}  \ottsym{\}} }{  \tmBigJoin{ n }{  \tmApp{ \mathit{neighbors} }{ \mathit{x} }  }{ \mathit{reaches} }  }  }{ n }  \]

\noindent
The function $\mathit{reaches}$ takes the name of a node in a graph and returns the set of node names that are reachable in any number of steps.
It uses a function $\mathit{neighbors}$, which encodes a graph by mapping the name of a
node to the set of the names of the nodes that can be reached in one step.
We see that, thanks to the presence of join, recursive
programs can be defined as fixed points of operators that could not otherwise
be expressed.  Under a call-by-value interpretation, replacing the join operator
in this code with a conventional function call would result in a meaningless
infinite loop whenever the graph encoded via $\mathit{neighbors}$ has a cycle. In
\gcalc, as in Datalog, a non-trivial fixed point exists.

\subsubsection*{Concurrent systems}
As a final example, we see how the features of \gcalc enable the
construction of systems of independent concurrent processes. The code in
Figure~\ref{fig:concurrent} implements a two-phase commit protocol with three
nodes: two peers and a coordinator. The coordinator proposes a value and
facilitates agreement between the peers.  Each of these three
parties is represented as a top-level function taking the current state of the
whole system (called the \emph{global state})
as input and producing the node's new \emph{local state} as output.
The global state at any given point in time can be thought of as the join of the
local states of all the nodes in the system.

\begin{figure}[t]
{\small
  \[
\begin{array}{l@{\qquad\qquad}l}
\tmId{peer_1}\ \tmRecord{\tmId{proposal}} =               &  \tmId{coordinator}\ state =  \\
\quad  \tmRecord{\tmFV{ok_1}{proposal > 4}}   &  \quad \tmRecord{ \tmFV{proposal}{5} } \ \vee \    \\
                                               &  \quad (\tmLetIn{\tmRecord{ \tmId{ok_1}, \tmId{ok_2} }}{state}{} \\
\tmId{peer_2}\ \tmRecord{ \tmId{proposal} } =             &  \quad \quad \tmRecord{ \tmFV{res}{displayResult\; (\tmId{ok_1}\; \mathtt{\&\&}\; \tmId{ok_2}) }}) \\
\quad  \tmRecord{\tmFV{ok_2}{\tmId{proposal} <= 6} } & \\
                              & \tmApp{\tmId{displayResult}}{\mathit{result}} = \ottkw{if} \, \mathit{result} \, \ottkw{then} \,  \litString{accepted}  \, \ottkw{else} \,  \litString{rejected}  \\
\\
\multicolumn{2}{l}{\tmId{system}\ () =
  \tmRecord{} \ \vee \
  \tmId{peer_1}\; ( \tmApp{ \mathit{system} }{ \ottsym{()} } ) \ \vee \
  \tmId{peer_2}\; ( \tmApp{ \mathit{system} }{ \ottsym{()} } ) \ \vee \
  \tmId{coordinator}\; ( \tmApp{ \mathit{system} }{ \ottsym{()} } )}
\end{array}
\]
}

\caption{Implementation of two-phase commit.}
\label{fig:concurrent}
\end{figure}
\begin{figure}
    \centering \scriptsize
\begin{tabular}{ c c c c }
  \tmId{peer_1} & \tmId{peer_2} & \tmId{coordinator} & \tmId{system} \\
  \hline
  $ \tmBotC $ & $ \tmBotC $ & $ \tmBotC $ & $ \tmBotC $ \\
  $ \tmBotC $ & $ \tmBotC $ & $ \tmBotC $ & $\tmEmptySet{}$ \\
  $ \tmBotC $ & $ \tmBotC $ & $\ottsym{\{}   \fieldName{proposal}   \ottsym{=}  \ottsym{5}  \ottsym{\}}$ & $\ottsym{\{}   \fieldName{proposal}   \ottsym{=}  \ottsym{5}  \ottsym{\}}$ \\
  $\ottsym{\{}   \fieldName{ok_1}   \ottsym{=}   \symName{true}   \ottsym{\}}$ & $\ottsym{\{}   \fieldName{ok_2}   \ottsym{=}   \symName{true}   \ottsym{\}}$ & $\ottsym{\{}   \fieldName{proposal}   \ottsym{=}  \ottsym{5}  \ottsym{\}}$
  & $\ottsym{\{}   \fieldName{ok_1}   \ottsym{=}   \symName{true}   \ottsym{,}   \fieldName{ok_2}   \ottsym{=}   \symName{true}   \ottsym{,}   \fieldName{proposal}   \ottsym{=}  \ottsym{5}  \ottsym{\}}$ \\
  $\ottsym{\{}   \fieldName{ok_1}   \ottsym{=}   \symName{true}   \ottsym{\}}$ & $\ottsym{\{}   \fieldName{ok_2}   \ottsym{=}   \symName{true}   \ottsym{\}}$ & $\ottsym{\{}   \fieldName{res}   \ottsym{=}   \litString{accepted}   \ottsym{,}   \fieldName{proposal}   \ottsym{=}  \ottsym{5}  \ottsym{\}}$
    & $\ottsym{\{}   \fieldName{res}   \ottsym{=}   \litString{accepted}   \ottsym{,}   \fieldName{ok_1}   \ottsym{=}   \symName{true}   \ottsym{,}   \fieldName{ok_2}   \ottsym{=}   \symName{true}   \ottsym{,}   \fieldName{proposal}   \ottsym{=}  \ottsym{5}  \ottsym{\}}$ \\
\end{tabular}
  \caption{Evolution of the two-phase commit protocol over time.}
  \label{fig:twopc}
\end{figure}

As shown in the figure, the system as a whole is defined as a recursive thunk.
It passes each process the previous global state (a recursive call)
and computes the next global state of the system by joining their results.
All states involved are records.

Illustrating the Datalog-style semantics, we see how the state of the system
evolves over time in Figure~\ref{fig:twopc}.
All states start at $ \tmBotC $ by fiat.
The system is defined so that its first non-trivial state is the empty record,
which kicks off computation.
At this point, $\tmId{peer_1}$ and $\tmId{peer_2}$ are not able to run because they
require as input a record containing a field $\fieldName{proposal}$.
Thus, their local states remain $ \tmBotC $.
The coordinator---at least, part of it---is able to run and proposes the value $\ottsym{5}$.
This allows the peers to execute.
Both agree with the proposed value of $\ottsym{5}$, setting their corresponding
record fields to $ \symName{true} $.
Once the peers agree, the coordinator produces a $\fieldName{res}$ field in its
local state indicating the proposed value was accepted.
At this time, the system has computed a fixed point.

\section{Approximate Operational Semantics}
\label{sec:opsem}
We now give an \emph{approximate operational semantics} to describe how \gcalc
programs may evolve over time.
This technique declaratively designates the valid partial runs of a program.

\subsection{Reduction Rules}
\label{sec:op-sem:red}
\begin{figure}
  {
  \renewcommand{\ottdrulename}[1]{}

  {\small
  \begin{drulepar}{$ \step{ \tmMeta }{ \tmMeta' } $}{reduction}
    \drule{stepECtx} \and

     \step{ \ottnt{E}  \ottsym{[}   \tmTopC   \ottsym{]} }{  \tmTopC  } 

    {\setlength{\fboxsep}{5pt}\colorbox{lightgray}{$ \step{ \tmMeta }{  \tmBotC  } $}} \\

     \step{  \tmApp{ \ottsym{(}   \valLam{ \mathit{x} }{ \tmMeta }   \ottsym{)} }{ \valMeta }  }{  \subst{ \tmMeta }{ \valMeta }{ \mathit{x} }  }  \and

     \step{ \ottkw{let} \, \ottsym{(}  \mathit{x_{{\mathrm{1}}}}  \ottsym{,}  \mathit{x_{{\mathrm{2}}}}  \ottsym{)}  \ottsym{=}   \tmPair{ \valMeta_{{\mathrm{1}}} }{ \valMeta_{{\mathrm{2}}} }  \, \ottkw{in} \, \tmMeta }{  \subst{  \subst{ \tmMeta }{ \valMeta_{{\mathrm{1}}} }{ \mathit{x_{{\mathrm{1}}}} }  }{ \valMeta_{{\mathrm{2}}} }{ \mathit{x_{{\mathrm{2}}}} }  } 
    \and

     \step{ \ottkw{let} \, \ottnt{s}  \ottsym{=}  \ottnt{s'} \, \ottkw{in} \, \tmMeta }{ \tmMeta } \ \text{ where }\ \ottnt{s}  \leq  \ottnt{s'}
    \and

     \tmBigJoin{ \mathit{x} }{ \ottsym{\{}  \valMeta_{{\mathrm{1}}}  \ottsym{,} \, ... \, \ottsym{,}  \valMeta_{\ottmv{n}}  \ottsym{\}} }{ \tmMeta }   \stepsym 
      \tmJoin{ \subst{ \tmMeta }{ \valMeta_{{\mathrm{1}}} }{ \mathit{x} } }{\tmJoin{\cdots}{ \subst{ \tmMeta }{ \valMeta_{\ottmv{n}} }{ \mathit{x} } }}
      \and

     \step{  \tmJoin{ \resMeta_{{\mathrm{1}}} }{ \resMeta_{{\mathrm{2}}} }  }{  \synjoin{ \resMeta_{{\mathrm{1}}} }{ \resMeta_{{\mathrm{2}}} }  } \and

     \step{ \ottsym{\{}  \tmMeta_{{\mathrm{1}}}  \ottsym{,} \, ... \, \ottsym{,}  \tmMeta_{\ottmv{n}}  \ottsym{,}   \tmBotC   \ottsym{,}  \tmMeta'_{{\mathrm{1}}}  \ottsym{,} \, ... \, \ottsym{,}  \tmMeta'_{\ottmv{m}}  \ottsym{\}} }{ \ottsym{\{}  \tmMeta_{{\mathrm{1}}}  \ottsym{,} \, ... \, \ottsym{,}  \tmMeta_{\ottmv{n}}  \ottsym{,}  \tmMeta'_{{\mathrm{1}}}  \ottsym{,} \, ... \, \ottsym{,}  \tmMeta'_{\ottmv{m}}  \ottsym{\}} } \and


  \end{drulepar}
}
  \begin{minipage}{1\textwidth}
    {\small
    \fbox{$ \synjoin{ \resMeta }{ \resMeta' } $}
    \[
       \synjoin{ \resMeta }{ \resMeta' }  = \left\{
        \begin{array}{ll@{\quad}|ll}
      \resMeta  & \text{if } \resMeta'  \ottsym{=}   \tmBotC    &  \ottnt{s_{{\mathrm{1}}}}  \sqcup  \ottnt{s_{{\mathrm{2}}}}  & \text{if } \resMeta_{{\mathrm{1}}}  \ottsym{=}  \ottnt{s_{{\mathrm{1}}}} \text{ and } \resMeta_{{\mathrm{2}}}  \ottsym{=}  \ottnt{s_{{\mathrm{2}}}}\\
      \resMeta' & \text{if } \resMeta  \ottsym{=}   \tmBotC     &  \pairop{  \synjoin{ \valMeta_{{\mathrm{1}}} }{ \valMeta'_{{\mathrm{1}}} }  }{  \synjoin{ \valMeta_{{\mathrm{2}}} }{ \valMeta'_{{\mathrm{2}}} }  }  & \text{if } \resMeta  \ottsym{=}   \tmPair{ \valMeta_{{\mathrm{1}}} }{ \valMeta_{{\mathrm{2}}} }  \text{ and }      \resMeta'  \ottsym{=}   \tmPair{ \valMeta'_{{\mathrm{1}}} }{ \valMeta'_{{\mathrm{2}}} }  \\
      \valMeta & \text{if } \resMeta  \ottsym{=}  \valMeta \text{ and } \resMeta'  \ottsym{=}   \valBotV  & \ottsym{\{}  \valMeta_{{\mathrm{1}}}  \ottsym{,} \, ... \, \ottsym{,}  \valMeta_{\ottmv{n}}  \ottsym{,}  \valMeta'_{{\mathrm{1}}}  \ottsym{,} \, ... \, \ottsym{,}  \valMeta'_{\ottmv{m}}  \ottsym{\}}
        & \text{if } \resMeta  \ottsym{=}  \ottsym{\{}  \valMeta_{{\mathrm{1}}}  \ottsym{,} \, ... \, \ottsym{,}  \valMeta_{\ottmv{n}}  \ottsym{\}} \textand
          \resMeta'  \ottsym{=}  \ottsym{\{}  \valMeta'_{{\mathrm{1}}}  \ottsym{,} \, ... \, \ottsym{,}  \valMeta'_{\ottmv{m}}  \ottsym{\}} \\
      \valMeta' & \text{if } \resMeta  \ottsym{=}   \valBotV  \text{ and } \resMeta'  \ottsym{=}  \valMeta' &    \tmJoin{  \valLam{ \mathit{x} }{ \tmMeta }  }{ \tmMeta' }  & \text{if } \resMeta  \ottsym{=}   \valLam{ \mathit{x} }{ \tmMeta }  \textand  \resMeta'  \ottsym{=}   \valLam{ \mathit{x} }{ \tmMeta' } \\
      & &  \tmTopC  & \text{otherwise}
        \end{array}\right.
    \]
    {\small
    \fbox{$ \pairop{ \resMeta }{ \resMeta' } $}
    \[
       \pairop{ \resMeta }{ \resMeta' }  =
      \left\{\begin{array}{ll}
         \tmBotC  & \text{if } \resMeta  \ottsym{=}   \tmBotC  \text{ or } (\resMeta  \ottsym{=}  \valMeta \textand \resMeta'  \ottsym{=}   \tmBotC )\\
         \tmTopC  & \text{if } \resMeta  \ottsym{=}   \tmTopC  \text{ or } (\resMeta  \ottsym{=}  \valMeta \textand \resMeta'  \ottsym{=}   \tmTopC )\\
         \tmPair{ \valMeta }{ \valMeta' }  & \text{if } \resMeta  \ottsym{=}  \valMeta \textand \resMeta'  \ottsym{=}  \valMeta'
      \end{array}\right.
    \]
    }
    }
  \end{minipage}
  }

  \caption{\gcalc Approximate Operational Semantics}
  \label{fig:opsem}
\end{figure}

\Fref{fig:opsem} defines a call-by-value semantics over the closed terms of
\gcalc.
To start, we will ignore the reduction rule highlighted in gray.
The reduction relation is defined to be closed over evaluation contexts.
The error $ \tmTopC $ propagates through these contexts.
Beta reduction for application makes use of the capture-avoiding substitution operation
written $ \subst{ \tmMeta }{ \valMeta }{ \mathit{x} } $ to mean $\tmMeta$ with all free occurrences of $\mathit{x}$
replaced by $\valMeta$.
Reduction of a pair elimination form $\ottkw{let} \, \ottsym{(}  \mathit{x_{{\mathrm{1}}}}  \ottsym{,}  \mathit{x_{{\mathrm{2}}}}  \ottsym{)}  \ottsym{=}   \tmPair{ \valMeta_{{\mathrm{1}}} }{ \valMeta_{{\mathrm{2}}} }  \, \ottkw{in} \, \tmMeta$
performs one substitution for each component of the pair.
The result is $ \subst{  \subst{ \tmMeta }{ \valMeta_{{\mathrm{1}}} }{ \mathit{x_{{\mathrm{1}}}} }  }{ \valMeta_{{\mathrm{2}}} }{ \mathit{x_{{\mathrm{2}}}} } $. Since reduction is defined over
closed terms, we need not worry that $\mathit{x_{{\mathrm{2}}}}$ might be free in $\valMeta_{{\mathrm{1}}}$.
Reduction for symbol elimination reduces only when given
a symbol meeting the threshold.

The $ \synjoin{ \resMeta }{ \resMeta' } $ metafunction defines how a join of two results evaluates.
Joins are distributed over abstractions.
They are also distributed over pairs but to obtain a well-formed result we must
a make use of a \emph{computational lifting} operation $ \pairop{ \resMeta_{{\mathrm{1}}} }{ \resMeta_{{\mathrm{2}}} } $.
This operation has an asymmetric definition which follows the left-to-right
sequential evaluation of pairs.
Symbols come with a primitive notion of join.
Joins of two unlike values, such as a pair with a function or two incomparable symbols with each other, result in $ \tmTopC $,
which we refer to as an \emph{ambiguity error}.
Joins involving $ \tmBotC $, $ \tmTopC $, and $ \valBotV $ are defined
according to the laws of bounded semilattices.

Because the decomposition of a term into an evaluation context and a redex is not unique,
reduction is nondeterministic.
Under the rules we are currently considering, the only source of nondeterminism
is the ability to reduce on either side of a join and in any position within a set.
It is therefore evident that the system so far is confluent.

\subsection{Dealing with Nontermination}
\label{sec:opsem:nondet}

As we have seen, nonterminating \gcalc programs like
$ \tmApp{ \mathit{fromN} }{ \ottsym{0} } $ and $ \tmApp{ \mathit{evens} }{ \ottsym{()} } $ can have non-trivial meaning.
It is not straightforward to capture this meaning in the semantics, however.
To see why, let the function $\mathit{head}$ be defined as
$ \valLam{ \mathit{x} }{ \ottkw{let} \, h  \ottsym{::}  \mathit{\_}  \ottsym{=}  \mathit{x} \, \ottkw{in} \, h } $.
When applied to a list, $\mathit{head}$ should return the first element of its argument.
Unfortunately, the reduction rules we have discussed so far do not reflect this
intended meaning.
Since the expression $ \tmApp{ \mathit{fromN} }{ \ottsym{0} } $ represents an
infinite stream, it never runs to a value.
Thus, in the program $ \tmApp{ \mathit{head} }{ \ottsym{(}   \tmApp{ \mathit{fromN} }{ \ottsym{0} }   \ottsym{)} } $, evaluation never reaches the body of
$\mathit{head}$.
It, like all call-by-value functions, requires a value as input.
Consequently, the program enters a meaningless infinite loop.

Our solution to this conundrum is to introduce another source of nondeterminism
in the form of \emph{approximation steps}.
The highlighted rule in \Fref{fig:opsem} states that a running program is able to
nondeterministically throw away its output by stepping to $ \tmBotC $.
In this way, programs reduce to their observations.
For example, we can rewrite an abbreviated version of the table from \Fref{fig:fromN} as:
\begin{center}
\begin{small}
\begin{tikzcd}
   \tmApp{ \mathit{fromN} }{ \ottsym{0} } 
  \dar[mapsto, shorten >=1.75pt][pos=1]{*}
  \rar[mapsto, shorten >=1.75pt][pos=1]{*} &
   \tmJoin{ \ottsym{(}   \tmApp{ \ottsym{0}  \ottsym{::}  \mathit{fromN} }{ \ottsym{1} }   \ottsym{)} }{  \valBotV  } 
  \dar[mapsto, shorten >=1.75pt][pos=1]{*}
  \rar[mapsto, shorten >=1.75pt][pos=1]{*} &
   \tmJoin{ \ottsym{(}  \ottsym{0}  \ottsym{::}  \ottsym{(}   \tmJoin{ \ottsym{(}   \tmApp{ \ottsym{1}  \ottsym{::}  \mathit{fromN} }{ \ottsym{2} }   \ottsym{)} }{  \valBotV  }   \ottsym{)}  \ottsym{)} }{  \valBotV  } 
  \dar[mapsto, shorten >=1.75pt][pos=1]{*}
  \rar[mapsto, shorten >=1.75pt][pos=1]{*} &
  \cdots
  \\
   \tmBotC  &  \valBotV  &  \ottsym{0}  \ottsym{::}   \valBotV 
\end{tikzcd}
\end{small}
\end{center}
\noindent
It follows that $ \tmApp{ \mathit{head} }{ \ottsym{(}   \tmApp{ \mathit{fromN} }{ \ottsym{0} }   \ottsym{)} }   \stepsym^{*}   \tmApp{ \mathit{head} }{ \ottsym{(}  \ottsym{0}  \ottsym{::}   \valBotV   \ottsym{)} }   \stepsym^{*}  \ottsym{0}$.

The presence of nondeterministic approximation steps means that a single trace does not
in general capture the entire meaning of a program.\footnote{Our use of nondeterminism is
loosely inspired by the semantics of the concurrent lambda calculus of
\citet{dezani94} and strongly resembles the \emph{clairvoyant} semantics of \citet{hackett19}.} %
Said differently, the existence of a reduction sequence $\tmMeta  \stepsym^{*}  \resMeta$ does
\emph{not} indicate $\resMeta$ is the unique or best result that $\tmMeta$ might produce.
Rather, $\resMeta$ is merely one possible approximation of the full meaning of $\tmMeta$.
Indeed, infinite computations like $ \tmApp{ \mathit{fromN} }{ \ottsym{0} } $ have no ``best result''
expressible in the syntax.
To fully understand the meaning of an expression, we need to take into account the
whole (often infinite) set of results it reduces to.

As another example, take the program $ \tmBigJoin{ \mathit{x} }{  \tmApp{ \mathit{evens} }{ \ottsym{()} }  }{ \ottkw{let} \, \ottsym{2}  \ottsym{=}  \mathit{x} \, \ottkw{in} \,  \litString{success}  } $.
As $\mathit{evens}$ generates the set of all even natural numbers, this program
is intended to search for the element $\ottsym{2}$ in the set and, if the search succeeds,
evaluate to the string $ \litString{success} $.
However, we face the same issue as before: without approximation steps, the infinite set $ \tmApp{ \mathit{evens} }{ \ottsym{()} } $ would never reduce to a value so the reduction rule for the big join operator would never fire.
Making use of approximation steps, we have:
\[{\small
  \begin{array}{rllll}
   \tmApp{ \mathit{evens} }{ \ottsym{()} }  &  \stepsym^{*}  &  \tmApp{  \tmJoin{ \ottsym{\{}  \ottsym{0}  \ottsym{\}} }{ \mathit{plus2all} }  }{ \ottsym{(}   \tmApp{ \mathit{evens} }{ \ottsym{()} }   \ottsym{)} }  &  \stepsym^{*}  &  \tmApp{  \tmJoin{ \ottsym{\{}  \ottsym{0}  \ottsym{\}} }{ \mathit{plus2all} }  }{ \ottsym{(}   \tmApp{  \tmJoin{ \ottsym{\{}  \ottsym{0}  \ottsym{\}} }{ \mathit{plus2all} }  }{ \ottsym{(}   \tmApp{ \mathit{evens} }{ \ottsym{()} }   \ottsym{)} }   \ottsym{)} }  \\
  &  \stepsym^{*}  &  \tmApp{  \tmJoin{ \ottsym{\{}  \ottsym{0}  \ottsym{\}} }{ \mathit{plus2all} }  }{ \ottsym{(}   \tmJoin{ \ottsym{\{}  \ottsym{0}  \ottsym{\}} }{  \tmBotC  }   \ottsym{)} }  &  \stepsym^{*}  &  \tyJoinSyn{ \ottsym{\{}  \ottsym{0}  \ottsym{\}} }{ \ottsym{\{}  \ottsym{2}  \ottsym{\}} }  \ \  \stepsym  \ \  \ottsym{\{}  \ottsym{0}  \ottsym{,}  \ottsym{2}  \ottsym{\}} \\
  \end{array}
  }
\]
This unblocks reduction in our example:
\[{\small
  \begin{array}{rll}
   \tmBigJoin{ \mathit{x} }{  \tmApp{ \mathit{evens} }{ \ottsym{()} }  }{ \ottkw{let} \, \ottsym{2}  \ottsym{=}  \mathit{x} \, \ottkw{in} \,  \litString{success}  }  &  \stepsym^{*}  &
   \tmBigJoin{ \mathit{x} }{ \ottsym{\{}  \ottsym{0}  \ottsym{,}  \ottsym{2}  \ottsym{\}} }{ \ottkw{let} \, \ottsym{2}  \ottsym{=}  \mathit{x} \, \ottkw{in} \,  \litString{success}  }  \\
  &  \stepsym^{*}  &  \tmJoin{ \ottsym{(}  \ottkw{let} \, \ottsym{2}  \ottsym{=}  \ottsym{0} \, \ottkw{in} \,  \litString{success}   \ottsym{)} }{ \ottsym{(}  \ottkw{let} \, \ottsym{2}  \ottsym{=}  \ottsym{2} \, \ottkw{in} \,  \litString{success}   \ottsym{)} }  \\
  &  \stepsym^{*}  &  \tmJoin{  \tmBotC  }{  \litString{success}  }  \ \  \stepsym  \ \  \litString{success}  \\

  \end{array}
}\]
We can see that approximation steps are useful for cutting off infinite
recursion as well as discarding otherwise stuck terms like
$\ottkw{let} \, \ottsym{2}  \ottsym{=}  \ottsym{0} \, \ottkw{in} \,  \litString{success} $.

Approximation steps enable a \gcalc function to compute (part of) its output
without having its entire input available; they model
\emph{pipeline parallelism}.
Their nondeterministic nature makes possible a declarative approach to describing
the operation of \gcalc programs in which the technical complexities of
scheduling are left implicit.
As we will see shortly, this relatively simple semantics is suitable for the
study of contextual equivalence.
On the other hand, its nondeterminism means that it does
not immediately give rise to an implementation.
We revisit this issue and discuss pipeline parallelism more explicitly in \Sref{sec:disc:impl}.


\paragraph{Convergence \& Approximation}
We define the \emph{convergence} of an expression as the existence of a
non-$ \tmBotC $ result that the expression can reduce to.
We write \(  \returns{ \tmMeta }{ \resMeta }  \) iff \( \tmMeta  \stepsym^{*}  \resMeta \textand \resMeta \neq  \tmBotC  \).
The result may be omitted for brevity; the notation \(  \returnsany{ \tmMeta }  \) means that some
such $\resMeta$ exists.
Given this, we define a notion of \emph{contextual approximation}.
Here, a program context $\ottnt{C}$ is an expression with one subexpression replaced with a hole $ \ctxHole $.
We write $\ottnt{C}  \ottsym{[}  \tmMeta  \ottsym{]}$ to mean $\ottnt{C}$ with $\tmMeta$ filled in for the hole.
Contextual approximation is defined as: $ \ctxapx{ \tmMeta_{{\mathrm{1}}} }{ \tmMeta_{{\mathrm{2}}} } $ iff $\mathforall{\ottnt{C}}{\, \returnsany{ \ottnt{C}  \ottsym{[}  \tmMeta_{{\mathrm{1}}}  \ottsym{]} }  \Rightarrow  \returnsany{ \ottnt{C}  \ottsym{[}  \tmMeta_{{\mathrm{2}}}  \ottsym{]} } }$.
\emph{Contextual equivalence}, written $ \ctxeq{ \tmMeta_{{\mathrm{1}}} }{ \tmMeta_{{\mathrm{2}}} } $, is contextual approximation in both directions.

\begin{remark}
Given that the goal of \gcalc is purportedly to enable
deterministic-by-construction parallel programming and the reduction system we
have studied for it is unapologetically nondeterministic, some reassurance is in
order.
Confluence, a conventional approach to arguing that nondeterministic rewriting
systems behave deterministically in a global sense is not appropriate in our setting.
The full system of reduction from \Fref{fig:opsem} is confluent, but in a
disappointingly trivial way: thanks to approximation steps, all terms reduce to $ \tmBotC $.

As noted in the introduction, our desired notion of determinism has to do with
the meaning of a program from the infinite limit perspective.
Thus, rather than concern ourselves with determinism here, it is better to
approach the property using the denotational semantics we will soon construct.
We return to the property in \Sref{sec:log-sem:domains}.
\end{remark}



\section{Logical Semantics: A Filter Model}
\label{sec:logsem}

\begin{figure}
  {\small
  \[
    \begin{array}{lrcrcl}
      \textit{computation formulae} &
       \synctype &\ni & \tyAc\longonly{, \tyBc} & ::= &
           \tmBotC  \sep  \tmTopC  \sep \tyAv \\
      \textit{value formulae} &
        \synvtype &\ni & \tyAv, \tyBv & ::= &
           \valBotV 
          \sep \ottnt{s}
          \sep \ottsym{(}  \tyAv_{{\mathrm{1}}}  \ottsym{,}  \tyAv_{{\mathrm{2}}}  \ottsym{)}
          \sep \ottsym{\{}  \tyAv_{\ottmv{i}}  \ottsym{\mbox{$\mid$}}   \mathin{ \mathit{i} }{ \mathit{I} }   \ottsym{\}}
          \sep  \tyFun{  \mathin{ \mathit{i} }{ \mathit{I} }  }{ \tyAv_{\ottmv{i}} }{ \tyAc_{\ottmv{i}} }  \\
      \textit{environments} &
        \synenv &\ni & \Gamma & ::= &
           \envEmpty  \sep \Gamma  \ottsym{,}  \mathit{x}  \ottsym{:}  \tyAv

    \end{array}
  \]}
{\small
  \drules[]{$ \tyapx{ \tyAc }{ \tyAc' } $}{streaming order}
         {TApxBot, TApxBotV, TApxTop,
           TApxSym, TApxPair}

\(\ottdruleTApxSet{}
\quad
\ottdrule[{}]{%
\ottpremise{ \mathforall{  \mathin{ \mathit{i} }{ \mathit{I} }  }{  \mathexists{  \mathit{J'}  \subseteq  \mathit{J}  }{  \tyapx{  \tyBigJoin{  \mathin{ \mathit{j} }{ \mathit{J'} }  }{ \tau'_{\ottmv{j}} }  }{ \tau_{\ottmv{i}} }  }  } 
\quad \tyapx{ \tyAc_{\ottmv{i}} }{  \tyBigJoin{  \mathin{ \mathit{j} }{ \mathit{J'} }  }{ \tyAc'_{\ottmv{j}} }  } }%
}{
 \tyapx{  \tyFun{  \mathin{ \mathit{i} }{ \mathit{I} }  }{ \tau_{\ottmv{i}} }{ \tyAc_{\ottmv{i}} }  }{  \tyFun{  \mathin{ \mathit{j} }{ \mathit{J} }  }{ \tau'_{\ottmv{j}} }{ \tyAc'_{\ottmv{j}} }  } }{%
{\ottdrulename{TApxFun}}{}%
}\)
}
  \caption{\gcalc Filter Model Formulae}
  \label{fig:formulae}
\end{figure}
\begin{figure}
  \begin{minipage}[t]{0.48\textwidth}
    {\small
    \drulesectionhead{$ \tyPairOp{ \tyAc_{{\mathrm{1}}} }{ \tyAc_{{\mathrm{2}}} } $}{pair lifting}
    \[
       \tyPairOp{ \tyAc_{{\mathrm{1}}} }{ \tyAc_{{\mathrm{2}}} }  = \left\{\begin{array}{ll}
         \tmTopC  & \text{if } \tyAc_{{\mathrm{1}}}  \ottsym{=}   \tmTopC  \text{ or } (\tyAc_{{\mathrm{1}}}  \ottsym{=}  \tyAv_{{\mathrm{1}}} \textand \tyAc_{{\mathrm{2}}}  \ottsym{=}   \tmTopC ) \\
         \tmBotC  & \text{if } \tyAc_{{\mathrm{1}}}  \ottsym{=}   \tmBotC  \text{ or } (\tyAc_{{\mathrm{1}}}  \ottsym{=}  \tyAv_{{\mathrm{1}}} \textand \tyAc_{{\mathrm{2}}}  \ottsym{=}   \tmBotC ) \\
        \ottsym{(}  \tyAv_{{\mathrm{1}}}  \ottsym{,}  \tyAv_{{\mathrm{2}}}  \ottsym{)} & \text{if } \tyAc_{{\mathrm{1}}}  \ottsym{=}  \tyAv_{{\mathrm{1}}} \textand \tyAc_{{\mathrm{2}}}  \ottsym{=}  \tyAv_{{\mathrm{2}}}
      \end{array}\right.
    \]
    }
  \end{minipage}\hfill
  \begin{minipage}[t]{0.48\textwidth}
    {\small
    \drulesectionhead{$ \tySng{ \tyAc } $}{singleton lifting}
    \[
       \tySng{ \tyAc }  = \left\{\begin{array}{ll}
         \tmTopC  & \text{if } \tyAc  \ottsym{=}   \tmTopC  \\
         \tmBotC  & \text{if } \tyAc  \ottsym{=}   \tmBotC  \\
        \ottsym{\{}  \tyAv  \ottsym{\}} & \text{if } \tyAc  \ottsym{=}  \tyAv \\
      \end{array}\right.
    \]
    }
  \end{minipage}

  {\small
  \drulesectionhead{$ \tyJoin{ \tyAc_{{\mathrm{1}}} }{ \tyAc_{{\mathrm{2}}} } $}{formula join}
  \[
       \tyJoin{ \tyAc_{{\mathrm{1}}} }{ \tyAc_{{\mathrm{2}}} }  = \left\{\begin{array}{ll@{\ }|ll}
        \tyAc_{{\mathrm{1}}} & \text{if } \tyAc_{{\mathrm{2}}}  \ottsym{=}   \tmBotC  &         \ottnt{s_{{\mathrm{1}}}}  \sqcup  \ottnt{s_{{\mathrm{2}}}}  & \text{if } \tyAc_{{\mathrm{1}}}  \ottsym{=}  \ottnt{s_{{\mathrm{1}}}} \text{ and } \tyAc_{{\mathrm{2}}}  \ottsym{=}  \ottnt{s_{{\mathrm{2}}}} \text{\todo{}}\\
        \tyAc_{{\mathrm{2}}} & \text{if } \tyAc_{{\mathrm{1}}}  \ottsym{=}   \tmBotC  &         \tyPairOp{  \tyJoin{ \tyAv'_{{\mathrm{1}}} }{ \tyAv'_{{\mathrm{2}}} }  }{  \tyJoin{ \tyAv''_{{\mathrm{1}}} }{ \tyAv''_{{\mathrm{2}}} }  }  & \text{if } \tyAc_{{\mathrm{1}}}  \ottsym{=}  \ottsym{(}  \tyAv'_{{\mathrm{1}}}  \ottsym{,}  \tyAv''_{{\mathrm{1}}}  \ottsym{)} \text{ and } \tyAc_{{\mathrm{2}}}  \ottsym{=}  \ottsym{(}  \tyAv'_{{\mathrm{2}}}  \ottsym{,}  \tyAv''_{{\mathrm{2}}}  \ottsym{)} \\
        \tyAv_{{\mathrm{1}}} & \text{if } \tyAc_{{\mathrm{1}}}  \ottsym{=}  \tyAv_{{\mathrm{1}}} \text{ and } \tyAc_{{\mathrm{2}}}  \ottsym{=}   \valBotV  &        \ottsym{\{}  \tyAv_{\ottmv{i}}  \ottsym{\mbox{$\mid$}}   \mathin{ \mathit{i} }{  \mathit{I_{{\mathrm{1}}}}  \cup  \mathit{I_{{\mathrm{2}}}}  }   \ottsym{\}} & \text{if } \tyAc_{{\mathrm{1}}}  \ottsym{=}  \ottsym{\{}  \tyAv_{\ottmv{i}}  \ottsym{\mbox{$\mid$}}   \mathin{ \mathit{i} }{ \mathit{I_{{\mathrm{1}}}} }   \ottsym{\}} \text{ and } \tyAc_{{\mathrm{2}}}  \ottsym{=}  \ottsym{\{}  \tyAv_{\ottmv{i}}  \ottsym{\mbox{$\mid$}}   \mathin{ \mathit{i} }{ \mathit{I_{{\mathrm{2}}}} }   \ottsym{\}} \\
        \tyAv_{{\mathrm{2}}} & \text{if } \tyAc_{{\mathrm{1}}}  \ottsym{=}   \valBotV  \text{ and } \tyAc_{{\mathrm{2}}}  \ottsym{=}  \tyAv_{{\mathrm{2}}} &         \tyFun{  \mathin{ \mathit{i} }{  \mathit{I_{{\mathrm{1}}}}  \cup  \mathit{I_{{\mathrm{2}}}}  }  }{ \tyAv'_{\ottmv{i}} }{ \tyAc'_{\ottmv{i}} }  & \text{if } \tyAc_{{\mathrm{1}}}  \ottsym{=}   \tyFun{  \mathin{ \mathit{i} }{ \mathit{I_{{\mathrm{1}}}} }  }{ \tyAv'_{\ottmv{i}} }{ \tyAc'_{\ottmv{i}} }  \text{ and }\tyAc_{{\mathrm{2}}}  \ottsym{=}   \tyFun{  \mathin{ \mathit{i} }{ \mathit{I_{{\mathrm{2}}}} }  }{ \tyAv'_{\ottmv{i}} }{ \tyAc'_{\ottmv{i}} } \\
        &&  \tmTopC  & \text{otherwise}
      \end{array}\right.
  \]}

\vspace{-2ex}
  \caption{Operations on Formulae}
  \label{fig:formula-ops}
\end{figure}
\begin{figure}
  {\small
  \drules[]{$ \judg{ \Gamma }{ \tmMeta }{ \tyAc } $}{formula assignment}
         {TSub, TBot, TBotV, TTop, TVar,
           TJoin, TSym, TPair,
           TSet, TFun, TLetSym,
           TLetPair, TForIn, TApp,
           TLetPairTop, TLetSymTop,TAppLTop,TAppRTop,
           TForInTop
         }
       }
       \vspace{-2ex}
  \caption{\gcalc Filter Model Formula Assignment}
  \label{fig:typing}
\end{figure}

A useful way of reasoning about \gcalc programs is by defining a denotational
semantics that captures their full (possibly infinite) meaning.
It turns out type systems are a useful tool for constructing such a
semantics, even though \gcalc is an untyped language.
Following \citet{barendregt83} and \citet{dezani94},
we will define a type system so precise that every part of the
behavior of a term can be described by a type.
In other words (as proven in \Sref{sec:logsem:results}), if two terms are assigned exactly
the same types, then they are contextually equivalent.
Thus, we can define the meaning of a term as the set of types that can be assigned
to it.
This construction is known as a \emph{filter model}.



Our filter model gives a ``logical'' semantics by assigning logical formulae
(which are essentially types) to terms via a system of inference rules.
Intuitively, logical formulae represent the \emph{finite} ``behaviors'' that a
term may have such as
``being a set containing at least the elements 1, 2, and 3,'' or
``behaving as a piecewise function that at least maps $ \symName{true} $ to
$ \symName{false} $ and $ \symName{false} $ to $ \symName{true} $.''
Terms with infinite behaviors, such as the set $ \tmApp{ \mathit{evens} }{ \ottsym{()} } $ from the
introduction, can still be handled; they are assigned an infinite number of
finite formulae.

\begin{remark}
It is natural to explicitly define the streaming order, which we have so far
discussed in informal terms, for logical formulae. \iflong
This definition of the streaming order is interesting in that it is analogous to
two well-known concepts that are usually thought of as distinct:
\begin{enumerate}
\item It corresponds to Scott's order of approximation on denotations.
\item It coincides with the \emph{opposite} of the usual order on types: the
  classic subtyping relation used in filter models and popularized by Cardelli.
\end{enumerate}
We choose to follow the order long used for denotational semantics by Scott, with
regret that the conventions of Scott and Cardelli are inconsistent with each
other.
\fi
\shortonly{Its definition, designed according to the order long used for denotational semantics by Scott, corresponds to the \emph{opposite} of the classic subtyping relation popularized by Cardelli.}
Consequently, joins of formulae in our setting play the role of
\emph{intersection types}~\cite{coppo78}.
The analog of the intersection type introduction rule will be shown to be
admissible in \Lref{lem:log-sem:directed}.
Our choice of order means that we are building a model of
\emph{ideals} (the order-theoretic dual of filters), but we use the
term \emph{filter model} to avoid confusion when comparing with past work.
\end{remark}

\subsection{Formulae \& Assignment}

The top of \Fref{fig:formulae} describes our logical formulae and the streaming
order over them.
Taking inspiration from call-by-push-value semantics~\cite{cbpv}, the
metavariable $\tyAc$ ranges over \emph{computation formulae} which describe the
behavior of all terms including both those that may fail and those that produce
a value.
\emph{Value formulae}, ranged over by $\tyAv$ and $\tyBv$, describe the behavior of terms that produce a successful
result.
Value formulae include the syntactic base values as well as pairs of value formulae.
We assume $\mathit{I}$ and $\mathit{J}$ range over \emph{finite} index sets.
Thus, the formula $\ottsym{\{}  \tyAv_{\ottmv{i}}  \ottsym{\mbox{$\mid$}}   \mathin{ \mathit{i} }{ \mathit{I} }   \ottsym{\}}$ contains a finite set of subformulae
of the shape $\tyAv_{\ottmv{i}}$.
The formula $ \tyFun{  \mathin{ \mathit{i} }{ \mathit{I} }  }{ \tyAv_{\ottmv{i}} }{ \tyAc_{\ottmv{i}} } $ is a join of a finite set of clauses.
Formulae of this shape are assigned to function values.
They describe the behavior of the function in terms of threshold queries; each clause
$ \tyFunOne{ \tyAv_{\ottmv{i}} }{ \tyAc_{\ottmv{i}} } $ (for some $ \mathin{ \mathit{i} }{ \mathit{I} } $) represents one such query in which
the input formula $\tyAv_{\ottmv{i}}$ is a threshold.
When this threshold is met by the input to the function, we say the clause for
$\mathit{i}$ is \emph{triggered} and the associated function produces a result of at
least $\tyAc_{\ottmv{i}}$.
The fact that function domains are restricted to value formulae reflects the
call-by-value nature of \gcalc.
As shorthand, we often omit the join symbol in the case $\mathit{I}$ is a singleton
or otherwise write it inline as in $ \tyFunOne{  \tyJoinSyn{  \tyFunOne{ \tyAv_{{\mathrm{1}}} }{ \tyAc_{{\mathrm{1}}} }  }{ \tyAv_{{\mathrm{2}}} }  }{ \tyAc_{{\mathrm{2}}} } $.
In formulae like this, the arrow constructor $\to$ binds tighter than joins.

Environments, ranged over by $\Gamma$, are finite partial mappings from
variables to value formulae.
The formula associated with a variable $\mathit{x}$ in $\Gamma$ is written
$\Gamma  \ottsym{(}  \mathit{x}  \ottsym{)}$.
The domain of $\Gamma$ is written $ \envDom{ \Gamma } $.
Environments separated by a comma are assumed to have disjoint domains.

The streaming order on formulae follows the order-theoretic intuition we
have seen so far.
In particular, \ottdrulename{TApxSet} states that as a set increases in the
streaming order, it may gain elements and existing elements may grow.
However, elements may not decrease or disappear completely.

Relating function formulae is a bit more involved.
We would like to define an order that somehow corresponds to the usual pointwise
ordering on functions.
In order theory, given functions $f$ and $g$ with domain $X$, we have
$\semapx{f}{g}$ iff $\mathforall{x \in X}{\semapx{f(x)}{g(x)}}$.
Suppose we have $\tyAv  \ottsym{=}   \tyFun{  \mathin{ \mathit{i} }{ \mathit{I} }  }{ \tyAv_{\ottmv{i}} }{ \tyAc_{\ottmv{i}} } $ and
$\tyAv'  \ottsym{=}   \tyFun{  \mathin{ \mathit{j} }{ \mathit{J} }  }{ \tyAv'_{\ottmv{j}} }{ \tyAc'_{\ottmv{j}} } $.
Consider an arbitrary input which we represent by the formula $\tyBv$.
Then, when applied to this input, the function denoted by $\tyAv$ will produce
an output denoted by $\tyAc  \ottsym{=}   \tyBigJoin{  \tyapx{ \tyAv_{\ottmv{i}} }{ \tyBv }  }{ \tyAc_{\ottmv{i}} } $.
This is the join of all of the outputs of the clauses of $\tyAv$ that are
triggered by $\tyBv$ (\ie the clauses whose input threshold $\tyBv$ meets).
Likewise, the corresponding output for $\tyAv'$ is $\tyAc'  \ottsym{=}   \tyBigJoin{  \tyapx{ \tyAv'_{\ottmv{j}} }{ \tyBv }  }{ \tyAc'_{\ottmv{j}} } $.
We need the definition of \ottdrulename{TApxFun} to ensure $ \tyapx{ \tyAv }{ \tyAv' } $ iff
$ \tyapx{ \tyAc }{ \tyAc' } $ for all $\tyBv$.
To do so, it requires that for each clause $ \tyFunOne{ \tyAv_{\ottmv{i}} }{ \tyAc_{\ottmv{i}} } $ of $\tyAv$ that
will be triggered by an input $\tyBv$ there exists a corresponding set of clauses
of $\tyAv'$, whose indices are given by $\mathit{J'}$, that meets two criteria:
\begin{enumerate}
  \item Each clause of $\mathit{J'}$ must be triggered by every input that
    might trigger the clause $ \tyFunOne{ \tyAv_{\ottmv{i}} }{ \tyAc_{\ottmv{i}} } $. In other words,
    $ \tyapx{  \tyBigJoin{  \mathin{ \mathit{j} }{ \mathit{J'} }  }{ \tyAv'_{\ottmv{j}} }  }{ \tyAv_{\ottmv{i}} } $.
  \item The combined output of all the clauses of $\mathit{J'}$ is at least $\tyAc_{\ottmv{i}}$.
    That is, $ \tyapx{ \tyAc_{\ottmv{i}} }{  \tyBigJoin{  \mathin{ \mathit{j} }{ \mathit{J'} }  }{ \tyAc'_{\ottmv{j}} }  } $.
\end{enumerate}
Note that in the case that $\mathit{I}$ and $\mathit{J}$ are each singleton sets,
\ottdrulename{TApxFun} specializes to the usual ordering for function types:
we have $ \tyapx{ \tyAv' }{ \tyAv } $ and $ \tyapx{ \tyAc }{ \tyAc' } $ imply $ \tyapx{  \tyFunOne{ \tyAv }{ \tyAc }  }{  \tyFunOne{ \tyAv' }{ \tyAc' }  } $.
Moreover, the following distributivity property holds.
\begin{lemma}
  \label{lem:log-sem:fun-join-dist}
  $ \tyapx{  \tyFunOne{ \tyAv }{ \ottsym{(}   \tyJoin{ \tyAc }{ \tyAc' }   \ottsym{)} }  }{  \tyJoinSyn{ \ottsym{(}   \tyFunOne{ \tyAv }{ \tyAc }   \ottsym{)} }{ \ottsym{(}   \tyFunOne{ \tyAv }{ \tyAc' }   \ottsym{)} }  } $
\end{lemma}

We lift the streaming order from formulae to environments, defining
the proposition ${ \envapx{ \Gamma }{ \Gamma' } }$ to hold iff ${ \envDom{ \Gamma }  \subseteq  \envDom{ \Gamma' } }$ and for all
$\mathit{x} \in  \envDom{ \Gamma } $ we have $ \tyapx{ \Gamma  \ottsym{(}  \mathit{x}  \ottsym{)} }{ \Gamma'  \ottsym{(}  \mathit{x}  \ottsym{)} } $.

Key operations on formulae are defined in \Fref{fig:formula-ops}.
The operations $ \tyPairOp{ \tyAc_{{\mathrm{1}}} }{ \tyAc_{{\mathrm{2}}} } $ and $ \tySng{ \tyAc } $  monadically
lift the construction of pairs and singleton sets from value formulae to
computation formulae in a way that mimics evaluation.
The figure also defines the join operation on formulae, written
$ \tyJoin{ \tyAc_{{\mathrm{1}}} }{ \tyAc_{{\mathrm{2}}} } $.
The definition resembles that of the corresponding operation on results
from \Sref{sec:opsem}.
The following properties will let us establish that it does indeed represent a least
upper bound and that all operations on formulae are monotone.
\begin{lemma}
  \label{lem:log-sem:formula-op-mono}
  \label{lem:log-sem:join-lub}
  The following rules are all admissible:
  {\small
  \begin{mathpar}
    \inferrule
        { \tyapx{ \tyAc_{{\mathrm{1}}} }{ \tyAc'_{{\mathrm{1}}} }  \\
           \tyapx{ \tyAc_{{\mathrm{2}}} }{ \tyAc'_{{\mathrm{2}}} } }
        { \tyapx{  \tyPairOp{ \tyAc_{{\mathrm{1}}} }{ \tyAc_{{\mathrm{2}}} }  }{  \tyPairOp{ \tyAc'_{{\mathrm{1}}} }{ \tyAc'_{{\mathrm{2}}} }  } }

    \inferrule
        { \tyapx{ \tyAc }{ \tyAc' } }
        { \tyapx{  \tySng{ \tyAc }  }{  \tySng{ \tyAc' }  } }

    \inferrule
        { \tyapx{ \tyAc' }{ \tyAc }  \\
           \tyapx{ \tyAc'' }{ \tyAc } }
        { \tyapx{  \tyJoin{ \tyAc' }{ \tyAc'' }  }{ \tyAc } }

        \inferrule
            { \tyapx{ \tyAc }{ \tyAc' } }
            { \tyapx{ \tyAc }{  \tyJoin{ \tyAc' }{ \tyAc'' }  } }

            \inferrule
                { \tyapx{ \tyAc }{ \tyAc'' } }
                { \tyapx{ \tyAc }{  \tyJoin{ \tyAc' }{ \tyAc'' }  } }
      \end{mathpar}
      }
\end{lemma}

We define the \textit{size} of a formula $\ottsym{\mbox{$\mid$}}  \tyAc  \ottsym{\mbox{$\mid$}}$ to be its height when viewed as a syntax
tree.
In proofs, it is often useful to perform induction on this metric.
This allows us to benefit from the following lemma, which states that the size of
the join of a pair of formulae is no larger than the size of the larger of the
two.
\begin{lemma}[Size of Joins]
  \label{lem:log-sem:join-size}

  $\ottsym{\mbox{$\mid$}}   \tyJoin{ \tyAc }{ \tyAc' }   \ottsym{\mbox{$\mid$}} \leq \max\{\ottsym{\mbox{$\mid$}}  \tyAc  \ottsym{\mbox{$\mid$}}, \ottsym{\mbox{$\mid$}}  \tyAc'  \ottsym{\mbox{$\mid$}}\}$
\end{lemma}
\noindent Thus, in a proof by induction on a formula, given induction hypotheses for some
finite set of subformulae, we also have an induction hypothesis for their least
upper bound.

The judgement $ \judg{ \Gamma }{ \tmMeta }{ \tyAc } $ indicates that under the environment $\Gamma$
the term $\tmMeta$ is assigned the formula $\tyAc$.
The intuition is that when given input for each $\mathit{x} \in  \envDom{ \Gamma } $ that is
\emph{at least} $\Gamma  \ottsym{(}  \mathit{x}  \ottsym{)}$, the term $\tmMeta$ has \emph{at least} the
behaviors of $\tyAc$.
Many different formulae may be assigned to the same term.
The definition of the formula assignment is given inductively in \Fref{fig:typing}.


A number of the inference rules are familiar from the literature on type systems
or are otherwise straightforward.
We now discuss the rest.
Joins are assigned the formula that is the least upper bound of the formulae
assigned to their subterms.
Assigning formulae to set literals is non-trivial.
Since expressions in a set literal evaluate in parallel, we want
$\ottsym{\{}  \tmMeta_{{\mathrm{1}}}  \ottsym{,} \, ... \, \ottsym{,}  \tmMeta_{\ottmv{n}}  \ottsym{\}}$ to be assigned formulae in the same way as
$\tmJoinMany{\ottsym{\{}  \tmMeta_{{\mathrm{1}}}  \ottsym{\}}}{\ottsym{\{}  \tmMeta_{\ottmv{n}}  \ottsym{\}}}$.
That is, the presence of $ \tmBotC $ in the set will not affect the final result
while the presence of $ \tmTopC $ makes the entire assigned formula $ \tmTopC $.
To achieve this, the rule \ottdrulename{TSet} uses a metafunction
$ \tySng{ \tyAc } $ which injects a value formula into the singleton set and
propagates errors.

The rule \ottdrulename{TFun} is essentially the usual typing rule for functions
except that here, we exploit the piecewise nature of function formulae which
allow them to describe a function's different output behaviors depending on which
input it is given.
To assign formulae to the set elimination form, the rule \ottdrulename{TForIn} computes an aggregate over the set $\tmMeta_{{\mathrm{1}}}$. It computes the join over all formulae that can be ascribed to the body $\tmMeta_{{\mathrm{2}}}$ of the join when the argument $\mathit{x}$ ranges over the formulae describing the elements of $\tmMeta_{{\mathrm{1}}}$.
The final rules of \Fref{fig:typing} mimic the propagation of the error $ \tmTopC $ through
evaluation contexts in the operational semantics.

\subsection{Properties of Formula Assignment}
We now build the metatheory supported by this inference system, omitting proofs when they are routine.
The first step is to verify that the streaming order on formulae is a
preorder.

\begin{lemma}[Reflexivity]
  For all $\tyAc$, we have $ \tyapx{ \tyAc }{ \tyAc } $.
\end{lemma}
\iflong
\begin{proof}
  Routine induction on $\tyAc$.
\end{proof}
\fi
\begin{lemma}[Transitivity]
  If $ \tyapx{ \tyAc_{{\mathrm{1}}} }{ \tyAc_{{\mathrm{2}}} } $ and $ \tyapx{ \tyAc_{{\mathrm{2}}} }{ \tyAc_{{\mathrm{3}}} } $ then $ \tyapx{ \tyAc_{{\mathrm{1}}} }{ \tyAc_{{\mathrm{3}}} } $.
\end{lemma}
\begin{proof}
  Induction on $\tyAc_{{\mathrm{2}}}$, using \Lref{lem:log-sem:join-size} and the accompanying induction principle. \longonly{In each
  case, we invert both
  premises.
  Due to the use of the join operator in the definition of the streaming order,
  we use \Lref{lem:log-sem:join-size} and \Lref{lem:log-sem:join-lub} in the
  function case.}
\end{proof}

With these results in hand, we can turn our attention to some properties of
formula assignment.
%
The definition of the formula assignment rules is structurally recursive in nature;
the formulae assigned to any term is based on the formulae assigned to its subterms.
This leads to the following compositionality principle.
\begin{lemma}[Compositionality]
  \label{lem:log-sem:compos}

  Suppose that $ \judg{ \Gamma' }{ \ottnt{C}  \ottsym{[}  \tmMeta_{{\mathrm{1}}}  \ottsym{]} }{ \tyAc' } $ and moreover for all $\Gamma$ and
  $\tyAc$ such that $ \judg{ \Gamma }{ \tmMeta_{{\mathrm{1}}} }{ \tyAc } $ we have $ \judg{ \Gamma }{ \tmMeta_{{\mathrm{2}}} }{ \tyAc } $.
  It then follows that $ \judg{ \Gamma' }{ \ottnt{C}  \ottsym{[}  \tmMeta_{{\mathrm{2}}}  \ottsym{]} }{ \tyAc' } $.
\end{lemma}

Next, we establish a standard weakening result.
\begin{lemma}[Weakening]
  \label{lem:log-sem:weaken}
  If $ \judg{ \Gamma' }{ \tmMeta }{ \tyAc } $ and $ \envapx{ \Gamma' }{ \Gamma } $ then $ \judg{ \Gamma }{ \tmMeta }{ \tyAc } $
\end{lemma}
\iflong
\begin{proof}
  Routine induction on $ \judg{ \Gamma' }{ \tmMeta }{ \tyAc } $.
\end{proof}
\fi
\noindent The following properties mean that the set of formulae assigned to a term
given a fixed environment is a non-empty downward-closed directed set known as an \emph{ideal}.
\nr{Should we define the ``denotation'' of a term earlier to make this more explicit?}
\begin{lemma}[Totality]
  \label{lem:log-sem:non-empty}

  For every $\Gamma$ and $\tmMeta$ there exists a formula $\tyAc$ such that $ \judg{ \Gamma }{ \tmMeta }{ \tyAc } $.
\end{lemma}
\iflong
\begin{proof}
  Immediate from rule \rulename{TBot}.
\end{proof}
\fi
\begin{lemma}[Downward Closure]
  \label{lem:log-sem:down-clos}

  If $ \judg{ \Gamma }{ \tmMeta }{ \tyAc' } $ and $ \tyapx{ \tyAc }{ \tyAc' } $ then $ \judg{ \Gamma }{ \tmMeta }{ \tyAc } $.
\end{lemma}
\iflong
\begin{proof}
  Immediate from rule \rulename{TSub}.
\end{proof}
\fi
\begin{lemma}[Directedness]
  \label{lem:log-sem:directed}
  If $ \judg{ \Gamma }{ \tmMeta }{ \tyAc } $ and $ \judg{ \Gamma }{ \tmMeta }{ \tyAc' } $ then $ \judg{ \Gamma }{ \tmMeta }{  \tyJoin{ \tyAc }{ \tyAc' }  } $.
\end{lemma}
\begin{proof}
  Induction on $\tmMeta$, inverting both premises and using
  Lemmas~\ref{lem:log-sem:fun-join-dist}, \ref{lem:log-sem:join-lub}, and \ref{lem:log-sem:weaken}.
\end{proof}

At this point we can begin to formally connect the logical semantics with the
approximate operational semantics.
One essential property from the literature on intersection types is a ``backwards
preservation'' lemma also known as \emph{subject expansion}.
In our setting, this tells us that given $\tmMeta  \stepsym^{*}  \tmMeta'$, every behavior of
$\tmMeta'$ is also a behavior of $\tmMeta$.
Proving this first requires a few inversion properties.
We use the notation $ \judg{ \Gamma }{ \substmeta }{ \Gamma' } $ to mean for all $ \mathin{ \mathit{x} }{  \envDom{ \Gamma' }  } $
we have $ \judg{ \Gamma }{ \substmeta  \ottsym{(}  \mathit{x}  \ottsym{)} }{ \Gamma'  \ottsym{(}  \mathit{x}  \ottsym{)} } $.

\iflong
\begin{lemma}[Inversion of Substitution Typing]
\else
\begin{lemma}
  \fi
  \label{lem:log-sem:subst-inv}

  If $ \judg{ \Gamma }{ \substmeta  \ottsym{(}  \tmMeta  \ottsym{)} }{ \tyAc } $ then there exists $\Gamma'$ such that
  $ \judg{ \Gamma }{ \substmeta }{ \Gamma' } $ and $ \judg{ \Gamma  \ottsym{,}  \Gamma' }{ \tmMeta }{ \tyAc } $.
\end{lemma}
\begin{proof}
  \iflong
  First, note that \Lref{lem:log-sem:directed} lifts to the typing of substitutions.
  That is, if we have environments $\Gamma_{{\mathrm{1}}}$ and $\Gamma_{{\mathrm{2}}}$ such that
  $ \judg{ \Gamma }{ \substmeta }{ \Gamma_{{\mathrm{1}}} } $ and $ \judg{ \Gamma }{ \substmeta }{ \Gamma_{{\mathrm{2}}} } $ then $\Gamma' =  \envJoin{ \Gamma_{{\mathrm{1}}} }{ \Gamma_{{\mathrm{2}}} } $ exists and
  $ \judg{ \Gamma }{ \substmeta }{ \Gamma' } $.
  With this in mind, we proceed by induction on $\tmMeta$, in each case inverting
  its derivation and making use of weakening and directedness.
  \else
  Induction on $\tmMeta$, noting that \Lref{lem:log-sem:directed} lifts to the
  typing of substitutions.
  \fi
\end{proof}
\iflong
\begin{lemma}[Inversion of Join Typing]
  \else
\begin{lemma}
\fi
  \label{lem:log-sem:join-typ-inv}
  If $ \judg{ \Gamma }{  \synjoin{ \resMeta_{{\mathrm{1}}} }{ \resMeta_{{\mathrm{2}}} }  }{ \tyAc } $ then there exists $\tyAc_{{\mathrm{1}}}$ and $\tyAc_{{\mathrm{2}}}$ such
  that $ \judg{ \Gamma }{ \resMeta_{{\mathrm{1}}} }{ \tyAc_{{\mathrm{1}}} } $ and $ \judg{ \Gamma }{ \resMeta_{{\mathrm{2}}} }{ \tyAc_{{\mathrm{2}}} } $ and $ \tyapx{ \tyAc }{  \tyJoin{ \tyAc_{{\mathrm{1}}} }{ \tyAc_{{\mathrm{2}}} }  } $.
\end{lemma}
\longonly{
\begin{proof}
  Induction on $\resMeta_{{\mathrm{1}}}$ and nested case analysis on $\resMeta_{{\mathrm{2}}}$.
  Uses \Lref{lem:log-sem:fun-join-dist} and \Lref{lem:log-sem:formula-op-mono}.
\end{proof}
}
\begin{lemma}
  \label{lem:log-sem:ectx-top}
  For all evaluation contexts $\ottnt{E}$ we have $ \judg{ \Gamma }{ \ottnt{E}  \ottsym{[}   \tmTopC   \ottsym{]} }{  \tmTopC  } $.
\end{lemma}
\iflong
\begin{proof}
  Routine induction on $\ottnt{E}$.
\end{proof}
\fi
\begin{lemma}[Subject Expansion]
  \label{lem:log-sem:subj-exp}
  If $ \step{ \tmMeta }{ \tmMeta' } $ and $ \judg{ \Gamma }{ \tmMeta' }{ \tyAc } $ then $ \judg{ \Gamma }{ \tmMeta }{ \tyAc } $.
\end{lemma}
\begin{proof}
  \iflong
  Induction on $ \step{ \tmMeta }{ \tmMeta' } $, inverting the derivation of $\tmMeta'$ in each case.
  The case in which $ \step{ \ottnt{E}  \ottsym{[}   \tmTopC   \ottsym{]} }{  \tmTopC  } $ follows from \Lref{lem:log-sem:ectx-top}.
  Beta reduction cases make use of \Lref{lem:log-sem:subst-inv}.
  The case where $ \step{  \tmJoin{ \resMeta_{{\mathrm{1}}} }{ \resMeta_{{\mathrm{2}}} }  }{  \synjoin{ \resMeta_{{\mathrm{1}}} }{ \resMeta_{{\mathrm{2}}} }  } $ follows from
  \Lref{lem:log-sem:join-typ-inv}.
  \else
  Induction on $ \step{ \tmMeta }{ \tmMeta' } $, using Lemmas~\ref{lem:log-sem:subst-inv}--\ref{lem:log-sem:ectx-top}.
  \fi
\end{proof}

\subsection{Semantic Results}
\label{sec:logsem:results}

We define the meaning of a closed term $\tmMeta$ as
\(  \denot{\tmMeta}{} = {\mathsetcomp{\tyAc}{ \judg{  \envEmpty  }{ \tmMeta }{ \tyAc } }} \).
The definition of \emph{logical approximation} for closed terms follows:
\[ \logapx{\tmMeta_{{\mathrm{1}}}}{\tmMeta_{{\mathrm{2}}}} \textiff \denot{\tmMeta_{{\mathrm{1}}}}{} \subseteq \denot{\tmMeta_{{\mathrm{2}}}}{} \]
This notion is the formalization of our intuition of streaming order.
It allows us to restate \Lref{lem:log-sem:compos} as a
monotonicity result:
\begin{theorem}[Monotonicity]
For any context $\ottnt{C}$ and $\logapx{\tmMeta}{\tmMeta'}$,
we have $\logapx{\ottnt{C}  \ottsym{[}  \tmMeta  \ottsym{]}}{\ottnt{C}  \ottsym{[}  \tmMeta'  \ottsym{]}}$.
\end{theorem}
\noindent Thus, the idea that every construct in \gcalc is monotone has been made formal.
An equivalent formulation of this theorem is the statement that logical
approximation is a \emph{precongruence} relation.

Our goal is now to show that logical approximation is included within contextual
approximation\longonly{, which justifies the view that the logical semantics are a device for
establishing contextual equivalences}.
This is a consequence of the following Soundness and Adequacy lemmas.

\begin{lemma}[Soundness]
  If $\tmMeta  \stepsym^{*}  \tmMeta'$ then $\logapx{\tmMeta'}{\tmMeta}$.
\end{lemma}
\begin{proof}
  Induction on $\tmMeta  \stepsym^{*}  \tmMeta'$, applying Subject Expansion
  (\Lref{lem:log-sem:subj-exp}) at each step.
\end{proof}
\begin{lemma}[Adequacy]
  \label{lem:logsem:adequacy}

  If $\logapx{\valMeta}{\tmMeta}$ then $ \returnsany{ \tmMeta } $.
\end{lemma}
\noindent The proof of Adequacy requires a relatively involved logical relations argument, so we
defer it to \Sref{sec:logsem:adequacy} to get to the main result of this section.

\begin{theorem}
  If $\logapx{\tmMeta_{{\mathrm{1}}}}{\tmMeta_{{\mathrm{2}}}}$
  then $\ctxapx{\tmMeta_{{\mathrm{1}}}}{\tmMeta_{{\mathrm{2}}}}$.
\end{theorem}
\begin{proof}
  Consider a context $\ottnt{C}$ such that $ \returns{ \ottnt{C}  \ottsym{[}  \tmMeta_{{\mathrm{1}}}  \ottsym{]} }{ \resMeta } $ where $\resMeta \neq  \tmBotC $.
  We must show $ \returnsany{ \ottnt{C}  \ottsym{[}  \tmMeta_{{\mathrm{2}}}  \ottsym{]} } $.

  We deduce the following:
  \[
    \begin{array}{rcll}
       \valBotV  &\logapxsym & \resMeta     &\text{Straightforward from the formula assignment rules.} \\
               &\logapxsym & \ottnt{C}  \ottsym{[}  \tmMeta_{{\mathrm{1}}}  \ottsym{]} &\text{Soundness} \\
               &\logapxsym & \ottnt{C}  \ottsym{[}  \tmMeta_{{\mathrm{2}}}  \ottsym{]} &\text{Monotonicity} \\
    \end{array}
  \]
  Therefore we may apply Adequacy to complete the proof.
\end{proof}

\subsection{Adequacy}
\label{sec:logsem:adequacy}

\begin{figure}
  {\small
  \[
  \begin{array}{rcl@{\qquad}rcl}
     \lrexp{ \tyAc }  &=  & \mathsetcomp{\tmMeta}{\mathexists{r}{\tmMeta  \stepsym^{*}  \resMeta} \textand  \mathin{ \resMeta }{  \lrres{ \tyAc }  } }
    &  \lrres{  \tmTopC  }  &= & \{  \tmTopC  \} \\
 &&&   \lrres{  \tmBotC  }  &= & \synres \\
 &&&   \lrres{ \tyAv }  &= &  \lrval{ \tyAv }  \cup \{  \tmTopC  \} \\ \\
  \end{array}
\]
\vspace{-2ex}
\[
  \begin{array}{rcl}
     \lrval{ \ottnt{s} }  &= & \mathsetcomp{\ottnt{s'}}{\ottnt{s}  \leq  \ottnt{s'}} \\
     \lrval{ \ottsym{(}  \tyAv_{{\mathrm{1}}}  \ottsym{,}  \tyAv_{{\mathrm{2}}}  \ottsym{)} }  &=
    & \mathsetcomp{ \tmPair{ \valMeta_{{\mathrm{1}}} }{ \valMeta_{{\mathrm{2}}} } }{ \mathin{ \valMeta_{{\mathrm{1}}} }{  \lrval{ \tyAv_{{\mathrm{1}}} }  }  \textand  \mathin{ \valMeta_{{\mathrm{2}}} }{  \lrval{ \tyAv_{{\mathrm{2}}} }  } } \\
     \lrval{ \ottsym{\{}  \tyAv_{\ottmv{i}}  \ottsym{\mbox{$\mid$}}   \mathin{ \mathit{i} }{ \mathit{I} }   \ottsym{\}} }  &=
    & \mathsetcomp{\ottsym{\{}  \valMeta_{\ottmv{j}}  \ottsym{\mbox{$\mid$}}   \mathin{ \mathit{j} }{ \mathit{J} }   \ottsym{\}}}
      {\mathexists{f \in I \to J}
         {\mathforall{ \mathin{ \mathit{j} }{ \mathit{J} } }
                     {\valMeta_{\ottmv{j}} \in \lrres{\tyBigJoin{\mathit{i} \in \mathinv{f}(\mathit{j})}
                                                    {\tyAv_{\ottmv{i}}}}}}} \\
     \lrval{  \tyFun{  \mathin{ \mathit{i} }{ \mathit{I} }  }{ \tyAv_{\ottmv{i}} }{ \tyAc_{\ottmv{i}} }  }  &=
    & \mathsetcomp{ \valLam{ \mathit{x} }{ \tmMeta } }
                  {\mathforall{ \mathit{I'}  \subseteq  \mathit{I} , \valMeta \in \lrval{\tyBigJoin{ \mathin{ \mathit{i} }{ \mathit{I'} } }{\tyAv_{\ottmv{i}}}}}
                              { \subst{ \tmMeta }{ \valMeta }{ \mathit{x} }  \in \lrexp{\tyBigJoin{ \mathin{ \mathit{i} }{ \mathit{I'} } }{\tyAc_{\ottmv{i}}}}}} \\ \\

\multicolumn{3}{l}{     \lrenv{ \Gamma }  =
       \mathsetcomp{\substmeta}{\mathforall{\mathit{x} \in  \envDom{ \Gamma } }{ \mathin{ \substmeta  \ottsym{(}  \mathit{x}  \ottsym{)} }{  \lrval{ \Gamma  \ottsym{(}  \mathit{x}  \ottsym{)} }  } }} 
\qquad \qquad
     \judglr{ \Gamma }{ \tmMeta }{ \tyAc } \ \text{ iff }
       \ \mathforall{ \mathin{ \substmeta }{  \lrenv{ \Gamma }  } }{ \mathin{ \substmeta  \ottsym{(}  \tmMeta  \ottsym{)} }{  \lrexp{ \tyAc }  } }}
  \end{array}
  \]
}
\vspace{-2ex}
  \caption{Logical Predicates}
  \label{fig:logrel}
\vspace{-2ex}
\end{figure}

\todo{Revise this first paragraph. The problem it is trying to communicate is
  unclear and the references to other work are opaque.
  Maybe talk about directedness and note a semantic analog is hard to prove?}

To prove \Lref{lem:logsem:adequacy}, we define the logical predicates given in
\Fref{fig:logrel}.
The general shape of our argument is standard, but the details are tricky.
Determinism and confluence of reduction tend to be important properties
for making use of operational logical predicates in the presence of intersection
types; adequacy proofs for models similar to ours make use of them freely~\cite{plfa}.
\longonly{\citet{dezani94} lack both, but are able to define a confluent
auxiliary relation containing the reduction relation and exploit the relationship
between the two.}
Unfortunately, the system of reduction from \Sref{sec:opsem} is not deterministic and only trivially confluent.
\longonly{Moreover, the technique of \citeauthor{dezani94} does not seem to be possible in our setting.}
Thus, designing the logical predicate in a way that its use will not depend on
determinism or confluence properties is a major challenge for us.
It turns out to be possible, but only through careful treatment of definitions involving
joins of formulae and through the strengthening of certain definitions so as to
provide an induction hypothesis capable of proving adequacy.

Given a formula $\tyAc$, the logical predicate interprets it as a set of closed
terms $ \lrexp{ \tyAc } $ which all have at least the operational behavior specified
by the formula.
The predicates $ \lrres{ \tyAc } $ and $ \lrval{ \tyAv } $ similarly define the
closed results and values associated with $\tyAc$ and $\tyAv$ respectively.
The first case of the value predicate states that any symbol $\ottnt{s'}$ can be
thought of as having the behavior of $\ottnt{s}$ so long as $\ottnt{s}  \leq  \ottnt{s'}$.
The second relates pairs of formulas to pairs of values in a pointwise fashion.

The manner in which the value predicate is defined for sets is a bit
counterintuitive.
Until now, we have suggested that one set is less than another when, for every
element $x$ of the first, there is a corresponding element $y$ of the second such that
$x$ is at least $y$.
This is, for example, the definition used in the order on formulas from
\rulename{TApxSet} in \Fref{fig:formulae}.
One might expect, then, a definition such as the following.
\begin{equation}
  \label{eq:logrel:vset}
   \lrval{ \ottsym{\{}  \tyAv_{\ottmv{i}}  \ottsym{\mbox{$\mid$}}   \mathin{ \mathit{i} }{ \mathit{I} }   \ottsym{\}} }  = \mathsetcomp{\ottsym{\{}  \valMeta_{\ottmv{j}}  \ottsym{\mbox{$\mid$}}   \mathin{ \mathit{j} }{ \mathit{J} }   \ottsym{\}}}
            {\mathforall{i \in I}
              {\mathexists{ \mathin{ \mathit{j} }{ \mathit{J} } }
                          { \mathin{ \valMeta_{\ottmv{j}} }{  \lrres{ \tyAv_{\ottmv{i}} }  } }}}
\end{equation}
\todo{Paragraph below is a bit opaque.}
Unfortunately, this turns out to not be strong
enough to prove the fundamental property of the logical
relation~(\Lref{lem:log-sem:fund-prop}).
\todo{revisit this point in the corresponding case of the proof of the fundamental
  property}
Instead we define a value $\ottsym{\{}  \valMeta_{\ottmv{j}}  \ottsym{\mbox{$\mid$}}   \mathin{ \mathit{j} }{ \mathit{J} }   \ottsym{\}}$ to be an element of
the value predicate at formula $\ottsym{\{}  \tyAv_{\ottmv{i}}  \ottsym{\mbox{$\mid$}}   \mathin{ \mathit{i} }{ \mathit{I} }   \ottsym{\}}$ when there is a mapping $f$
from positions in the set formula to positions in the set value such that
each element $\valMeta_{\ottmv{j}}$ of the set value is in the logical predicate at the least
upper bound of the set of formulas that $f$ maps to $\valMeta_{\ottmv{j}}$.
In other words, $ \mathin{ \valMeta_{\ottmv{j}} }{  \lrres{  \tyBigJoin{  \mathin{ \mathit{i} }{ \mathit{I'} }  }{ \tyAv_{\ottmv{i}} }  }  } $ where
$\mathit{I'} = \mathsetcomp{ \mathin{ \mathit{i} }{ \mathit{I} } }{f(\mathit{i}) = \mathit{j}}$.\sz{Maybe say something about how this lets us ``collect together'' a join of elements in the value based on the formulae?}
The result predicate is used because the
join of a set of value formulae is not necessarily itself a value formula.
To concisely express this definition, in \Fref{fig:logrel} the notation $\mathinv{f}$ refers
to the inverse image of $f$.
It is not hard to check that the definition of the value predicate for sets from
the figure is included in that of equation \eqref{eq:logrel:vset}.
Demonstrating the opposite containment appears intractable thanks to the
nondeterministic nature of our reduction relation.

The refrain ``related functions map related inputs to related outputs,'' is commonly given as
a motto of logical relations.
\iflong
We must apply this philosophy with care in our setting, however.
A naive approach might lead to the seemingly reasonable definition given below.
Unfortunately, this first attempt once again turns out to be too weak.
\begin{equation*}
  \begin{array}{rcll}
     \lrval{  \tyFunOne{ \tyAv }{ \tyAc }  }  &=
    & \mathsetcomp{ \valLam{ \mathit{x} }{ \tmMeta } }{ \mathforall{  \mathin{ \valMeta }{  \lrval{ \tyAv }  }  }{  \mathin{  \subst{ \tmMeta }{ \valMeta }{ \mathit{x} }  }{  \lrexp{ \tyAc }  }  } } &\\
     \lrval{  \tyFun{  \mathin{ \mathit{i} }{ \mathit{I} }  }{ \tyAv_{\ottmv{i}} }{ \tyAc_{\ottmv{i}} }  }  &=
    & \mathsetcomp{\valMeta}{ \mathforall{  \mathin{ \mathit{i} }{ \mathit{I} }  }{  \mathin{ \valMeta }{  \lrval{  \tyFunOne{ \tyAv_{\ottmv{i}} }{ \tyAc_{\ottmv{i}} }  }  }  } }
    & \text{where $\mathit{I}$ is not a singleton}
  \end{array}
\end{equation*}
\fi
\longonly{Instead, given}\shortonly{Given} a function
$ \mathin{  \valLam{ \mathit{x} }{ \tmMeta }  }{  \lrval{  \tyFun{  \mathin{ \mathit{i} }{ \mathit{I} }  }{ \tyAv_{\ottmv{i}} }{ \tyAc_{\ottmv{i}} }  }  } $, the logical
predicate demands that for an input $\valMeta$ which satisfies the input
requirement $\tyAv_{\ottmv{i}}$ for any subset $ \mathit{I'}  \subseteq  \mathit{I} $ of the clauses in the formula,
the function must provide an output which is in the expression predicate
at the least upper bound of the set of output formulae for the triggered clauses.

Logical relations are traditionally defined by induction on a type and
make use of self-reference only through structural recursion.
In contrast, the definition in \Fref{fig:logrel} is not structurally recursive
due to the set and function cases of the value relation.
Nevertheless, it is still well-defined by induction on the size $\ottsym{\mbox{$\mid$}}  \tyAc  \ottsym{\mbox{$\mid$}}$ of the formula
indexing each predicate $\tyAc$.
To make this argument for the set and function predicates, we rely upon
\Lref{lem:log-sem:join-size}.

We now establish some properties of the logical predicate using
basic facts of reduction.
\begin{lemma}[Closure Under Antireduction]
  \label{lem:logsem:antired}

  If $\tmMeta  \stepsym^{*}  \tmMeta'$ and $ \mathin{ \tmMeta' }{  \lrexp{ \tyAc }  } $ then $ \mathin{ \tmMeta }{  \lrexp{ \tyAc }  } $.
\end{lemma}
\iflong
\begin{proof}
  Immediate consequence of the transitivity of the reduction relation.
\end{proof}
\fi
\begin{lemma}[Monadic Unit]
  \label{lem:logsem:mon-unit}
   \longonly{$ \lrres{ \tyAc }  \subseteq  \lrexp{ \tyAc } $ and }
   $ \lrval{ \tyAv }  \subseteq  \lrexp{ \tyAv } $
\end{lemma}
\iflong
\begin{proof}
  Immediate consequence of the reflexivity of the reduction relation.
\end{proof}
\fi
\begin{lemma}[Monadic Bind]
  \label{lem:logsem:mon-bind}
  \longonly{
  \leavevmode
  \begin{enumerate}
  \item If $ \mathin{ \tmMeta }{  \lrexp{ \tyAc }  } $ and for all $ \mathin{ \resMeta }{  \lrres{ \tyAc }  } $ we have
  $ \mathin{ \ottnt{E}  \ottsym{[}  \resMeta  \ottsym{]} }{  \lrexp{ \tyAc' }  } $ then $ \mathin{ \ottnt{E}  \ottsym{[}  \tmMeta  \ottsym{]} }{  \lrexp{ \tyAc' }  } $.

  \item If $ \mathin{ \tmMeta }{  \lrexp{ \tyAv }  } $ and for all $ \mathin{ \valMeta }{  \lrval{ \tyAv }  } $ we have
  $ \mathin{ \ottnt{E}  \ottsym{[}  \valMeta  \ottsym{]} }{  \lrexp{ \tyAc' }  } $ then $ \mathin{ \ottnt{E}  \ottsym{[}  \tmMeta  \ottsym{]} }{  \lrexp{ \tyAc' }  } $.
  \end{enumerate}
  }
  \shortonly{
    If $ \mathin{ \tmMeta }{  \lrexp{ \tyAv }  } $ and for all $ \mathin{ \valMeta }{  \lrval{ \tyAv }  } $ we have
    $ \mathin{ \ottnt{E}  \ottsym{[}  \valMeta  \ottsym{]} }{  \lrexp{ \tyAc' }  } $ then $ \mathin{ \ottnt{E}  \ottsym{[}  \tmMeta  \ottsym{]} }{  \lrexp{ \tyAc' }  } $.
  }
\end{lemma}
\longonly{
\begin{proof}
  For the first part, we have a result $ \mathin{ \resMeta }{  \lrres{ \tyAc }  } $ such that $\tmMeta  \stepsym^{*}  \resMeta$ from
  the definition of the expression predicate.
  It follows that $\ottnt{E}  \ottsym{[}  \tmMeta  \ottsym{]}  \stepsym^{*}  \ottnt{E}  \ottsym{[}  \resMeta  \ottsym{]}$ so by \Lref{lem:logsem:antired} all we need to show
  is $ \mathin{ \ottnt{E}  \ottsym{[}  \resMeta  \ottsym{]} }{  \lrexp{ \tyAc' }  } $.
  This is immediate from our premise.

  To prove the second part, let $ \mathin{ \resMeta }{  \lrres{ \tyAv }  } $.
  By the first part of this lemma, it suffices to show $ \mathin{ \ottnt{E}  \ottsym{[}  \resMeta  \ottsym{]} }{  \lrexp{ \tyAc' }  } $.
  Examining the definition of $ \lrres{ \tyAv } $, we see that either $\resMeta  \ottsym{=}   \tmTopC $
  or $\resMeta$ is a value in $ \lrval{ \tyAv } $.
  In the former case, we have $\ottnt{E}  \ottsym{[}   \tmTopC   \ottsym{]}  \stepsym^{*}   \tmTopC $ and $ \mathin{  \tmTopC  }{  \lrres{ \tyAc' }  } $.
  The latter case is immediate.
\end{proof}
}

\begin{lemma}[Semantic Downward Closure]
  \label{lem:logsem:sem-downclos}
  Suppose $ \tyapx{ \tyAc }{ \tyAc' } $ and $ \tyapx{ \tyAv }{ \tyAv' } $. Then:
  \[ (1)\ \  \lrexp{ \tyAc' }  \subseteq  \lrexp{ \tyAc } 
     \qquad (2)\ \  \lrres{ \tyAc' }  \subseteq  \lrres{ \tyAc } 
     \qquad (3)\ \  \lrval{ \tyAv' }  \subseteq  \lrval{ \tyAv }  \]
\end{lemma}
\begin{proof}
  We prove the three parts simultaneously; the first two are straightforward.
  For the third we proceed by induction on $\max \{\ottsym{\mbox{$\mid$}}  \tyAv  \ottsym{\mbox{$\mid$}}, \ottsym{\mbox{$\mid$}}  \tyAv'  \ottsym{\mbox{$\mid$}}\}$ and
  case analysis on $ \tyapx{ \tyAv }{ \tyAv' } $.
  \longonly{See the proof of \Lref{lem:apdx:logsem:sem-downclos} in the
    appendices for the details.}
\end{proof}

Although we have now proven a semantic analog of Downward Closure, we cannot
do the same for Directedness. However the following lemma about joins is sufficient.
\begin{lemma}[Semantic Join]
  \label{lem:logsem:join}
  \longonly{
    \leavevmode
    \begin{enumerate}
  \item If $ \mathin{ \valMeta }{  \lrval{ \tyAv }  } $ or $ \mathin{ \valMeta' }{  \lrval{ \tyAv }  } $ then $ \mathin{  \synjoin{ \valMeta }{ \valMeta' }  }{  \lrres{  \tyJoin{ \tyAv }{ \tyAv' }  }  } $.
  \item If $ \mathin{ \valMeta }{  \lrval{ \tyAv }  } $ and $ \mathin{ \valMeta' }{  \lrval{ \tyAv' }  } $ then $ \mathin{  \synjoin{ \valMeta }{ \valMeta' }  }{  \lrres{  \tyJoin{ \tyAv }{ \tyAv' }  }  } $.
  \item If $ \mathin{ \tmMeta }{  \lrexp{ \tyAc }  } $ and $ \mathin{ \tmMeta' }{  \lrexp{ \tyAc' }  } $ then $ \mathin{  \tmJoin{ \tmMeta }{ \tmMeta' }  }{  \lrexp{  \tyJoin{ \tyAc }{ \tyAc' }  }  } $.
  \item Suppose there exists $f : \mathit{I} \to \mathit{J}$ such that for all
    $ \mathin{ \mathit{j} }{ \mathit{J} } $ we have $ \mathin{ \tmMeta_{\ottmv{j}} }{  \lrexp{  \tyBigJoin{  \mathin{ \mathit{i} }{  \mathinv{ \mathit{f} }   \ottsym{(}  \mathit{j}  \ottsym{)} }  }{ \tyAc_{\ottmv{i}} }  }  } $.
    Then $ \mathin{  \tmJoinManySet{  \mathin{ \mathit{j} }{ \mathit{J} }  }{ \tmMeta_{\ottmv{j}} }  }{  \lrexp{  \tyBigJoin{  \mathin{ \mathit{i} }{ \mathit{I} }  }{ \tyAc_{\ottmv{i}} }  }  } $.
  \end{enumerate}
  }
  \shortonly{
    Suppose there exists $f : \mathit{I} \to \mathit{J}$ such that for all
    $ \mathin{ \mathit{j} }{ \mathit{J} } $ we have $ \mathin{ \tmMeta_{\ottmv{j}} }{  \lrexp{  \tyBigJoin{  \mathin{ \mathit{i} }{  \mathinv{ \mathit{f} }   \ottsym{(}  \mathit{j}  \ottsym{)} }  }{ \tyAc_{\ottmv{i}} }  }  } $.
    Then $ \mathin{  \tmJoinManySet{  \mathin{ \mathit{j} }{ \mathit{J} }  }{ \tmMeta_{\ottmv{j}} }  }{  \lrexp{  \tyBigJoin{  \mathin{ \mathit{i} }{ \mathit{I} }  }{ \tyAc_{\ottmv{i}} }  }  } $.
  }
\end{lemma}
\longonly{
\begin{proof}
  The proof of the first part proceeds by routine induction on $\tyAv$.
  The second and third parts are proven by simultaneous induction on the maximum
  size of the two types involved.
  The second part uses the first in addition to \Lref{lem:logsem:sem-downclos}; see the proof of \Lref{lem:apdx:logsem:join:val-strong} in the appendices for details.
  The third part is a straightforward application of the second and \Lref{lem:logsem:mon-bind}.
  The final part follows from repeated application of the third; see the proof of \Lref{lem:apdx:logsem:join}.
\end{proof}
}

\subsubsection*{The Fundamental Property}
The last part of \Fref{fig:logrel} lifts the definitions of the logical predicate
from closed terms to open terms.
It gives a predicate $ \lrenv{ \Gamma } $ over closing substitutions and uses it
to define a judgment over open terms written $ \judglr{ \Gamma }{ \tmMeta }{ \tyAc } $ whose intuitive meaning is
``$\tmMeta$ is \emph{semantically assigned} the formula $\tyAc$'' in contrast with the \emph{syntactic assignment} rules of \Fref{fig:typing}.
These definitions allow us to state the Fundamental Property of the logical predicate, namely that
under any environment, every formula syntactically assigned to a term $\tmMeta$ is also
semantically assigned to $\tmMeta$.
\begin{lemma}[Fundamental Property]
  \label{lem:log-sem:fund-prop}
  If $ \judg{ \Gamma }{ \tmMeta }{ \tyAc } $ then $ \judglr{ \Gamma }{ \tmMeta }{ \tyAc } $.
\end{lemma}
\begin{proof}
  We proceed by nested induction first on $\tmMeta$ and then on $ \judg{ \Gamma }{ \tmMeta }{ \tyAc } $.\longonly{ See the proof of \Lref{lem:apdx:log-sem:fund-prop} in the appendices for details.}
\end{proof}
\noindent Using the Fundamental Property, it is not hard to verify
Adequacy~(\Lref{lem:logsem:adequacy}).
\iflong
\begin{proof}
  Suppose $\logapx{\valMeta}{\tmMeta}$.
  We must show $\tmMeta  \stepsym^{*}   \tmTopC $ or $\mathexists{\valMeta'}{\tmMeta  \stepsym^{*}  \valMeta'}$.
  It is immediate from the definition of formula assignment (specifically the rule
  \ottdrulename{TBotV}) that there exists some value formula $\tyAv$ such that
  $ \judg{  \envEmpty  }{ \valMeta }{ \tyAv } $. 
  By assumption, we have $ \judg{  \envEmpty  }{ \tmMeta }{ \tyAv } $ and thus $ \judglr{  \envEmpty  }{ \tmMeta }{ \tyAv } $ thanks
  to the Fundamental Property (\Lref{lem:log-sem:fund-prop}).
  The definition of $ \lrexp{ \tyAv } $ then gives us a result $\resMeta$
  such that $\tmMeta  \stepsym^{*}  \resMeta$ and $ \mathin{ \resMeta }{  \lrres{ \tyAv }  } $.
  Examining the definition of $ \lrres{ \tyAv } $ reveals that $\resMeta$ must
  either be a value or $ \tmTopC $.
\end{proof}
\else
\begin{proof}
  Suppose $\logapx{\valMeta}{\tmMeta}$.
  We must show $\tmMeta  \stepsym^{*}   \tmTopC $ or $\mathexists{\valMeta'}{\tmMeta  \stepsym^{*}  \valMeta'}$.
  From \ottdrulename{TBotV} we have $ \judg{  \envEmpty  }{ \valMeta }{  \valBotV  } $.
  By assumption, we have $ \judg{  \envEmpty  }{ \tmMeta }{  \valBotV  } $ and thus $ \judglr{  \envEmpty  }{ \tmMeta }{  \valBotV  } $ thanks
  to the Fundamental Property (\Lref{lem:log-sem:fund-prop}).
  The definitions of $ \lrexp{  \valBotV  } $ and $ \lrres{  \valBotV  } $ provide $\resMeta$ such that $ \returns{ \tmMeta }{ \resMeta } $.
\end{proof}
\fi
\subsection{Domain Theory}
\label{sec:log-sem:domains}
We now give a brief description of the domain-theoretic view of \gcalc
and the relevance of our filter model.
The goal is not to be completely rigorous but rather to intuitively justify the
allusions to domain theory that we have made throughout the paper.
We assume some knowledge of the topic.
Our terminology and definitions follow those of \citet{cartwright16}; the proofs and much of the approach originate with \citet{scott82}.
\longonly{The full details are given in \Aref{sec:apdx:domains}.}

It is known that filter models lead to Scott-style models in a straightforward
way.
We have seen that a single formula represents a finite behavior of a program.
The sets of formulae $\synvtype$ and $\synctype$ each form a type of
preorder that is known as a \emph{finitary basis}.
That is, formulae correspond to the finite or \emph{compact} elements of a
domain.
The entire domain---including both finite and infinite elements---can be obtained through a construct known as the \emph{ideal completion}.

As mentioned previously, an \emph{ideal} over a preorder $X$ is a non-empty
downward-closed directed subset of $X$.
The set of ideals over a finitary basis $\fbasis{A}$ is written
$\ideals{\fbasis{A}}$ and forms a Scott domain.
The meaning of any term (defined in \Sref{sec:logsem:results}) is an ideal
over computation formulae thanks to Lemmas~\ref{lem:log-sem:non-empty}-\ref{lem:log-sem:directed}
and the domain $\ideals{\synvtype}$ is a solution $D$ of the following
domain equation.\longonly{ (See \Tref{lem:apdx:densem:iso} in the appendices.)}
\begin{equation}
  \label{eqn:logsem:dom-eqn}
  \domval  \cong (\ideals{\symset} + \domval \times \domval + \powdomh{\domval} + (\domval \to \domval_{\bot\top}))_\botv
\end{equation}
In this equation, the notation $D \to D'$ represents the
\emph{continuous function space} over domains.
The operators $D_\bot$, $D_\botv$, and $D_\top$ represent the domain $D$ extended
with a least or greatest element.
The cartesian product of two domains is written $D \times D'$.
The \emph{Hoare powerdomain}~\cite{winskel85} is written $\powdomh{D}$.

\todo{One of the most important points of the paper, but it is buried!}
The domain equation~\eqref{eqn:logsem:dom-eqn} highlights that the meaning of
\gcalc programs is essentially deterministic.
We can see, for example, that a program might \emph{mean} only one of $ \symName{true} $ or $ \symName{false} $.
To understand the isomorphism, the notion of \emph{approximable mapping} is helpful.
\begin{definition}
  An approximable mapping on finitary bases $\fbasis{A}$ and $\fbasis{B}$ is a relation
  $R \subseteq \fbasis{A} \times \fbasis{B}$ such that:
  \begin{itemize}
    \item $\mathforall{a \in \fbasis{A}}{\mathexists{b \in \fbasis{B}}{(a, b) \in R}}$
    \item If $(a, b) \in R$ and $b' \sqsubseteq_\fbasis{B} b$ then $(a, b') \in R$.
    \item If $(a, b) \in R$ and $a \sqsubseteq_\fbasis{A} a'$ then $(a', b) \in R$.
    \item If $(a, b) \in R$ and $(a, b') \in R$ then $(a, b \sqcup b') \in R$.
  \end{itemize}
\end{definition}

The criteria for approximable mappings correspond to lemmas we established
earlier in this section, so the relation $\{(\tyAv, \tyAc) \mid  \judg{  \envEmpty  }{  \valLam{ \mathit{x} }{ \tmMeta }  }{  \tyFunOne{ \tyAv }{ \tyAc }  }  \}$ is an approximable mapping for any function $ \valLam{ \mathit{x} }{ \tmMeta } $.
It turns out that the approximable mappings between two finitary bases are
isomorphic as a partial order to the space of continuous functions over the
corresponding full domains.


One traditionally defines a meaning function by structural
recursion on terms that produces the meaning of a program in a domain.
Here, we conjecture that the equations that usually define the meaning function
can be proven in terms of the notion of meaning arising from the filter model.


\section{Discussion}
\label{sec:disc}

While we have seen that \gcalc has an appealing theoretical foundation, questions remain with respect to its application in practice.
This section discusses future directions designers of languages based on \gcalc might take to
devise practical implementation techniques and
provide the expressivity needed to support programmers of distributed systems.

\subsection{Considerations for Implementation}
\label{sec:disc:impl}



Interacting with a \gcalc program involves running it and then
watching the observations it produces over time.
In general, one should \emph{not} wait for the program to produce a result in its
entirety before taking action in response because the full result may be
infinite.
In \Sref{sec:opsem:nondet} we studied the challenges of defining a semantics
that supports applying strict functions like $\mathit{head}$ to infinite arguments
like $ \tmApp{ \mathit{fromN} }{ \ottsym{0} } $.
This problem motivated the use of nondeterministic approximation steps, which we noted
were a declarative approach to specifying a system providing pipeline parallelism.

\begin{figure}
  {\small
  \begin{center}
    \begin{tabular}{c|c c c c}
      Program & \multicolumn{4}{ c }{Observations} \\
      \hline
      $\tmMeta$ & $\valMeta_{{\mathrm{1}}}$ & $\valMeta_{{\mathrm{2}}}$ & $\valMeta_{{\mathrm{3}}}$ & $\ldots$ \\
      \hline
      $ \subst{ \tmMeta' }{ \valMeta_{{\mathrm{1}}} }{ \mathit{x} } $ & \cellcolor{lightgray} $\resMeta'_{1,1}$ & $\resMeta'_{1,2}$ & $\resMeta'_{1,3}$ & $\ldots$ \\
      $ \subst{ \tmMeta' }{ \valMeta_{{\mathrm{2}}} }{ \mathit{x} } $ & $\resMeta'_{2,1}$ & \cellcolor{lightgray} $\resMeta'_{2,2}$ & $\resMeta'_{2,3}$ & $\ldots$ \\
      $ \subst{ \tmMeta' }{ \valMeta_{{\mathrm{3}}} }{ \mathit{x} } $ & $\resMeta'_{3,1}$ & $\resMeta'_{3,2}$ & \cellcolor{lightgray} $\resMeta'_{3,3}$ & $\ldots$ \\
      $\vdots$ & $\vdots$ & $\vdots$ & $\vdots$ & $\ddots$ \\
      \hline
      $ \tmApp{ \ottsym{(}   \valLam{ \mathit{x} }{ \tmMeta' }   \ottsym{)} }{ \tmMeta } $ & \cellcolor{lightgray} $\resMeta'_{1,1}$ & \cellcolor{lightgray} $\resMeta'_{2,2}$ & \cellcolor{lightgray} $\resMeta'_{3,3}$ & $\ldots$
    \end{tabular}
  \end{center}
  }
  \caption{Interleaved evaluation of $ \tmApp{ \ottsym{(}   \valLam{ \mathit{x} }{ \tmMeta' }   \ottsym{)} }{ \tmMeta } $ simulating pipeline parallelism.}
  \label{fig:interleaved}
\end{figure}

This problem is also the chief concern an implementation must deal with since
approximation steps are not realizable in practice.
It is expected that an implementation interleaves computation of the
output that a function produces with the computation of the input that its
argument provides.
To illustrate this idea, consider the term $ \tmApp{ \ottsym{(}   \valLam{ \mathit{x} }{ \tmMeta' }   \ottsym{)} }{ \tmMeta } $.
How should an interpreter evaluate it?
One approach is to evaluate \gcalc expressions to streams of observations that
improve over time, much as we saw in \Sref{sec:language}.
In this case, suppose we observe the input stream $\valMeta_{{\mathrm{1}}}, \valMeta_{{\mathrm{2}}}, \ldots$
from the evaluation of $\tmMeta$.
For each one of these observations $\valMeta_{\ottmv{i}}$, we may obtain a stream of
observations $\resMeta'_{i,1}, \resMeta'_{i,2}, \ldots$ by evaluating
$ \subst{ \tmMeta' }{ \valMeta_{\ottmv{i}} }{ \mathit{x} } $.
To \emph{fairly} interleave the input and output computations, we take the
diagonal as depicted by \Fref{fig:interleaved}.

A concrete implementation might represent streams of elements of the set $X$ as
functions in the set ${\mathnat \to X}$, which in Haskell can be represented as values of the
monadic type \code{Reader Nat X} where \code{Nat} is the type of natural numbers.
The monadic join operation (which is unrelated to semilattice joins despite the
name) performs diagonalization.\footnote{The monadic
nature of streams appears to be folklore. See, for example,
the blog post of \citet{gibbons10}.}
Its type is given below.
\begin{lstlisting}[language=Haskell]
    join :: Reader Nat (Reader Nat X) → Reader Nat X
\end{lstlisting}
In this way, the details of interleaving can be hidden behind a monadic
abstraction and the definition of an interpreter can remain largely conventional.\footnote{A proof-of-concept implementation in Haskell is available online.~\cite{rioux25impl}}

While this strategy is easy to implement, it is inefficient.
Enumerating the elements of a diagonalized stream is slow.
Moreover, each time we compute any $\resMeta'_{i,j}$, we are recomputing the
output of $\tmMeta'$ from scratch on the input $\valMeta_{\ottmv{i}}$.
This involves much repeated work; it would be desirable to find an
incremental approach to evaluation that does only the work needed to calculate
the change in output for each change in input.
Such an approach might resemble the \emph{seminaive evaluation} of Datalog, which
\citet{arntzenius19} have already adapted to support higher-order functional
programming in Datafun.


Recall the function $\mathit{reaches}$ from \Sref{sec:datalog}.
It streams the correct output for all graphs, but it does not terminate on cyclic
inputs.
On one hand, this is not a problem \emph{semantically} since termination does not
affect the meaning of a \gcalc program.
On the other, it is clearly preferable \emph{in practice} that an implementation terminate when
possible to conserve resources.
Logic programming languages can achieve termination on computations similar to
$\mathit{reaches}$ using a technique called \emph{tabling} that, in the functional
setting, corresponds to \emph{memoization}.
Although memoization has long been used to improve the efficiency of functional
programs, it seems particularly important in realistic implementations of \gcalc in order to obtain
good termination behavior.



\subsection{Monotonicity \& Beyond}
\label{sec:disc:ext}
What should we make of monotonicity?  Is it too stringent of a limitation on the
programmer?
\longonly{One might worry that ruling out non-monotone operations leaves
\gcalc impoverished, with limited ability to express useful computations.}
After all, its set data type does not permit computing differences or complements
and it is impossible to implement a boolean-valued membership function because
monotone operations cannot detect absence from a set.

One important point is that these are \emph{not} limitations as compared with traditional functional programming languages such as Haskell or ML.
Comparing \gcalc sets with ML's Set is comparing apples
and oranges: they are simply different data types with different trade-offs.
While the elements of each type may appear similar, their order structure is not.
In \gcalc, the streaming order describes the set $\ottsym{\{}  \ottsym{1}  \ottsym{\}}$ as an approximation
of $\ottsym{\{}  \ottsym{1}  \ottsym{,}  \ottsym{2}  \ottsym{\}}$.
In ML, they are entirely different values, incomparable in ML's semantic order.
ML's order allows for operations like set membership and difference to be seen as
monotone, at the cost of streaming.
A similar ``discretely ordered set'' could be added to \gcalc  and it would
support the operations of the ML type\longonly{ (just as numeric operations are monotone
with respect to the discrete order)}, but it would not be able to express examples
like those from \Sref{sec:datalog}.
Thus, monotonicity does not make \gcalc less expressive than conventional functional languages.
Indeed, as is clear from the study of denotational semantics, such languages are
also monotone; it is the join operator and the rich order
structure on data types needed for streaming behavior that they lack.

Moreover, the restrictions of \gcalc are largely familiar in the distributed
setting.
The \gcalc set data type generalizes \emph{grow-only set} CRDTs, which do not
support the removal of elements.
Users of set LVars face this limitation as well as the lack of a boolean
membership test.
That said, there is good reason to combine streaming data with
non-monotone updates.
A distributed key value store, for instance, needs to be able to accept arbitrary updates
from clients, not just inflationary ones.



\paragraph{Frozen Values}
The point of monotonicity is that we do not want a program to take any action
that will have to be undone later when more input arrives.
As noted in the introduction, if our input is a set and we produce some output because the
set lacks a particular element, we would have to retract the output if that
element were to arrive later.
However, if at some point the program receives all of the input and knows for
sure that no more elements of the set are headed its way, then, intuitively, there
is no problem with asking whether or not an element is in the set.

Datafun~\cite{arntzenius16} and LVish~\cite{kuper14} both provide mechanisms to
address this situation.
In a similar manner, we propose that a producer of a value be able to
\emph{freeze} it by setting a flag promising the context that no further data
will be produced.
This enables the consumer to safely perform otherwise non-monotone operations.

In \gcalc, we might write $\ottkw{frz} \, \valMeta$ to indicate the frozen value $\valMeta$.
While $\ottsym{\{}  \ottsym{1}  \ottsym{,}  \ottsym{2}  \ottsym{\}}$ represents the knowledge that a set contains the elements
$1$ and $2$, the value $\ottkw{frz} \, \ottsym{\{}  \ottsym{1}  \ottsym{,}  \ottsym{2}  \ottsym{\}}$ additionally contains the knowledge
that all other elements are absent from the set.
Given a value $\valMeta$, we expect to have $ \ctxapx{ \valMeta }{ \ottkw{frz} \, \valMeta } $ since $\valMeta$ may
be frozen in the future.
\iflong
Freezing should also respect equivalence, of course, so
if $ \ctxeq{ \valMeta }{ \valMeta' } $ then we should have $ \ctxeq{ \ottkw{frz} \, \valMeta }{ \ottkw{frz} \, \valMeta' } $.
However, $ \ctxapx{ \valMeta }{ \valMeta' } $ should \emph{not} imply $ \ctxapx{ \ottkw{frz} \, \valMeta }{ \ottkw{frz} \, \valMeta' } $:
we want $\ottkw{frz} \, \ottsym{\{}  \ottsym{1}  \ottsym{\}}$ to be incomparable from $\ottkw{frz} \, \ottsym{\{}  \ottsym{1}  \ottsym{,}  \ottsym{2}  \ottsym{\}}$ just as the
corresponding ML sets are incomparable.
\fi
\shortonly{However, $ \ctxapx{ \valMeta }{ \valMeta' } $ should not imply $ \ctxapx{ \ottkw{frz} \, \valMeta }{ \ottkw{frz} \, \valMeta' } $.}
\iflong
In order to rule out the non-monotone function $ \valLam{ \mathit{x} }{ \ottkw{frz} \, \mathit{x} } $,
we need to prevent unfrozen streaming variables from appearing inside a frozen
value.
This could be accomplished by defining a set of closed \emph{freezable values} or
by imposing a modal type system in the style of \citet{arntzenius16}.
\fi

\paragraph{Versioned Values}
While freezing allows for otherwise non-monotone operations on data that is no
longer changing, programmers of distributed systems occasionally require a way to
model data that changes arbitrarily over time.
This is an old problem in the distributed systems literature with known solutions
like those of Amazon's Dynamo~\cite{dynamo}.
To follow Dynamo's approach, we might add \emph{lexicographic pairs} $ \tmLexPair{ \valMeta_{{\mathrm{1}}} }{ \valMeta_{{\mathrm{2}}} } $ to \gcalc.
A lexicographic ordering allows the programmer to tag a datum $\valMeta_{{\mathrm{2}}}$ with a
version $\valMeta_{{\mathrm{1}}}$.
The datum can change arbitrarily so long as the version increases.
The version is frequently a vector clock.
To ensure monotonicity is preserved, the elimination form will need to take the
form of a monadic bind operator $ \tmMonBind{ \mathit{x} }{ \tmMeta_{{\mathrm{1}}} }{ \tmMeta_{{\mathrm{2}}} } $.
This operator evaluates $\tmMeta_{{\mathrm{1}}}$ to a pair $ \tmLexPair{ \valMeta_{{\mathrm{1}}} }{ \valMeta'_{{\mathrm{1}}} } $, and then evaluates $ \subst{ \tmMeta_{{\mathrm{2}}} }{ \valMeta'_{{\mathrm{1}}} }{ \mathit{x} } $.
This should return a pair $ \tmLexPair{ \valMeta_{{\mathrm{2}}} }{ \valMeta'_{{\mathrm{2}}} } $.
The final result is $ \tmLexPair{  \synjoin{ \valMeta_{{\mathrm{1}}} }{ \valMeta_{{\mathrm{2}}} }  }{ \valMeta'_{{\mathrm{2}}} } $.
Combining lexicographic pairs with \gcalc's sets, one can even model
\emph{multiversioning} in which multiple irreconcilable versions of a piece of
data may exist\longonly{ due to conflicting writes}.


The Bloom programming language~~\cite{alvaro11bloom,conway12} has similar
lattice-based data types for dealing with non-monotone updates in
distributed systems.
These enable the implementation of systems like the Anna key-value
store~\cite{wu18}\longonly{, which provide a wide variety of consistency guarantees}.
It seems that Bloom's data types could be adopted in \gcalc without issue.

\section{Related Work}
\label{sec:related}

\paragraph{Datalog}


The negation-free fragment of Datalog epitomizes
``monotonic-by-construction'' program semantics, and its constraints make it
amenable to very efficient implementations.
\iflong
In our terminology, the relevant streaming order is subset inclusion on sets of
facts, and the programs are fixed points of monotone functions on those sets.

\fi
The declarative nature of Datalog programs, according to \citet{hellerstein21},
``is so natural in the cloud, it almost seems to be crying out for it.''
Backing this up, a venerable line of inquiry due to Hellerstein and his
collaborators~\cite{hellerstein10,loo09,alvaro11dedalus,alvaro11bloom,conway12}
has shown Datalog to be a fruitful basis for describing distributed computations.





However, classic Datalog is characterized by a limited programming
model---it has no higher-order functions, only rudimentary data types, and is
always terminating.
This makes it an inexpressive starting point\longonly{: it cannot implement many
programs that can be written in general-purpose languages}.
One prior attempt to encode functional programming in Datalog handles only the
first-order fragment automatically, resorting to defunctionalization for
higher-order functions~\cite{pacak22}.
Datafun~\cite{arntzenius16} is another language that combines Datalog-style logic
programs with functional programming, to which the present work owes much
inspiration.
Still, it lacks facilities for writing recursive functions.

\paragraph{Infinite Data \& Lazy Functional Programming}
In a lazy functional programming language like Haskell, the input to a function
is evaluated only as needed.
Lazy functions can be more time efficient by avoiding unnecessary computation.
However, they risk sacrificing space efficiency as unevaluated thunks can build
up at runtime.
Moreover, in Haskell, lazy functions support infinite data while strict functions do
not.
Programmers of such languages must balance these tradeoffs between laziness and
strictness.
Ensuring a function is just \emph{lazy enough} can be tricky to get right; this
property is an implicit consequence of how the function is defined and
is not documented by type signatures.
Empirically, many operations on standard Haskell data types are too strict to
support \gcalc-style computation~\cite{breitner23}.

Since all functions in \gcalc are strict, programmers instead must explicitly
make the choice to defer evaluation by wrapping a computation in a thunk.
Thanks to its streaming semantics, \gcalc supports infinite data in a first-class
manner without the need for laziness.
This ability to handle unrestricted infinite values is unusual in a strict functional
programming language, though there is some precedent from ML-family
languages~\cite{jeannin13} that support programming with \emph{regular}
(\ie cyclic) infinite data.

\iflong
From the parallel streaming perspective, one might view lazy programs as
\emph{pull}-based systems: a program produces output only when it is
\emph{demanded} by the context.
Strict programs are \emph{push}-based: a context is ``subscribed'' to receive
updates from a term but a term evaluates independently, producing output on its
own without the need to be requested by a consumer.
\fi

\paragraph{LVars}
The \emph{LVish} Haskell library~\cite{kuper13}, another influence on this
work, provides one way to address the issue of nondeterminism in the presence
of shared state. Values stored in shared mutable reference
cells called LVars are endowed with a semilattice structure and updates
to the cell can only join the old cell contents with a newly written value.
The semilattice properties ensure that after a sequence of LVar writes occurs,
the resulting value in memory is the same, regardless of the order in which the
writes were applied.
LVars are accessed via monotone threshold queries which ensure that
read-write races do not cause nondeterminism.
Since Haskell is typed, LVish can statically rule out sources of errors that
\gcalc cannot.

Our work on \gcalc builds on LVars in a number of ways.
Whereas LVars provide an imperative interface, \gcalc allows for pure and
declarative descriptions of programs that are conducive to equational reasoning.
LVish and \gcalc take different approaches to compositionality: whereas each class
of LVars must be declared as a new type implementing an appropriate
interface, in \gcalc the provided data constructors can be nested ad hoc to build
compound streaming values.

The step from lattice theory to domain theory helps \gcalc support a
variety of streaming data, including infinite and higher-order values.
In contrast, the LVars formalism does not support higher-order shared data:
function application, though monotone, is not expressible as a threshold query.
For a similar reason, it seems unlikely that an analog of \gcalc's set data type
can be implemented as an LVar.

\paragraph{CRDTs}
\emph{Conflict-free replicated data types}~\cite{shapiro11}
(specifically the convergent variety) are one technique that guarantees eventual
consistency.
In this model of programming, distributed replicas share data that is endowed
with semilattice structure that can only increase over time in response to
updates.
Replicas share their local state with each other and use the join operation to
merge their current state with that of others.
\longonly{The commutativity of join implies tolerance to the reordering of
  updates over the network, idempotence protects against duplication, and
  associativity allows multiple updates to be batched together.}
The use of CRDTs alone does not provide the
application-level guarantees one might wish for.
In particular, when the replicas (or a client) take actions based on the values
read from the shared state, the end-to-end behavior of the program might be
not be monotonic and thus not deterministic.
For this reason, \citet{kuper2014joining} and \citet{laddad22} propose restricting the
\textit{uses} of CRDTs to be monotonic.

\paragraph{Functional Reactive Programming} Functional reactive
programming~(FRP)~\cite{elliott97} is a widely known approach to programming with
time-varying values.
It has been applied to distributed
programming~\cite{shibanai18,moriguchi23} among other domains.
Although we discuss values changing over time in describing the streaming behavior
of \gcalc, the similarities end there.
FRP allows values to change in arbitrary ways and makes strong assumptions about
time itself.
We, on the other hand, exploit the simplifying assumptions that values only
increase according to the streaming order over time and that all functions are
monotone.
Importantly, we do not assume time is continuous or even totally ordered.
Thanks to these differing assumptions, we predict an implementation of
\gcalc would not have to deal with issues like glitching that arise in FRP.


\iflong
\paragraph{Fortress} \citet{park13} propose using ``big'' versions of associative
operators to express MapReduce-style computations~\cite{dean08}.
This feature has its origins in the Fortress programming language~\cite{allen08}.
Such a strategy might be useful for automatically distributing \gcalc programs.
\fi

\paragraph{Dunfield's Merge Operator} Perhaps unexpectedly, our join operator was
inspired in part by the \emph{merge operator} of \citet{dunfield14} whose
design led to the study of \emph{disjoint intersection types}~\cite{oliveira16}.
These types have proven useful to encode solutions to the expression
problem~\cite{zhang21} and prevent ambiguous overloading~\cite{rioux23}.
In the future, this line of work may inform the design of a type system for
\gcalc capable of ruling out ambiguity errors.

\paragraph{Dataflow, Stream Processing, and Incremental Computation}
Kahn Process Networks~(KPNs)~\cite{kahn74} are concurrent computations described by
directed graphs in which the edges are streams and nodes are functions over
streams.
Kahn gives a domain-theoretic denotational semantics that captures the idea that
the information order describes how data evolves over time.
In contrast to \gcalc, the only streaming data type in a KPN is the collection of
streams over a fixed data type.
The parallelism offered by KPNs is also limited in expressiveness compared to
\gcalc: they cannot encode parallel or, for example.
KPNs can be embedded in other programming languages, providing some parallel
streaming functionality.
On the other hand, parallelism is a first-class feature in \gcalc that does not
depend on stratification of the language into parallel and functional parts.

Following Kahn, much attention has been given to studying
\emph{dataflow programs}~\cite{akidau15}.
They are commonly used to implement parallel
streaming~\cite{flink,laddad25} and functional reactive~\cite{cooper06}
systems.
Compared to \gcalc, these often lack composable higher-order streaming data types and equational reasoning.

Work on incremental programming~\cite{dbsp,mcsherry13,cutler24} aims
to efficiently compute updates to a program's output in response to changes.
The present work is not concerned with efficiency of implementation; we focus on a rich programming model.


\paragraph{Domain Theory \& Concurrent Lambda Calculi}
Domain-theoretic joins have a storied history as hypothetical functional
programming language features---and in support of parallelism, no less.  In a
seminal paper, \citet{plotkin77} demonstrated that the Scott semantics for PCF is
not fully abstract, essentially because \emph{parallel or} exists in the
model but cannot be defined in the language's syntax.  By adding syntax for it,
full abstraction can be obtained.
We demonstrated how to encode parallel or using the join operator in
\Sref{sec:por} and conjecture that, as a result, our filter model is fully
abstract.
The results of \citet{dezani94} further bolster this belief.

In the wake of Plotkin's result, significant effort was expended
investigating language features inspired by parallel or.
This led to various concurrent lambda calculi~\cite{abramsky93,boudol94,dezani94}
whose semantics have influenced those of \gcalc.
Some of these feature incarnations of the join operator.
Whereas past efforts focus on composing \emph{computations} in parallel, the order
structure present at the \emph{value level} of \gcalc is its characteristic feature.
This structure enables the description of iterative fixed point computation and
reveals the connection between concurrent lambda calculi and Datalog.
To that end, the presence of a set data type whose meaning is a powerdomain and
which supports the big join elimination form is novel; designing the
logical relation in \Sref{sec:logsem:adequacy} to accommodate this feature was
a significant challenge.
\longonly{The goal of determinism in this setting and the connection of domain-theoretic
notions to distributed computing are also novel to the best of our knowledge.}

\longonly{
\section{Conclusion}
\citeauthor{dezani94} state that operators like parallel or
``deserve interest not just because of their ability of
filling the gap between operational and denotational semantics, but also because
they model, though in a very crude way, relevant aspects of computation
strategies in actual implementations.''
On the other hand, a common folk belief today is that operational semantics are
better suited for describing parallel computation than domain theoretic models;
the former approach captures the possible interleavings of
behaviors present in actual implementations.

Our work is witness to yet another perspective.
It may be true that domain theory offers only limited insight into certain approaches to
parallelism that are popular albeit fraught with nondeterminism.
However, the theory \emph{is} a suitable foundation for
deterministic-by-construction parallel programming.
It provides a blueprint for integrating streaming behaviors into functional programming
which work with rich data types and are compatible with call-by-value evaluation.
}


\begin{acks}
  We would like to thank the PLDI\longonly{ 2025} reviewers for extensive feedback
that greatly improved this work.
We are also extremely grateful to Michael Arntzenius, Lindsey Kuper, and Joey
Velez-Ginorio for their comments on early versions of this paper.

This work was supported by the National Science Foundation under grant number
2247088.
Any opinions, findings, and conclusions or recommendations expressed in this
material are those of the authors and do not necessarily reflect the views of the
NSF.

\end{acks}

\bibliographystyle{ACM-Reference-Format}
\bibliography{biblio}

\iflong
\clearpage
\appendix
\section{Proofs: A Filter Model}
\subsection{Formulae}

\begin{lemma}
  \label{lem:apdx:log-sem:join-lub}
  The following rules are admissible:
  \begin{mathpar}
    \inferrule
        { \tyapx{ \tyAc' }{ \tyAc }  \\
          \tyapx{ \tyAc'' }{ \tyAc } }
        { \tyapx{  \tyJoin{ \tyAc' }{ \tyAc'' }  }{ \tyAc } }

     \inferrule
       { \tyapx{ \tyAc }{ \tyAc' } }
       { \tyapx{ \tyAc }{  \tyJoin{ \tyAc' }{ \tyAc'' }  } }

     \inferrule
       { \tyapx{ \tyAc }{ \tyAc'' } }
       { \tyapx{ \tyAc }{  \tyJoin{ \tyAc' }{ \tyAc'' }  } }
  \end{mathpar}
\end{lemma}
\begin{proof}
  Routine induction on $\tyAc'$ for the first two rules and $\tyAc''$ for the third.
\end{proof}

\begin{lemma}[Order Inversion]
  \label{lem:apdx:log-sem:subtyping-inv}
  \leavevmode
  \begin{enumerate}
    \item If $ \tyapx{  \tmTopC  }{ \tyAc } $ then $\tyAc  \ottsym{=}   \tmTopC $.
    \item If $ \tyapx{ \tyAc }{  \tmBotC  } $ then $\tyAc  \ottsym{=}   \tmBotC $.
    \item If $ \tyapx{ \tyAc }{ \tyAv } $ then $\tyAc  \ottsym{=}   \tmBotC $
      or $\tyAc \in \synvtype$.
    \item If $ \tyapx{ \tyAv }{ \tyAc } $ then $\tyAc  \ottsym{=}   \tmTopC $
      or $\tyAc \in \synvtype$.
    \item If $ \tyapx{ \tyAv }{ \ottnt{s} } $ then either $\tyAv  \ottsym{=}   \valBotV $ or
      $\tyAv  \ottsym{=}  \ottnt{s'}$ for some $\ottnt{s'}  \leq  \ottnt{s}$.
    \item If $ \tyapx{ \ottnt{s} }{ \tyAv } $ then $\tyAv  \ottsym{=}  \ottnt{s'}$ where $\ottnt{s}  \leq  \ottnt{s'}$.
    \item If $ \tyapx{ \tyAv }{ \ottsym{(}  \tyAv'_{{\mathrm{1}}}  \ottsym{,}  \tyAv'_{{\mathrm{2}}}  \ottsym{)} } $ then either $\tyAv  \ottsym{=}   \valBotV $ or
      $\tyAv  \ottsym{=}  \ottsym{(}  \tyAv_{{\mathrm{1}}}  \ottsym{,}  \tyAv_{{\mathrm{2}}}  \ottsym{)}$ where we have both $ \tyapx{ \tyAv_{{\mathrm{1}}} }{ \tyAv'_{{\mathrm{1}}} } $ and
      $ \tyapx{ \tyAv_{{\mathrm{2}}} }{ \tyAv'_{{\mathrm{2}}} } $.
    \item If $ \tyapx{ \ottsym{(}  \tyAv_{{\mathrm{1}}}  \ottsym{,}  \tyAv_{{\mathrm{2}}}  \ottsym{)} }{ \tyAv' } $ then $\tyAv'  \ottsym{=}  \ottsym{(}  \tyAv'_{{\mathrm{1}}}  \ottsym{,}  \tyAv'_{{\mathrm{2}}}  \ottsym{)}$ where
      $ \tyapx{ \tyAv_{{\mathrm{1}}} }{ \tyAv'_{{\mathrm{1}}} } $ and $ \tyapx{ \tyAv_{{\mathrm{2}}} }{ \tyAv'_{{\mathrm{2}}} } $.
    \item If $ \tyapx{ \tyAv }{ \ottsym{\{}  \tyAv'_{\ottmv{j}}  \ottsym{\mbox{$\mid$}}   \mathin{ \mathit{j} }{ \mathit{J} }   \ottsym{\}} } $ then either $\tyAv  \ottsym{=}   \valBotV $ or
      $\tyAv  \ottsym{=}  \ottsym{\{}  \tyAv_{\ottmv{i}}  \ottsym{\mbox{$\mid$}}   \mathin{ \mathit{i} }{ \mathit{I} }   \ottsym{\}}$ where we have
      $ \mathforall{  \mathin{ \mathit{i} }{ \mathit{I} }  }{  \mathexists{  \mathin{ \mathit{j} }{ \mathit{J} }  }{  \tyapx{ \tyAv_{\ottmv{i}} }{ \tyAv'_{\ottmv{j}} }  }  } $.
    \item If $ \tyapx{ \ottsym{\{}  \tyAv_{\ottmv{i}}  \ottsym{\mbox{$\mid$}}   \mathin{ \mathit{i} }{ \mathit{I} }   \ottsym{\}} }{ \tyAv' } $ then $\tyAv'  \ottsym{=}  \ottsym{\{}  \tyAv'_{\ottmv{j}}  \ottsym{\mbox{$\mid$}}   \mathin{ \mathit{j} }{ \mathit{J} }   \ottsym{\}}$
      where $ \mathforall{  \mathin{ \mathit{i} }{ \mathit{I} }  }{  \mathexists{  \mathin{ \mathit{j} }{ \mathit{J} }  }{  \tyapx{ \tyAv_{\ottmv{i}} }{ \tyAv'_{\ottmv{j}} }  }  } $.
    \item If $ \tyapx{ \tyAv }{  \tyFun{  \mathin{ \mathit{j} }{ \mathit{J} }  }{ \tyAv'_{\ottmv{j}} }{ \tyAc'_{\ottmv{j}} }  } $ then either $\tyAv  \ottsym{=}   \valBotV $ or
      $\tyAv  \ottsym{=}   \tyFun{  \mathin{ \mathit{i} }{ \mathit{I} }  }{ \tyAv_{\ottmv{i}} }{ \tyAc_{\ottmv{i}} } $
      where \[   \mathforall{  \mathin{ \mathit{i} }{ \mathit{I} }  }{  \mathexists{  \mathit{J'}  \subseteq  \mathit{J}  }{  \tyapx{  \tyBigJoin{  \mathin{ \mathit{j} }{ \mathit{J'} }  }{ \tyAv'_{\ottmv{j}} }  }{  \tyBigJoin{  \mathin{ \mathit{i} }{ \mathit{I} }  }{ \tyAv_{\ottmv{i}} }  }  }  }   \textand   \tyapx{  \tyBigJoin{  \mathin{ \mathit{i} }{ \mathit{I} }  }{ \tyAc_{\ottmv{i}} }  }{  \tyBigJoin{  \mathin{ \mathit{j} }{ \mathit{J'} }  }{ \tyAc'_{\ottmv{j}} }  }   \]
    \item If $ \tyapx{  \tyFun{  \mathin{ \mathit{i} }{ \mathit{I} }  }{ \tyAv_{\ottmv{i}} }{ \tyAc_{\ottmv{i}} }  }{ \tyAv' } $ then $\tyAv'  \ottsym{=}   \tyFun{  \mathin{ \mathit{j} }{ \mathit{J} }  }{ \tyAv'_{\ottmv{j}} }{ \tyAc'_{\ottmv{j}} } $
      where \[   \mathforall{  \mathin{ \mathit{i} }{ \mathit{I} }  }{  \mathexists{  \mathit{J'}  \subseteq  \mathit{J}  }{  \tyapx{  \tyBigJoin{  \mathin{ \mathit{j} }{ \mathit{J'} }  }{ \tyAv'_{\ottmv{j}} }  }{  \tyBigJoin{  \mathin{ \mathit{i} }{ \mathit{I} }  }{ \tyAv_{\ottmv{i}} }  }  }  }   \textand   \tyapx{  \tyBigJoin{  \mathin{ \mathit{i} }{ \mathit{I} }  }{ \tyAc_{\ottmv{i}} }  }{  \tyBigJoin{  \mathin{ \mathit{j} }{ \mathit{J'} }  }{ \tyAc'_{\ottmv{j}} }  }   \]
  \end{enumerate}
\end{lemma}
\begin{proof}
  Each part is a straightforward case analysis on the rules defining the order
  on formulae.
\end{proof}

\begin{lemma}
  \label{lem:apdx:log-sem:formula-op-mono}
  The following rules are admissible:
  \begin{mathpar}
      \inferrule
          { \tyapx{ \tyAc_{{\mathrm{1}}} }{ \tyAc'_{{\mathrm{1}}} }  \\
             \tyapx{ \tyAc_{{\mathrm{2}}} }{ \tyAc'_{{\mathrm{2}}} } }
          { \tyapx{  \tyPairOp{ \tyAc_{{\mathrm{1}}} }{ \tyAc_{{\mathrm{2}}} }  }{  \tyPairOp{ \tyAc'_{{\mathrm{1}}} }{ \tyAc'_{{\mathrm{2}}} }  } }

      \inferrule
          { \tyapx{ \tyAc }{ \tyAc' } }
          { \tyapx{  \tySng{ \tyAc }  }{  \tySng{ \tyAc' }  } }
    \end{mathpar}
\end{lemma}
\begin{proof}
  Both rules can be seen admissible by straightforward
  case analysis on formulae using \Lref{lem:apdx:log-sem:subtyping-inv}.
\end{proof}

\begin{lemma}[Size of Joins]
  \label{lem:apdx:log-sem:join-size}
  \leavevmode
  \begin{enumerate}
  \item $\ottsym{\mbox{$\mid$}}   \tyPairOp{ \tyAc }{ \tyAc' }   \ottsym{\mbox{$\mid$}} \leq \max \{\ottsym{\mbox{$\mid$}}  \tyAc  \ottsym{\mbox{$\mid$}}, \ottsym{\mbox{$\mid$}}  \tyAc'  \ottsym{\mbox{$\mid$}}\} + 1$
  \item $\ottsym{\mbox{$\mid$}}   \tyJoin{ \tyAc }{ \tyAc' }   \ottsym{\mbox{$\mid$}} \leq \max\{\ottsym{\mbox{$\mid$}}  \tyAc  \ottsym{\mbox{$\mid$}}, \ottsym{\mbox{$\mid$}}  \tyAc'  \ottsym{\mbox{$\mid$}}\}$
  \end{enumerate}
\end{lemma}
\begin{proof}
  The first part can be verified by a routine case analysis on $\tyAc$ and
  $\tyAc'$.
  The second part proceeds by induction on
  $\max \{ \ottsym{\mbox{$\mid$}}  \tyAc  \ottsym{\mbox{$\mid$}}, \ottsym{\mbox{$\mid$}}  \tyAc'  \ottsym{\mbox{$\mid$}} \}$ with a straightforward case analysis on
  $\tyAc$ and $\tyAc'$.
  The pair case uses the first part.
\end{proof}

\begin{lemma}[Reflexivity]
  \label{lem:apdx:log-sem:subtyping-refl}
  For all $\tyAc$, we have $ \tyapx{ \tyAc }{ \tyAc } $.
\end{lemma}
\begin{proof}
  Routine induction on $\tyAc$.
\end{proof}

\begin{lemma}[Transitivity]
  If $ \tyapx{ \tyAc }{ \tyAc' } $ and $ \tyapx{ \tyAc' }{ \tyAc'' } $ then $ \tyapx{ \tyAc }{ \tyAc'' } $.
\end{lemma}
\begin{proof}
  Induction on $\ottsym{\mbox{$\mid$}}  \tyAc'  \ottsym{\mbox{$\mid$}}$, performing case analysis on $\tyAc'$. In each
  case, we invert both premises with \Lref{lem:apdx:log-sem:subtyping-inv}.
  Due to the use of the join operator in the definition of the streaming order,
  we use \Lref{lem:apdx:log-sem:join-size} and \Lref{lem:apdx:log-sem:join-lub} in the
  function case.
\end{proof}

\begin{lemma}
  \label{lem:apdx:log-sem:fun-join-dist}
  $ \tyapx{  \tyFunOne{ \tyAv }{ \ottsym{(}   \tyJoin{ \tyAc }{ \tyAc' }   \ottsym{)} }  }{  \tyJoinSyn{ \ottsym{(}   \tyFunOne{ \tyAv }{ \tyAc }   \ottsym{)} }{ \ottsym{(}   \tyFunOne{ \tyAv }{ \tyAc' }   \ottsym{)} }  } $
\end{lemma}
\begin{proof}
  Straightforward application of \rulename{TApxFun},
  \Lref{lem:apdx:log-sem:join-lub},
  and \Lref{lem:apdx:log-sem:subtyping-refl}.
\end{proof}

\subsection{Formula Assignment}

\begin{lemma}[Expression Formula Assignment Inversion]
  Suppose $ \judg{ \Gamma }{ \tmMeta }{ \tyAc } $. Then either $\tyAc  \ottsym{=}   \tmBotC $ or all of the following hold.
  \begin{enumerate}
    \item If $\tmMeta  \ottsym{=}   \tmTopC $ then $\tyAc  \ottsym{=}   \tmTopC $.
    \item If $\tmMeta  \ottsym{=}   \tmJoin{ \tmMeta_{{\mathrm{1}}} }{ \tmMeta_{{\mathrm{2}}} } $ then $ \tyapx{ \tyAc }{  \tyJoin{ \tyAc_{{\mathrm{1}}} }{ \tyAc_{{\mathrm{2}}} }  } $ where
      $ \judg{ \Gamma }{ \tmMeta_{{\mathrm{1}}} }{ \tyAc_{{\mathrm{1}}} } $ and $ \judg{ \Gamma }{ \tmMeta_{{\mathrm{2}}} }{ \tyAc_{{\mathrm{2}}} } $.
    \item If $\tmMeta  \ottsym{=}   \tmPair{ \tmMeta_{{\mathrm{1}}} }{ \tmMeta_{{\mathrm{2}}} } $ then $ \judg{ \Gamma }{ \tmMeta_{{\mathrm{1}}} }{ \tyAc_{{\mathrm{1}}} } $ and
      $ \judg{ \Gamma }{ \tmMeta_{{\mathrm{2}}} }{ \tyAc_{{\mathrm{2}}} } $ and $ \tyapx{ \tyAc }{  \tyPairOp{ \tyAc_{{\mathrm{1}}} }{ \tyAc_{{\mathrm{2}}} }  } $.
    \item If $\tmMeta  \ottsym{=}  \ottsym{\{}  \tmMeta_{\ottmv{i}}  \ottsym{\mbox{$\mid$}}   \mathin{ \mathit{i} }{ \mathit{I} }   \ottsym{\}}$ then $\tyAc  \ottsym{=}   \tyJoin{ \ottsym{\{\}} }{  \tyBigJoin{  \mathin{ \mathit{i} }{ \mathit{I} }  }{  \tySng{ \tyAc_{\ottmv{i}} }  }  } $
      where $ \mathforall{  \mathin{ \mathit{i} }{ \mathit{I} }  }{  \judg{ \Gamma }{ \tmMeta_{\ottmv{i}} }{ \tyAc_{\ottmv{i}} }  } $.
    \item If $\tmMeta  \ottsym{=}  \ottkw{let} \, \ottnt{s}  \ottsym{=}  \tmMeta_{{\mathrm{1}}} \, \ottkw{in} \, \tmMeta_{{\mathrm{2}}}$ then either
      \begin{enumerate}
      \item $ \judg{ \Gamma }{ \tmMeta_{{\mathrm{1}}} }{ \ottnt{s} } $ and $ \judg{ \Gamma }{ \tmMeta_{{\mathrm{2}}} }{ \tyAc } $, or
      \item $\tyAc  \ottsym{=}   \tmTopC $ and $ \judg{ \Gamma }{ \tmMeta_{{\mathrm{1}}} }{  \tmTopC  } $.
      \end{enumerate}
    \item If $\tmMeta  \ottsym{=}  \ottkw{let} \, \ottsym{(}  \mathit{x_{{\mathrm{1}}}}  \ottsym{,}  \mathit{x_{{\mathrm{2}}}}  \ottsym{)}  \ottsym{=}  \tmMeta' \, \ottkw{in} \, \tmMeta''$ then either
      \begin{enumerate}
      \item $ \judg{ \Gamma }{ \tmMeta' }{ \ottsym{(}  \tyAv_{{\mathrm{1}}}  \ottsym{,}  \tyAv_{{\mathrm{2}}}  \ottsym{)} } $ and
        $ \judg{ \Gamma  \ottsym{,}  \mathit{x_{{\mathrm{1}}}}  \ottsym{:}  \tyAv_{{\mathrm{1}}}  \ottsym{,}  \mathit{x_{{\mathrm{2}}}}  \ottsym{:}  \tyAv_{{\mathrm{2}}} }{ \tmMeta'' }{ \tyAc } $, or
      \item $\tyAc  \ottsym{=}   \tmTopC $ and $ \judg{ \Gamma }{ \tmMeta' }{  \tmTopC  } $.
      \end{enumerate}
    \item If $\tmMeta  \ottsym{=}   \tmBigJoin{ \mathit{x} }{ \tmMeta_{{\mathrm{1}}} }{ \tmMeta_{{\mathrm{2}}} } $ then either
      \begin{enumerate}
      \item $\tyAc  \ottsym{=}   \tyBigJoin{  \mathin{ \mathit{i} }{ \mathit{I} }  }{ \tyAc_{\ottmv{i}} } $ where $ \judg{ \Gamma }{ \tmMeta_{{\mathrm{1}}} }{ \ottsym{\{}  \tyAv_{\ottmv{i}}  \ottsym{\mbox{$\mid$}}   \mathin{ \mathit{i} }{ \mathit{I} }   \ottsym{\}} } $
        and $ \mathforall{  \mathin{ \mathit{i} }{ \mathit{I} }  }{  \judg{ \Gamma  \ottsym{,}  \mathit{x}  \ottsym{:}  \tyAv_{\ottmv{i}} }{ \tmMeta_{{\mathrm{2}}} }{ \tyAc_{\ottmv{i}} }  } $, or
      \item $\tyAc  \ottsym{=}   \tmTopC $ and $ \judg{ \Gamma }{ \tmMeta_{{\mathrm{1}}} }{  \tmTopC  } $.
      \end{enumerate}
    \item If $\tmMeta  \ottsym{=}   \tmApp{ \tmMeta_{{\mathrm{1}}} }{ \tmMeta_{{\mathrm{2}}} } $ then either
      \begin{enumerate}
      \item $ \judg{ \Gamma }{ \tmMeta_{{\mathrm{1}}} }{  \tyFunOne{ \tyAv }{ \tyAc }  } $ and $ \judg{ \Gamma }{ \tmMeta_{{\mathrm{2}}} }{ \tyAv } $,
      \item $\tyAc  \ottsym{=}   \tmTopC $ and $ \judg{ \Gamma }{ \tmMeta_{{\mathrm{1}}} }{  \tmTopC  } $, or
      \item $\tyAc  \ottsym{=}   \tmTopC $ and $ \judg{ \Gamma }{ \tmMeta_{{\mathrm{1}}} }{ \tyAv } $ and $ \judg{ \Gamma }{ \tmMeta_{{\mathrm{2}}} }{  \tmTopC  } $.
      \end{enumerate}
  \end{enumerate}
\end{lemma}
\begin{proof}
  Straightforward case analysis on the formula assignment rules.
\end{proof}

\begin{lemma}[Result Formula Assignment Inversion]
  If $ \judg{ \Gamma }{ \resMeta }{ \tyAc } $ then either
  \begin{enumerate}
  \item $\tyAc  \ottsym{=}   \tmBotC $,
  \item $\resMeta  \ottsym{=}   \tmTopC $, or
  \item $\resMeta \in \synres$ and $\tyAc \in \synvtype$.
  \end{enumerate}
\end{lemma}
\begin{proof}
  Straightforward case analysis on the formula assignment rules.
\end{proof}

\begin{lemma}[Value Formula Assignment Inversion]
  If $ \judg{ \Gamma }{ \valMeta }{ \tyAv } $ then either $\tyAv  \ottsym{=}   \valBotV $ or all of the following hold.
  \begin{enumerate}
  \item If $\valMeta  \ottsym{=}  \ottnt{s}$ then $\tyAv  \ottsym{=}  \ottnt{s'}$ and $\ottnt{s'}  \leq  \ottnt{s}$.
  \item If $\valMeta  \ottsym{=}   \tmPair{ \valMeta_{{\mathrm{1}}} }{ \valMeta_{{\mathrm{2}}} } $ then $\tyAv  \ottsym{=}  \ottsym{(}  \tyAv_{{\mathrm{1}}}  \ottsym{,}  \tyAv_{{\mathrm{2}}}  \ottsym{)}$ where $ \judg{ \Gamma }{ \valMeta_{{\mathrm{1}}} }{ \tyAv_{{\mathrm{1}}} } $ and
    $ \judg{ \Gamma }{ \valMeta_{{\mathrm{2}}} }{ \tyAv_{{\mathrm{2}}} } $.
  \item If $\valMeta  \ottsym{=}  \ottsym{\{}  \valMeta_{\ottmv{j}}  \ottsym{\mbox{$\mid$}}   \mathin{ \mathit{j} }{ \mathit{J} }   \ottsym{\}}$ then $\tyAv  \ottsym{=}  \ottsym{\{}  \tyAv_{\ottmv{i}}  \ottsym{\mbox{$\mid$}}   \mathin{ \mathit{i} }{ \mathit{I} }   \ottsym{\}}$ and
    $ \mathforall{  \mathin{ \mathit{i} }{ \mathit{I} }  }{  \mathexists{  \mathin{ \mathit{j} }{ \mathit{J} }  }{  \judg{ \Gamma }{ \valMeta_{\ottmv{j}} }{ \tyAv_{\ottmv{i}} }  }  } $.
  \item If $\valMeta  \ottsym{=}   \valLam{ \mathit{x} }{ \tmMeta } $ then $\tyAv  \ottsym{=}   \tyFun{  \mathin{ \mathit{i} }{ \mathit{I} }  }{ \tyAv_{\ottmv{i}} }{ \tyAc_{\ottmv{i}} } $ and
    $ \mathforall{  \mathin{ \mathit{i} }{ \mathit{I} }  }{  \judg{ \Gamma  \ottsym{,}  \mathit{x}  \ottsym{:}  \tyAv_{\ottmv{i}} }{ \tmMeta }{ \tyAc_{\ottmv{i}} }  } $.
  \item If $\valMeta  \ottsym{=}  \mathit{x}$ then $ \tyapx{ \tyAv }{ \Gamma  \ottsym{(}  \mathit{x}  \ottsym{)} } $.
  \end{enumerate}
\end{lemma}
\begin{proof}
  Straightforward case analysis on the formula assignment rules.
\end{proof}

\begin{lemma}[Compositionality]
  \label{lem:apdx:log-sem:compos}

  If $ \judg{ \Gamma' }{ \ottnt{C}  \ottsym{[}  \tmMeta_{{\mathrm{1}}}  \ottsym{]} }{ \tyAc' } $ and for all $\Gamma$ and $\tyAc$ such that
  $ \judg{ \Gamma }{ \tmMeta_{{\mathrm{1}}} }{ \tyAc } $ we have $ \judg{ \Gamma }{ \tmMeta_{{\mathrm{2}}} }{ \tyAc } $ then
  $ \judg{ \Gamma' }{ \ottnt{C}  \ottsym{[}  \tmMeta_{{\mathrm{2}}}  \ottsym{]} }{ \tyAc' } $.
\end{lemma}
\begin{proof}
  We proceed by structural induction on $\ottnt{C}$.
  In each case, we invert the derivation of $ \judg{ \Gamma' }{ \ottnt{C}  \ottsym{[}  \tmMeta_{{\mathrm{1}}}  \ottsym{]} }{ \tyAc' } $.
  Applying the induction hypothesis then allows us to construct a derivation of
  $ \judg{ \Gamma' }{ \ottnt{C}  \ottsym{[}  \tmMeta_{{\mathrm{2}}}  \ottsym{]} }{ \tyAc' } $.
\end{proof}

\begin{lemma}[Weakening]
  \label{lem:apdx:log-sem:weaken}
  If $ \judg{ \Gamma' }{ \tmMeta }{ \tyAc } $ and $ \envapx{ \Gamma' }{ \Gamma } $ then $ \judg{ \Gamma }{ \tmMeta }{ \tyAc } $
\end{lemma}
\begin{proof}
  Routine induction on $ \judg{ \Gamma' }{ \tmMeta }{ \tyAc } $.
\end{proof}

\begin{lemma}[Totality]
  \label{lem:apdx:log-sem:non-empty}

  For every $\Gamma$ and $\tmMeta$ there exists a formula $\tyAc$ such that $ \judg{ \Gamma }{ \tmMeta }{ \tyAc } $.
\end{lemma}
\begin{proof}
  Immediate from \rulename{TBot}.
\end{proof}

\begin{lemma}[Downward Closure]
  \label{lem:apdx:log-sem:down-clos}

  If $ \judg{ \Gamma }{ \tmMeta }{ \tyAc' } $ and $ \tyapx{ \tyAc }{ \tyAc' } $ then $ \judg{ \Gamma }{ \tmMeta }{ \tyAc } $.
\end{lemma}
\begin{proof}
  Immediate from \rulename{TSub}.
\end{proof}

\begin{lemma}[Directedness]
  \label{lem:apdx:log-sem:directed}
  If $ \judg{ \Gamma }{ \tmMeta }{ \tyAc } $ and $ \judg{ \Gamma }{ \tmMeta }{ \tyAc' } $ then $ \judg{ \Gamma }{ \tmMeta }{  \tyJoin{ \tyAc }{ \tyAc' }  } $.
\end{lemma}
\begin{proof}
  Induction on $\tmMeta$, inverting both premises and making use of
  Lemmas \ref{lem:apdx:log-sem:fun-join-dist}, \ref{lem:apdx:log-sem:join-lub}, and \ref{lem:apdx:log-sem:weaken}.
\end{proof}

\begin{lemma}
  If $ \judg{ \Gamma  \ottsym{,}  \mathit{x}  \ottsym{:}  \tyAv }{ \tmMeta }{ \tyAc } $ and $ \judg{ \Gamma  \ottsym{,}  \mathit{x}  \ottsym{:}  \tyAv' }{ \tmMeta }{ \tyAc' } $ and
  $\tyBv =  \tyJoin{ \tyAv }{ \tyAv' } $ is a value formula then $ \judg{ \Gamma  \ottsym{,}  \mathit{x}  \ottsym{:}  \tyBv }{ \tmMeta }{  \tyJoin{ \tyAc }{ \tyAc' }  } $.
\end{lemma}
\begin{proof}
  Follows directly from Lemmas~\ref{lem:apdx:log-sem:weaken}~and~\ref{lem:apdx:log-sem:directed}.
\end{proof}

\subsection{Subject Expansion}
\begin{lemma}[Inversion of Substitution Typing]
  \label{lem:apdx:log-sem:subst-inv}

  If $ \judg{ \Gamma }{ \substmeta  \ottsym{(}  \tmMeta  \ottsym{)} }{ \tyAc } $ then there exists $\Gamma'$ such that
  $ \judg{ \Gamma }{ \substmeta }{ \Gamma' } $ and $ \judg{ \Gamma  \ottsym{,}  \Gamma' }{ \tmMeta }{ \tyAc } $.
\end{lemma}
\begin{proof}
  First, note that \Lref{lem:apdx:log-sem:directed} lifts to the typing of substitutions.
  That is, if we have environments $\Gamma_{{\mathrm{1}}}$ and $\Gamma_{{\mathrm{2}}}$ such that
  $ \judg{ \Gamma }{ \substmeta }{ \Gamma_{{\mathrm{1}}} } $ and $ \judg{ \Gamma }{ \substmeta }{ \Gamma_{{\mathrm{2}}} } $ then $\Gamma' =  \envJoin{ \Gamma_{{\mathrm{1}}} }{ \Gamma_{{\mathrm{2}}} } $ exists and
  $ \judg{ \Gamma }{ \substmeta }{ \Gamma' } $.
  With this in mind, we proceed by induction on $\tmMeta$, in each case inverting
  its derivation and making use of weakening and directedness.
\end{proof}

\begin{lemma}
  If $ \synjoin{ \resMeta_{{\mathrm{1}}} }{ \resMeta_{{\mathrm{2}}} }   \ottsym{=}   \tmTopC $ then there exists $\tyAc_{{\mathrm{1}}}$ and $\tyAc_{{\mathrm{2}}}$ such that
  $ \tyJoin{ \tyAc_{{\mathrm{1}}} }{ \tyAc_{{\mathrm{2}}} }  =  \tmTopC $ and $ \judg{  \envEmpty  }{ \resMeta_{{\mathrm{1}}} }{ \tyAc_{{\mathrm{1}}} } $ and $ \judg{  \envEmpty  }{ \resMeta_{{\mathrm{2}}} }{ \tyAc_{{\mathrm{2}}} } $.
\end{lemma}
\begin{proof}
  Straightforward induction on $\resMeta_{{\mathrm{1}}}$ and an inner case analysis on $\resMeta_{{\mathrm{2}}}$.
\end{proof}

\begin{lemma}[Inversion of Join Typing]
  \label{lem:apdx:log-sem:join-typ-inv}
  If $ \judg{ \Gamma }{  \synjoin{ \resMeta_{{\mathrm{1}}} }{ \resMeta_{{\mathrm{2}}} }  }{ \tyAc } $ then there exists $\tyAc_{{\mathrm{1}}}$ and $\tyAc_{{\mathrm{2}}}$ such
  that $ \judg{ \Gamma }{ \resMeta_{{\mathrm{1}}} }{ \tyAc_{{\mathrm{1}}} } $ and $ \judg{ \Gamma }{ \resMeta_{{\mathrm{2}}} }{ \tyAc_{{\mathrm{2}}} } $ and $ \tyapx{ \tyAc }{  \tyJoin{ \tyAc_{{\mathrm{1}}} }{ \tyAc_{{\mathrm{2}}} }  } $.
\end{lemma}
\begin{proof}
  Induction on $\resMeta_{{\mathrm{1}}}$ and nested case analysis on $\resMeta_{{\mathrm{2}}}$.
  Uses \Lref{lem:apdx:log-sem:fun-join-dist} and \Lref{lem:apdx:log-sem:formula-op-mono}.
\end{proof}

\begin{lemma}
  \label{lem:apdx:log-sem:ectx-top}
  For all evaluation contexts $\ottnt{E}$ we have $ \judg{ \Gamma }{ \ottnt{E}  \ottsym{[}   \tmTopC   \ottsym{]} }{  \tmTopC  } $.
\end{lemma}
\begin{proof}
  Routine induction on $\ottnt{E}$.
\end{proof}

\begin{lemma}[Subject Expansion]
  \label{lem:apdx:log-sem:subj-exp}
  If $ \step{ \tmMeta }{ \tmMeta' } $ and $ \judg{ \Gamma }{ \tmMeta' }{ \tyAc } $ then $ \judg{ \Gamma }{ \tmMeta }{ \tyAc } $.
\end{lemma}
\begin{proof}
  Induction on $ \step{ \tmMeta }{ \tmMeta' } $, inverting the derivation of $\tmMeta'$ in each case.
  The case in which $ \step{ \ottnt{E}  \ottsym{[}   \tmTopC   \ottsym{]} }{  \tmTopC  } $ follows from \Lref{lem:apdx:log-sem:ectx-top}.
  Beta reduction cases make use of \Lref{lem:apdx:log-sem:subst-inv}.
  The case where $ \step{  \tmJoin{ \resMeta_{{\mathrm{1}}} }{ \resMeta_{{\mathrm{2}}} }  }{  \synjoin{ \resMeta_{{\mathrm{1}}} }{ \resMeta_{{\mathrm{2}}} }  } $ follows from
  \Lref{lem:apdx:log-sem:join-typ-inv}.
\end{proof}

\subsection{Adequacy}
\begin{lemma}[Closure Under Antireduction]
  \label{lem:apdx:logsem:antired}

  If $\tmMeta  \stepsym^{*}  \tmMeta'$ and $ \mathin{ \tmMeta' }{  \lrexp{ \tyAc }  } $ then $ \mathin{ \tmMeta }{  \lrexp{ \tyAc }  } $.
\end{lemma}
\begin{proof}
  Immediate consequence of the transitivity of the reduction relation.
\end{proof}
\begin{lemma}[Monadic Unit]
  \label{lem:apdx:logsem:mon-unit}
   \longonly{$ \lrres{ \tyAc }  \subseteq  \lrexp{ \tyAc } $ and }
   $ \lrval{ \tyAv }  \subseteq  \lrexp{ \tyAv } $
\end{lemma}
\begin{proof}
  Immediate consequence of the reflexivity of the reduction relation.
\end{proof}
\begin{lemma}[Monadic Bind]
  \label{lem:apdx:logsem:mon-bind}
  \leavevmode
  \begin{enumerate}
  \item If $ \mathin{ \tmMeta }{  \lrexp{ \tyAc }  } $ and for all $ \mathin{ \resMeta }{  \lrres{ \tyAc }  } $ we have
  $ \mathin{ \ottnt{E}  \ottsym{[}  \resMeta  \ottsym{]} }{  \lrexp{ \tyAc' }  } $ then $ \mathin{ \ottnt{E}  \ottsym{[}  \tmMeta  \ottsym{]} }{  \lrexp{ \tyAc' }  } $.

  \item If $ \mathin{ \tmMeta }{  \lrexp{ \tyAv }  } $ and for all $ \mathin{ \valMeta }{  \lrval{ \tyAv }  } $ we have
  $ \mathin{ \ottnt{E}  \ottsym{[}  \valMeta  \ottsym{]} }{  \lrexp{ \tyAc' }  } $ then $ \mathin{ \ottnt{E}  \ottsym{[}  \tmMeta  \ottsym{]} }{  \lrexp{ \tyAc' }  } $.
  \end{enumerate}
\end{lemma}
\begin{proof}
  For the first part, we have a result $ \mathin{ \resMeta }{  \lrres{ \tyAc }  } $ such that $\tmMeta  \stepsym^{*}  \resMeta$ from
  the definition of the expression predicate.
  It follows that $\ottnt{E}  \ottsym{[}  \tmMeta  \ottsym{]}  \stepsym^{*}  \ottnt{E}  \ottsym{[}  \resMeta  \ottsym{]}$ so by \Lref{lem:apdx:logsem:antired} all we need to show
  is $ \mathin{ \ottnt{E}  \ottsym{[}  \resMeta  \ottsym{]} }{  \lrexp{ \tyAc' }  } $.
  This is immediate from our premise.

  To prove the second part, let $ \mathin{ \resMeta }{  \lrres{ \tyAv }  } $.
  By the first part of this lemma, it suffices to show $ \mathin{ \ottnt{E}  \ottsym{[}  \resMeta  \ottsym{]} }{  \lrexp{ \tyAc' }  } $.
  Examining the definition of $ \lrres{ \tyAv } $, we see that either $\resMeta  \ottsym{=}   \tmTopC $
  or $\resMeta$ is a value in $ \lrval{ \tyAv } $.
  In the former case, we have $\ottnt{E}  \ottsym{[}   \tmTopC   \ottsym{]}  \stepsym^{*}   \tmTopC $ and $ \mathin{  \tmTopC  }{  \lrres{ \tyAc' }  } $.
  The latter case is a consequence of our premise.
\end{proof}

\begin{lemma}[Semantic Downward Closure]
  \label{lem:apdx:logsem:sem-downclos}
  Suppose $ \tyapx{ \tyAc }{ \tyAc' } $ and $ \tyapx{ \tyAv }{ \tyAv' } $. Then:
  \begin{enumerate}
  \item $ \lrexp{ \tyAc' }  \subseteq  \lrexp{ \tyAc } $
  \item $ \lrres{ \tyAc' }  \subseteq  \lrres{ \tyAc } $
  \item $ \lrval{ \tyAv' }  \subseteq  \lrval{ \tyAv } $
  \end{enumerate}
\end{lemma}
\begin{proof}
  We prove the three parts simultaneously; the first two are straightforward.
  For the third we proceed by induction on $\max \{\ottsym{\mbox{$\mid$}}  \tyAv  \ottsym{\mbox{$\mid$}}, \ottsym{\mbox{$\mid$}}  \tyAv'  \ottsym{\mbox{$\mid$}}\}$ and
  case analysis on $ \tyapx{ \tyAv }{ \tyAv' } $.


  \begin{case}{$ \tyapx{  \valBotV  }{ \tyAv' } $}
    $ \lrval{ \tyAv' }  \subseteq \synval  =  \lrval{  \valBotV  } $
  \end{case}


  \begin{case}{$ \tyapx{  \ottnt{s}  }{ \ottnt{s'} } $ where $\ottnt{s}  \leq  \ottnt{s'}$}
    Suppose $ \mathin{ \ottnt{s_{{\mathrm{0}}}} }{  \lrval{ \ottnt{s'} }  } $.
    It follows $\ottnt{s'}  \leq  \ottnt{s_{{\mathrm{0}}}}$.
    We need to show $ \mathin{ \ottnt{s_{{\mathrm{0}}}} }{  \lrval{ \ottnt{s} }  } $, or $\ottnt{s}  \leq  \ottnt{s_{{\mathrm{0}}}}$, which follows by
    transitivity.
  \end{case}

  \begin{case}{$ \tyapx{ \ottsym{(}  \tyAv_{{\mathrm{1}}}  \ottsym{,}  \tyAv_{{\mathrm{2}}}  \ottsym{)} }{ \ottsym{(}  \tyAv'_{{\mathrm{1}}}  \ottsym{,}  \tyAv'_{{\mathrm{2}}}  \ottsym{)} } $ where $ \tyapx{ \tyAv_{{\mathrm{1}}} }{ \tyAv'_{{\mathrm{1}}} } $ and
      $ \tyapx{ \tyAv_{{\mathrm{2}}} }{ \tyAv'_{{\mathrm{2}}} } $}
    Suppose $ \mathin{  \tmPair{ \valMeta_{{\mathrm{1}}} }{ \valMeta_{{\mathrm{2}}} }  }{  \lrval{ \ottsym{(}  \tyAv'_{{\mathrm{1}}}  \ottsym{,}  \tyAv'_{{\mathrm{2}}}  \ottsym{)} }  } $.
    It follows $ \mathin{ \valMeta_{{\mathrm{1}}} }{  \lrval{ \tyAv'_{{\mathrm{1}}} }  } $ and $ \mathin{ \valMeta_{{\mathrm{2}}} }{  \lrval{ \tyAv'_{{\mathrm{2}}} }  } $.
    We need to show $ \mathin{  \tmPair{ \valMeta_{{\mathrm{1}}} }{ \valMeta_{{\mathrm{2}}} }  }{  \lrval{ \ottsym{(}  \tyAv_{{\mathrm{1}}}  \ottsym{,}  \tyAv_{{\mathrm{2}}}  \ottsym{)} }  } $.
    It suffices to show $ \mathin{ \valMeta_{{\mathrm{1}}} }{  \lrval{ \tyAv_{{\mathrm{1}}} }  } $ and $ \mathin{ \valMeta_{{\mathrm{2}}} }{  \lrval{ \tyAv_{{\mathrm{2}}} }  } $.
    By the induction hypothesis, we have $ \lrval{ \tyAv'_{{\mathrm{1}}} }  \subseteq  \lrval{ \tyAv_{{\mathrm{1}}} } $
    and $ \lrval{ \tyAv'_{{\mathrm{2}}} }  \subseteq  \lrval{ \tyAv_{{\mathrm{2}}} } $.
    Our goal immediately follows.
  \end{case}

  \begin{case}{$ \tyapx{ \ottsym{\{}  \tyAv_{\ottmv{i}}  \ottsym{\mbox{$\mid$}}   \mathin{ \mathit{i} }{ \mathit{I} }   \ottsym{\}} }{ \ottsym{\{}  \tyAv'_{\ottmv{j}}  \ottsym{\mbox{$\mid$}}   \mathin{ \mathit{j} }{ \mathit{J} }   \ottsym{\}} } $ where
      $ \mathforall{  \mathin{ \mathit{i} }{ \mathit{I} }  }{  \mathexists{  \mathin{ \mathit{j} }{ \mathit{J} }  }{  \tyapx{ \tyAv_{\ottmv{i}} }{ \tyAv'_{\ottmv{j}} }  }  } $}
    Let $ \mathin{ \ottsym{\{}  \valMeta_{\ottmv{k}}  \ottsym{\mbox{$\mid$}}   \mathin{ \mathit{k} }{ \mathit{K} }   \ottsym{\}} }{  \lrval{ \ottsym{\{}  \tyAv'_{\ottmv{j}}  \ottsym{\mbox{$\mid$}}   \mathin{ \mathit{j} }{ \mathit{J} }   \ottsym{\}} }  } $.
    We must show $ \mathin{ \ottsym{\{}  \valMeta_{\ottmv{k}}  \ottsym{\mbox{$\mid$}}   \mathin{ \mathit{k} }{ \mathit{K} }   \ottsym{\}} }{  \lrval{ \ottsym{\{}  \tyAv_{\ottmv{i}}  \ottsym{\mbox{$\mid$}}   \mathin{ \mathit{i} }{ \mathit{I} }   \ottsym{\}} }  } $.
    From the definition of the value predicate, we have $f : J \to K$ such that:
    \begin{equation}
      \label{eqn:logsem:downclos:set}
      \mathforall{ \mathin{ \mathit{k} }{ \mathit{K} } }
        {\valMeta_{\ottmv{k}} \in \lrval{ \tyBigJoin{\mathit{j} \in \mathinv{f}(\mathit{k}) }{\tyAv'_{\ottmv{j}}} }}
    \end{equation}
    From our assumption for this case, we also have $g : I \to J$ such that $\mathforall{ \mathin{ \mathit{i} }{ \mathit{I} } }{\tyapx{\tyAv_{\ottmv{i}}}{\tyAv'_{g(i)}}}$.
    Consider arbitrary $ \mathin{ \mathit{k} }{ \mathit{K} } $.
    The definition of the value predicate requires us to show:
    \[
      \valMeta_{\ottmv{k}} \in \lrval{\tyBigJoin{\mathit{i} \in \mathinv{(\comp{f}{g})}(\mathit{k})}{\tyAv_{\ottmv{i}}}}
    \]
    We now derive:
    \[
      \begin{array}{rcll}
        \tyBigJoin{\mathit{i} \in \mathinv{(\comp{f}{g})}(\mathit{k})}
                  {\tyAv_{\ottmv{i}}} &\tyapxsym
      & \tyBigJoin{\mathit{i} \in \mathinv{(\comp{f}{g})}(\mathit{k})}
                  {\tyAv'_{g(\mathit{i})}}
      & \text{monotonicity of joins and the fact
          $\tyapx{\tyAv_{\ottmv{i}}}{\tyAv'_{g(\mathit{i})}}$ for $ \mathin{ \mathit{i} }{ \mathit{I} } $}\\
      &= & \tyBigJoin{\mathit{j} \in (\comp{g}{\comp{\mathinv{g}}{\mathinv{f}}})(\mathit{k})}{\tyAv'_{\ottmv{j}}}
      & \text{properties of inverse images} \\
      &\tyapxsym
      & \tyBigJoin{\mathit{j} \in \mathinv{f}(\mathit{k})}
                  {\tyAv'_{\ottmv{j}}}
      & \text{properties of joins and inverse images}
      \end{array}
    \]
    We have an induction hypothesis corresponding the above fact:
    \[
      \lrval{\tyBigJoin{\mathit{j} \in \mathinv{f}(\mathit{k})}{\tyAv'_{\ottmv{j}}}}
      \subseteq \lrval{\tyBigJoin{\mathit{i} \in
      \mathinv{(\comp{f}{g})}(\mathit{k})}{\tyAv_{\ottmv{i}}}}
    \]
    As a result, \eqref{eqn:logsem:downclos:set} completes the proof of this case.
  \end{case}

  \begin{case}{$ \tyapx{  \tyFun{  \mathin{ \mathit{i} }{ \mathit{I} }  }{ \tyAv_{\ottmv{i}} }{ \tyAc_{\ottmv{i}} }  }{  \tyFun{  \mathin{ \mathit{i} }{ \mathit{J} }  }{ \tyAv'_{\ottmv{i}} }{ \tyAc'_{\ottmv{i}} }  } $ where
      $\mathforall{ \mathin{ \mathit{i} }{ \mathit{I} } }
                 {\mathexists{ \mathit{J'}  \subseteq  \mathit{J} }{ \tyapx{  \tyBigJoin{  \mathin{ \mathit{j} }{ \mathit{J'} }  }{ \tyAv'_{\ottmv{j}} }  }{ \tyAv_{\ottmv{i}} }  \textand
                      \tyapx{ \tyAc_{\ottmv{i}} }{  \tyBigJoin{  \mathin{ \mathit{j} }{ \mathit{J'} }  }{ \tyAc'_{\ottmv{j}} }  } }}$}
    Let $ \mathin{  \valLam{ \mathit{x} }{ \tmMeta }  }{  \lrval{  \tyFun{  \mathin{ \mathit{i} }{ \mathit{J} }  }{ \tyAv'_{\ottmv{i}} }{ \tyAc'_{\ottmv{i}} }  }  } $.
    We must show $ \mathin{  \valLam{ \mathit{x} }{ \tmMeta }  }{  \lrval{  \tyFun{  \mathin{ \mathit{i} }{ \mathit{J} }  }{ \tyAv_{\ottmv{i}} }{ \tyAc_{\ottmv{i}} }  }  } $.
    This requires us to consider an arbitrary set $ \mathit{I'}  \subseteq  \mathit{I} $ such that
    $ \mathin{ \valMeta }{  \lrval{  \tyBigJoin{  \mathin{ \mathit{i} }{ \mathit{I'} }  }{ \tyAv_{\ottmv{i}} }  }  } $ and prove
    $ \mathin{  \subst{ \tmMeta }{ \valMeta }{ \mathit{x} }  }{  \lrexp{  \tyBigJoin{  \mathin{ \mathit{i} }{ \mathit{I'} }  }{ \tyAc_{\ottmv{i}} }  }  } $.

    From our assumption for this case, we have a function
    $f : \mathit{I} \to \powerset{\mathit{J}}$ such that for all $ \mathin{ \mathit{i} }{ \mathit{I} } $:
    \begin{equation}
      \label{eqn:logsem:downclos:fun}
      \tyapx{\tyBigJoin{\mathit{j} \in f(\mathit{j})}{\tyAv'_{\ottmv{j}}}}
            {\tyAv_{\ottmv{i}}}\ \textand \ 
      \tyapx{\tyAc_{\ottmv{i}}}
            {\tyBigJoin{\mathit{j} \in f(\mathit{i})}{\tyAc'_{\ottmv{j}}}}
    \end{equation}
    Let $\mathit{J'} = \{ j \mid \mathexists{i \in I'}{j \in f(i)} \}$.
    Suppose we knew the following two facts:
    \begin{equation}
       \tyBigJoin{  \mathin{ \mathit{j} }{ \mathit{J'} }  }{ \tyAv'_{\ottmv{j}} }  \tyapxsym  \tyBigJoin{  \mathin{ \mathit{i} }{ \mathit{I'} }  }{ \tyAv_{\ottmv{i}} } 
      \ \textand\ %
       \tyBigJoin{  \mathin{ \mathit{i} }{ \mathit{I'} }  }{ \tyAc_{\ottmv{i}} }  \tyapxsym  \tyBigJoin{  \mathin{ \mathit{j} }{ \mathit{J'} }  }{ \tyAc'_{\ottmv{j}} } 
      \label{eqn:logsem:downclos:goal1}
    \end{equation}
    Then we would have the corresponding induction hypotheses:
    \[
     \lrval{  \tyBigJoin{  \mathin{ \mathit{i} }{ \mathit{I'} }  }{ \tyAv_{\ottmv{i}} }  }  \subseteq  \lrval{  \tyBigJoin{  \mathin{ \mathit{j} }{ \mathit{J'} }  }{ \tyAv'_{\ottmv{j}} }  } \ \textand\ 
     \lrexp{  \tyBigJoin{  \mathin{ \mathit{j} }{ \mathit{J'} }  }{ \tyAc'_{\ottmv{j}} }  }  \subseteq  \lrexp{  \tyBigJoin{  \mathin{ \mathit{i} }{ \mathit{I'} }  }{ \tyAc_{\ottmv{i}} }  } 
    \]
    The former would allow us to instantiate the fact that
    $ \mathin{  \valLam{ \mathit{x} }{ \tmMeta }  }{  \lrval{  \tyFun{  \mathin{ \mathit{i} }{ \mathit{J} }  }{ \tyAv'_{\ottmv{i}} }{ \tyAc'_{\ottmv{i}} }  }  } $ with $\valMeta$.
    It follows ${ \mathin{  \subst{ \tmMeta }{ \valMeta }{ \mathit{x} }  }{  \lrexp{  \tyBigJoin{  \mathin{ \mathit{j} }{ \mathit{J'} }  }{ \tyAc'_{\ottmv{j}} }  }  } }$.
    Finally, the latter induction hypothesis would complete the proof.

    Thus, to complete the proof it suffices to show \eqref{eqn:logsem:downclos:goal1}.
    We derive it below.
    \[
      \begin{array}{rcll}
         \tyBigJoin{  \mathin{ \mathit{j} }{ \mathit{J'} }  }{ \tyAv'_{\ottmv{j}} }  &\tyapxsym
        &\tyBigJoin{ \mathin{ \mathit{i} }{ \mathit{I'} } }
          {\tyBigJoin{\mathit{j} \in f(\mathit{i})}{\tyAv'_{\ottmv{j}}}}
          & \text{properties of joins and definition of $J'$} \\
        &\tyapxsym
          & \tyBigJoin{  \mathin{ \mathit{i} }{ \mathit{I'} }  }{ \tyAv_{\ottmv{i}} }  & \text{monotonicity of joins and \eqref{eqn:logsem:downclos:fun}}
      \end{array}
    \]
    \[
      \begin{array}{rcll}
         \tyBigJoin{  \mathin{ \mathit{i} }{ \mathit{I'} }  }{ \tyAc_{\ottmv{i}} }  &\tyapxsym
        &\tyBigJoin{ \mathin{ \mathit{i} }{ \mathit{I'} } }
        {\tyBigJoin{\mathit{j} \in f(\mathit{i})}{\tyAc'_{\ottmv{j}}}}
        & \text{monotonicity of joins and \eqref{eqn:logsem:downclos:fun}} \\
                           &\tyapxsym & \tyBigJoin{  \mathin{ \mathit{j} }{ \mathit{J'} }  }{ \tyAc'_{\ottmv{j}} } 
        &\text{properties of joins and definition of $J'$}
      \end{array}
    \]
  \end{case}
\end{proof}

\begin{lemma}
  \label{lem:apdx:logsem:join:val-weak}

  If $ \mathin{ \valMeta }{  \lrval{ \tyAv }  } $ or $ \mathin{ \valMeta' }{  \lrval{ \tyAv }  } $ then $ \mathin{  \synjoin{ \valMeta }{ \valMeta' }  }{  \lrres{ \tyAv }  } $.
\end{lemma}
\begin{proof}
  Routine induction on $\tyAv$.
\end{proof}

\begin{lemma}
  \label{lem:apdx:logsem:join:val-strong}
  \leavevmode
  \begin{enumerate}
  \item If $ \mathin{ \valMeta }{  \lrval{ \tyAv }  } $ and $ \mathin{ \valMeta' }{  \lrval{ \tyAv' }  } $ then
    $ \mathin{  \synjoin{ \valMeta }{ \valMeta' }  }{  \lrres{  \tyJoin{ \tyAv }{ \tyAv' }  }  } $.
  \item If $ \mathin{ \resMeta }{  \lrres{ \tyAc }  } $ and $ \mathin{ \resMeta' }{  \lrres{ \tyAc' }  } $ then
    $ \mathin{  \synjoin{ \resMeta }{ \resMeta' }  }{  \lrres{  \tyJoin{ \tyAc }{ \tyAc' }  }  } $.
  \end{enumerate}
\end{lemma}
\begin{proof}
  We prove both parts simultaneously by induction on the maximum
  size of the two types involved.
  The first part uses \Lref{lem:apdx:logsem:sem-downclos} and
  \Lref{lem:apdx:logsem:join:val-weak}.

  We prove the second part by case analysis.
  If either $\resMeta$ or $\resMeta'$ is $ \tmTopC $ then so is $ \synjoin{ \resMeta }{ \resMeta' } $,
  making the goal trivial.
  We can thus assume $\resMeta$, $\resMeta'$, $\tyAc$ and $\tyAc'$ are not $ \tmTopC $.
  If $\resMeta$ is $ \tmBotC $ then so is $\tyAc$ and we have
  $ \synjoin{ \resMeta }{ \resMeta' }   \ottsym{=}  \resMeta'$ and $ \tyJoin{ \tyAc }{ \tyAc' }   \ottsym{=}  \tyAc'$, making the goal immediate.
  Without loss of generality, from here we can assume both $\resMeta$ and
  $\resMeta'$ are values.
  It is still possible that $\tyAc$ or $\tyAc'$ is $ \tmBotC $; a
  straightforward case analysis on these formulae allows us to apply
  \Lref{lem:apdx:logsem:join:val-weak} and the first part of this lemma to
  complete the proof.
\end{proof}

\begin{lemma}
  \label{lem:apdx:logsem:join:exp}

  If $ \mathin{ \tmMeta }{  \lrexp{ \tyAc }  } $ and $ \mathin{ \tmMeta' }{  \lrexp{ \tyAc' }  } $ then $ \mathin{  \tmJoin{ \tmMeta }{ \tmMeta' }  }{  \lrexp{  \tyJoin{ \tyAc }{ \tyAc' }  }  } $.
\end{lemma}
\begin{proof}
  Straightforward application of \Lref{lem:apdx:logsem:join:val-strong} and
  \Lref{lem:apdx:logsem:mon-bind}.
\end{proof}

\begin{lemma}[Semantic Join]
  \label{lem:apdx:logsem:join}

  Suppose there exists $f : \mathit{I} \to \mathit{J}$ such that for all
  $ \mathin{ \mathit{j} }{ \mathit{J} } $ we have $ \mathin{ \tmMeta_{\ottmv{j}} }{  \lrexp{  \tyBigJoin{  \mathin{ \mathit{i} }{  \mathinv{ \mathit{f} }   \ottsym{(}  \mathit{j}  \ottsym{)} }  }{ \tyAc_{\ottmv{i}} }  }  } $.
  Then $ \mathin{  \tmJoinManySet{  \mathin{ \mathit{j} }{ \mathit{J} }  }{ \tmMeta_{\ottmv{j}} }  }{  \lrexp{  \tyBigJoin{  \mathin{ \mathit{i} }{ \mathit{I} }  }{ \tyAc_{\ottmv{i}} }  }  } $.
\end{lemma}
\begin{proof}
  Note that since $f$ is a total function, $\mathit{I} = \bigcup_{ \mathin{ \mathit{j} }{ \mathit{J} } }{\mathinv{f}(j)}$.
  It follows $ \tyapx{  \tyBigJoin{  \mathin{ \mathit{i} }{ \mathit{I} }  }{ \tyAc_{\ottmv{i}} }  }{  \tyBigJoin{  \mathin{ \mathit{j} }{ \mathit{J} }  }{  \tyBigJoin{  \mathin{ \mathit{i} }{  \mathinv{ \mathit{f} }   \ottsym{(}  \mathit{j}  \ottsym{)} }  }{ \tyAc_{\ottmv{i}} }  }  } $.
  Thus by \Lref{lem:apdx:logsem:sem-downclos}, it suffices to show:
  \[
     \mathin{  \tmJoinManySet{  \mathin{ \mathit{j} }{ \mathit{J} }  }{ \tmMeta_{\ottmv{j}} }  }{  \lrexp{  \tyBigJoin{  \mathin{ \mathit{j} }{ \mathit{J} }  }{  \tyBigJoin{  \mathin{ \mathit{i} }{  \mathinv{ \mathit{f} }   \ottsym{(}  \mathit{j}  \ottsym{)} }  }{ \tyAc_{\ottmv{i}} }  }  }  } 
  \]
  Let $ \mathin{ \mathit{j} }{ \mathit{J} } $.
  By applying \Lref{lem:apdx:logsem:join:exp} repeatedly (by convention, $\mathit{J}$ is assumed to be finite), our goal reduces to showing
  $ \mathin{ \tmMeta_{\ottmv{j}} }{  \lrexp{  \tyBigJoin{  \mathin{ \mathit{i} }{  \mathinv{ \mathit{f} }   \ottsym{(}  \mathit{j}  \ottsym{)} }  }{ \tyAc_{\ottmv{i}} }  }  } $, which we assumed as a premise.
\end{proof}


\begin{lemma}
  \label{lem:apdx:logsem:set-mono-res}
  Suppose there exists $f : \mathit{I} \to \mathit{J}$ such that for all
  $ \mathin{ \mathit{j} }{ \mathit{J} } $ we have $ \mathin{ \resMeta_{\ottmv{j}} }{  \lrres{  \tyBigJoin{  \mathin{ \mathit{i} }{  \mathinv{ \mathit{f} }   \ottsym{(}  \mathit{j}  \ottsym{)} }  }{ \tyAc_{\ottmv{i}} }  }  } $.
  Then $ \mathin{ \ottsym{\{}  \resMeta_{\ottmv{j}}  \ottsym{\mbox{$\mid$}}   \mathin{ \mathit{j} }{ \mathit{J} }   \ottsym{\}} }{  \lrexp{  \tyBigJoin{  \mathin{ \mathit{i} }{ \mathit{I} }  }{  \tySng{ \tyAc_{\ottmv{i}} }  }  }  } $.
\end{lemma}
\begin{proof}
  If $\resMeta_{\ottmv{j}}  \ottsym{=}   \tmTopC $ for \emph{any} $ \mathin{ \mathit{j} }{ \mathit{J} } $, then
  $\ottsym{\{}  \resMeta_{\ottmv{j}}  \ottsym{\mbox{$\mid$}}   \mathin{ \mathit{j} }{ \mathit{J} }   \ottsym{\}}  \stepsym^{*}   \tmTopC $ and we have shown our goal.
  Thus, we will assume all $\resMeta_{\ottmv{j}}$ and $\tyAc_{\ottmv{i}}$ are not $ \tmTopC $.

  Let $\mathit{J'} = \{ \mathit{j} \mid  \mathin{ \mathit{j} }{ \mathit{J} }  \textand \resMeta_{\ottmv{j}} \neq  \tmBotC  \}$.
  Every result $\resMeta_{\ottmv{j}}$ where $ \mathin{ \mathit{j} }{ \mathit{J'} } $ is a value.
  Using the operational semantics and \Lref{lem:apdx:logsem:antired}, it suffices
  to show:
  \[  \mathin{ \ottsym{\{}  \resMeta_{\ottmv{j}}  \ottsym{\mbox{$\mid$}}   \mathin{ \mathit{j} }{ \mathit{J'} }   \ottsym{\}} }{  \lrexp{  \tyBigJoin{  \mathin{ \mathit{i} }{ \mathit{I} }  }{  \tySng{ \tyAc_{\ottmv{i}} }  }  }  }  \]

  Let $\mathit{I'} = \{ \mathit{i} \mid  \mathin{ \mathit{i} }{ \mathit{I} }  \textand \tyAc_{\ottmv{i}} \neq  \tmBotC  \}$.
  Every formula $\tyAc_{\ottmv{i}}$ where $ \mathin{ \mathit{i} }{ \mathit{I'} } $ is a value formula;
  let $\tyAv_{\ottmv{i}}  \ottsym{=}  \tyAc_{\ottmv{i}}$.
  Using this and the definition of $\tySng{\mathemdash}$, we have
  $ \tyapx{  \tyBigJoin{  \mathin{ \mathit{i} }{ \mathit{I} }  }{  \tySng{ \tyAc_{\ottmv{i}} }  }  }{  \tyBigJoin{  \mathin{ \mathit{i} }{ \mathit{I'} }  }{ \ottsym{\{}  \tyAv_{\ottmv{i}}  \ottsym{\}} }  }  = \ottsym{\{}  \tyAv_{\ottmv{i}}  \ottsym{\mbox{$\mid$}}   \mathin{ \mathit{i} }{ \mathit{I'} }   \ottsym{\}}$.
  Thus, by \Lref{lem:apdx:logsem:sem-downclos}, it suffices to show:
  \[
     \mathin{ \ottsym{\{}  \resMeta_{\ottmv{j}}  \ottsym{\mbox{$\mid$}}   \mathin{ \mathit{j} }{ \mathit{J'} }   \ottsym{\}} }{  \lrexp{ \ottsym{\{}  \tyAv_{\ottmv{i}}  \ottsym{\mbox{$\mid$}}   \mathin{ \mathit{i} }{ \mathit{I'} }   \ottsym{\}} }  }  \supseteq  \lrval{ \ottsym{\{}  \tyAv_{\ottmv{i}}  \ottsym{\mbox{$\mid$}}   \mathin{ \mathit{i} }{ \mathit{I'} }   \ottsym{\}} } 
  \]
  Let $g : \mathit{I'} \to \mathit{J'}$ be
  $\mathit{f}$ restricted to the domain $\mathit{I'}$.
  Note the following:
  \begin{enumerate}
    \item The function $g$ is well defined.
      Suppose $ \mathin{ \mathit{i} }{ \mathit{I'} } $.
      By definition, $\tyAc_{\ottmv{i}}$ is not $ \tmBotC $.
      Our premise states that
      $\resMeta_{\mathit{f}  \ottsym{(}  \mathit{i}  \ottsym{)}} \in  \lrres{  \tyBigJoin{  \mathin{ \mathit{k} }{  \mathinv{ \mathit{f} }   \ottsym{(}  \mathit{f}  \ottsym{(}  \mathit{i}  \ottsym{)}  \ottsym{)} }  }{ \tyAc_{\ottmv{k}} }  } $.
      By properties of inverse images, $ \mathin{ \mathit{i} }{  \mathinv{ \mathit{f} }   \ottsym{(}  \mathit{f}  \ottsym{(}  \mathit{i}  \ottsym{)}  \ottsym{)} } $.
      These facts and the definition of the result predicate imply that
      $\resMeta_{\mathit{f}  \ottsym{(}  \mathit{i}  \ottsym{)}} \neq  \tmBotC $.
      Thus, $\mathit{f}  \ottsym{(}  \mathit{i}  \ottsym{)} \in \mathit{J'}$ for any $ \mathin{ \mathit{i} }{ \mathit{I'} } $.
    \item $ \mathinv{ g }   \ottsym{(}  \mathit{j}  \ottsym{)} =   \mathinv{ \mathit{f} }   \ottsym{(}  \mathit{j}  \ottsym{)}  \cap  \mathit{I'} $ for all $ \mathin{ \mathit{j} }{ \mathit{J'} } $ by properties
      of inverse images.
  \end{enumerate}
  Consider arbitrary $ \mathin{ \mathit{j} }{ \mathit{J'} } $.
  By the definition of the value predicate, it is enough to demonstrate:
  \[
   \mathin{ \resMeta_{\ottmv{j}} }{  \lrres{  \tyBigJoin{  \mathin{ \mathit{i} }{  \mathinv{ g }   \ottsym{(}  \mathit{j}  \ottsym{)} }  }{ \tyAv_{\ottmv{i}} }  }  } 
  \]
  From our premise we have:
  \[
   \mathin{ \resMeta_{\ottmv{j}} }{  \lrres{  \tyBigJoin{  \mathin{ \mathit{i} }{  \mathinv{ \mathit{f} }   \ottsym{(}  \mathit{j}  \ottsym{)} }  }{ \tyAc_{\ottmv{i}} }  }  } 
  \]
  By \Lref{lem:apdx:logsem:sem-downclos}, it suffices to show
  $ \tyapx{  \tyBigJoin{  \mathin{ \mathit{i} }{  \mathinv{ \mathit{f} }   \ottsym{(}  \mathit{j}  \ottsym{)} }  }{ \tyAc_{\ottmv{i}} }  }{  \tyBigJoin{  \mathin{ \mathit{i} }{  \mathinv{ g }   \ottsym{(}  \mathit{j}  \ottsym{)} }  }{ \tyAv_{\ottmv{i}} }  } $.
  Note $\tyAc_{\ottmv{i}}  \ottsym{=}   \tmBotC $ for every $ \mathin{ \mathit{i} }{  \mathit{I}  \setminus  \mathit{I'}  } $ so we have
  \[
     \tyapx{  \tyBigJoin{  \mathin{ \mathit{i} }{  \mathinv{ \mathit{f} }   \ottsym{(}  \mathit{j}  \ottsym{)} }  }{ \tyAc_{\ottmv{i}} }  }{  \tyBigJoin{  \mathin{ \mathit{i} }{   \mathinv{ \mathit{f} }   \ottsym{(}  \mathit{j}  \ottsym{)}  \cap  \mathit{I'}  }  }{ \tyAv_{\ottmv{i}} }  }  =  \tyBigJoin{  \mathin{ \mathit{i} }{  \mathinv{ g }   \ottsym{(}  \mathit{j}  \ottsym{)} }  }{ \tyAv_{\ottmv{i}} } 
  \]

\end{proof}

\begin{lemma}
  \label{lem:apdx:logsem:set-mono}
  Suppose $ \mathin{ \tmMeta_{\ottmv{i}} }{  \lrexp{ \tyAc_{\ottmv{i}} }  } $ for all $ \mathin{ \mathit{i} }{ \mathit{I} } $.
  Then we have $ \mathin{ \ottsym{\{}  \tmMeta_{\ottmv{i}}  \ottsym{\mbox{$\mid$}}   \mathin{ \mathit{i} }{ \mathit{I} }   \ottsym{\}} }{  \lrexp{  \tyBigJoin{  \mathin{ \mathit{i} }{ \mathit{I} }  }{  \tySng{ \tyAc_{\ottmv{i}} }  }  }  } $.
\end{lemma}
\begin{proof}
  Let $ \mathin{ \resMeta_{\ottmv{i}} }{  \lrres{ \tyAc_{\ottmv{i}} }  } $ for all $ \mathin{ \mathit{i} }{ \mathit{I} } $.
  By \Lref{lem:apdx:logsem:mon-bind}, it suffices to show
  $ \mathin{ \ottsym{\{}  \resMeta_{\ottmv{i}}  \ottsym{\mbox{$\mid$}}   \mathin{ \mathit{i} }{ \mathit{I} }   \ottsym{\}} }{  \lrexp{  \tyBigJoin{  \mathin{ \mathit{i} }{ \mathit{I} }  }{  \tySng{ \tyAc_{\ottmv{i}} }  }  }  } $.
  This follows from \Lref{lem:apdx:logsem:set-mono-res}
\end{proof}

\begin{lemma}
  \label{lem:apdx:logsem:pair}
  If $ \mathin{ \resMeta_{{\mathrm{1}}} }{  \lrres{ \tyAc_{{\mathrm{1}}} }  } $ and $ \mathin{ \resMeta_{{\mathrm{2}}} }{  \lrres{ \tyAc_{{\mathrm{2}}} }  } $ then
  $ \mathin{  \tmPair{ \resMeta_{{\mathrm{1}}} }{ \resMeta_{{\mathrm{2}}} }  }{  \lrexp{  \tyPairOp{ \tyAc_{{\mathrm{1}}} }{ \tyAc_{{\mathrm{2}}} }  }  } $.
\end{lemma}
\begin{proof}
  We proceed by case analysis on $\tyAc_{{\mathrm{1}}}$.
  If $\tyAc_{{\mathrm{1}}}  \ottsym{=}   \tmTopC $ then $\resMeta_{{\mathrm{1}}}  \ottsym{=}   \tmTopC $ and $ \tmPair{ \resMeta_{{\mathrm{1}}} }{ \resMeta_{{\mathrm{2}}} }   \stepsym^{*}   \tmTopC $.
  If $\tyAc_{{\mathrm{2}}}  \ottsym{=}   \tmBotC $ then $ \tyPairOp{ \tyAc_{{\mathrm{1}}} }{ \tyAc_{{\mathrm{2}}} }   \ottsym{=}   \tmBotC $ and our goal is immediate.
  Otherwise, let $\tyAv_{{\mathrm{1}}}  \ottsym{=}  \tyAc_{{\mathrm{1}}}$.

  We continue by case analysis on $\tyAc_{{\mathrm{2}}}$.
  As before, if $\tyAc_{{\mathrm{2}}}  \ottsym{=}   \tmTopC $ or $\tyAc_{{\mathrm{2}}}  \ottsym{=}   \tmBotC $, the goal is easily fulfilled.
  Otherwise, let $\tyAv_{{\mathrm{2}}}  \ottsym{=}  \tyAc_{{\mathrm{2}}}$.
  We have $ \tyPairOp{ \tyAc_{{\mathrm{1}}} }{ \tyAc_{{\mathrm{2}}} }   \ottsym{=}   \tyPairOp{ \tyAv_{{\mathrm{1}}} }{ \tyAv_{{\mathrm{2}}} }  = \ottsym{(}  \tyAv_{{\mathrm{1}}}  \ottsym{,}  \tyAv_{{\mathrm{2}}}  \ottsym{)}$.

  It remains to show $ \mathin{  \tmPair{ \resMeta_{{\mathrm{1}}} }{ \resMeta_{{\mathrm{2}}} }  }{  \lrexp{ \ottsym{(}  \tyAv_{{\mathrm{1}}}  \ottsym{,}  \tyAv_{{\mathrm{2}}}  \ottsym{)} }  } $.
  Since $ \tmBotC $ is not included in $ \lrres{ \tyAv_{\ottmv{i}} } $, we know $\resMeta_{{\mathrm{1}}}$ and
  $\resMeta_{{\mathrm{2}}}$ are not $ \tmBotC $.
  If either is $ \tmTopC $ then $ \tmPair{ \resMeta_{{\mathrm{1}}} }{ \resMeta_{{\mathrm{2}}} }   \stepsym^{*}   \tmTopC $, completing the proof.
  Otherwise, let $\valMeta_{{\mathrm{1}}}  \ottsym{=}  \resMeta_{{\mathrm{1}}}$ and $\valMeta_{{\mathrm{2}}}  \ottsym{=}  \resMeta_{{\mathrm{2}}}$.

  We must show $ \mathin{  \tmPair{ \valMeta_{{\mathrm{1}}} }{ \valMeta_{{\mathrm{2}}} }  }{  \lrexp{ \ottsym{(}  \tyAv_{{\mathrm{1}}}  \ottsym{,}  \tyAv_{{\mathrm{2}}}  \ottsym{)} }  }  \supseteq  \lrval{ \ottsym{(}  \tyAv_{{\mathrm{1}}}  \ottsym{,}  \tyAv_{{\mathrm{2}}}  \ottsym{)} } $.
  From the definition of the value predicate, it suffices to show
  $ \mathin{ \valMeta_{{\mathrm{1}}} }{  \lrval{ \tyAv_{{\mathrm{1}}} }  } $ and $ \mathin{ \valMeta_{{\mathrm{2}}} }{  \lrres{ \tyAv_{{\mathrm{2}}} }  } $.
  This follows from our premise that $ \mathin{ \valMeta_{{\mathrm{1}}} }{  \lrres{ \tyAv_{{\mathrm{1}}} }  } $ and $ \mathin{ \valMeta_{{\mathrm{2}}} }{  \lrres{ \tyAv_{{\mathrm{2}}} }  } $.
\end{proof}

\begin{lemma}[Fundamental Property]
  \label{lem:apdx:log-sem:fund-prop}
  If $ \judg{ \Gamma }{ \tmMeta }{ \tyAc } $ then $ \judglr{ \Gamma }{ \tmMeta }{ \tyAc } $.
\end{lemma}
\begin{proof}
  We proceed by nested induction first on $\tmMeta$ and then on $ \judg{ \Gamma }{ \tmMeta }{ \tyAc } $.
  In every case, we start by assuming $ \mathin{ \substmeta }{  \lrenv{ \Gamma }  } $ and then demonstrate
  that $ \mathin{ \substmeta  \ottsym{(}  \tmMeta  \ottsym{)} }{  \lrexp{ \tyAc }  } $.
  We show the cases of the \emph{inner} induction below, sometimes making use of
  an \emph{outer} induction hypothesis for subterms of $\tmMeta$.

  \begin{case}{$ \judg{ \Gamma }{ \tmMeta }{ \tyAc' } $ and $ \tyapx{ \tyAc }{ \tyAc' } $}
    The inner induction hypothesis is $ \judglr{ \Gamma }{ \tmMeta }{ \tyAc' } $.
    From this we have $ \mathin{ \substmeta  \ottsym{(}  \tmMeta  \ottsym{)} }{  \lrexp{ \tyAc' }  } $.
    By Semantic Downward Closure, $ \mathin{ \substmeta  \ottsym{(}  \tmMeta  \ottsym{)} }{  \lrexp{ \tyAc }  } $.
  \end{case}

  \begin{case}{$\tyAc  \ottsym{=}   \tmBotC $}
    We must show $ \mathin{ \substmeta  \ottsym{(}  \tmMeta  \ottsym{)} }{  \lrexp{  \tmBotC  }  } $.
    By the operational semantics, $\substmeta  \ottsym{(}  \tmMeta  \ottsym{)}  \stepsym^{*}   \tmBotC $ so it remains
    to show $ \mathin{  \tmBotC  }{  \lrres{  \tmBotC  }  } $.
    This is immediate from the definition of the result predicate.
  \end{case}

  \begin{case}{$\tmMeta  \ottsym{=}  \valMeta$ and $\tyAc  \ottsym{=}   \valBotV $}
    We must show $ \mathin{ \substmeta  \ottsym{(}  \valMeta  \ottsym{)} }{  \lrexp{  \valBotV  }  }  \supseteq  \lrval{  \valBotV  } $.
    But $ \lrval{  \valBotV  }  = \synval$ so this is immediate.
  \end{case}

  \begin{case}{$\tmMeta  \ottsym{=}   \tmTopC $ and $\tyAc  \ottsym{=}   \tmTopC $}
    We must show $ \mathin{  \tmTopC  }{  \lrexp{  \tmTopC  }  }  \supseteq  \lrres{  \tmTopC  } $.
    But $ \lrres{  \tmTopC  }  = \{  \tmTopC  \}$ so this is immediate.
  \end{case}

  \begin{case}{$\tmMeta  \ottsym{=}  \mathit{x}$ where $\Gamma  \ottsym{(}  \mathit{x}  \ottsym{)}  \ottsym{=}  \tyAv$ and $\tyAc  \ottsym{=}  \tyAv$}
    We have $ \mathin{ \substmeta  \ottsym{(}  \mathit{x}  \ottsym{)} }{  \lrval{ \tyAv }  } $ from the assumption that $ \mathin{ \substmeta }{  \lrenv{ \Gamma }  } $.
    \Lref{lem:apdx:logsem:mon-unit} yields $ \mathin{ \substmeta  \ottsym{(}  \mathit{x}  \ottsym{)} }{  \lrexp{ \tyAv }  } $.
  \end{case}

  \begin{case}{$\tmMeta  \ottsym{=}   \tmJoin{ \tmMeta_{{\mathrm{1}}} }{ \tmMeta_{{\mathrm{2}}} } $ and $\tyAc  \ottsym{=}   \tyJoin{ \tyAc_{{\mathrm{1}}} }{ \tyAc_{{\mathrm{2}}} } $ where
      $ \judg{ \Gamma }{ \tmMeta_{{\mathrm{1}}} }{ \tyAc_{{\mathrm{1}}} } $ and $ \judg{ \Gamma }{ \tmMeta_{{\mathrm{2}}} }{ \tyAc_{{\mathrm{2}}} } $}
    We must show $ \mathin{  \tmJoin{ \substmeta  \ottsym{(}  \tmMeta_{{\mathrm{1}}}  \ottsym{)} }{ \substmeta  \ottsym{(}  \tmMeta_{{\mathrm{2}}}  \ottsym{)} }  }{  \lrexp{  \tyJoin{ \tyAc_{{\mathrm{1}}} }{ \tyAc_{{\mathrm{2}}} }  }  } $.
    We have induction hypotheses for $\tmMeta_{{\mathrm{1}}}$ and $\tmMeta_{{\mathrm{2}}}$:
    \[
        \judg{ \Gamma }{ \tmMeta_{{\mathrm{1}}} }{ \tyAc_{{\mathrm{1}}} }  \textand  \judg{ \Gamma }{ \tmMeta_{{\mathrm{2}}} }{ \tyAc_{{\mathrm{2}}} } 
    \]
    It follows $ \mathin{ \substmeta  \ottsym{(}  \tmMeta_{{\mathrm{1}}}  \ottsym{)} }{  \lrexp{ \tyAc_{{\mathrm{1}}} }  } $ and $ \mathin{ \substmeta  \ottsym{(}  \tmMeta_{{\mathrm{2}}}  \ottsym{)} }{  \lrexp{ \tyAc_{{\mathrm{2}}} }  } $.
    Our goal is then a consequence of \Lref{lem:apdx:logsem:join:exp}.
  \end{case}

  \begin{case}{$\tmMeta  \ottsym{=}  \ottnt{s}$ and $\tyAc  \ottsym{=}  \ottnt{s}$}
    It is immediate that $ \mathin{ \ottnt{s} }{  \lrexp{ \ottnt{s} }  }  \supseteq  \lrval{ \ottnt{s} } $.
  \end{case}

  \begin{case}{$\tmMeta  \ottsym{=}   \tmPair{ \tmMeta_{{\mathrm{1}}} }{ \tmMeta_{{\mathrm{2}}} } $ and $\tyAc =  \tyPairOp{ \tyAc_{{\mathrm{1}}} }{ \tyAc_{{\mathrm{2}}} } $ where $ \judg{ \Gamma }{ \tmMeta_{{\mathrm{1}}} }{ \tyAc_{{\mathrm{1}}} } $ and $ \judg{ \Gamma }{ \tmMeta_{{\mathrm{2}}} }{ \tyAc_{{\mathrm{2}}} } $}
    We must show $ \mathin{  \tmPair{ \substmeta  \ottsym{(}  \tmMeta_{{\mathrm{1}}}  \ottsym{)} }{ \substmeta  \ottsym{(}  \tmMeta_{{\mathrm{2}}}  \ottsym{)} }  }{  \lrexp{  \tyPairOp{ \tyAc_{{\mathrm{1}}} }{ \tyAc_{{\mathrm{2}}} }  }  } $.
    We have induction hypotheses for $\tmMeta_{{\mathrm{1}}}$ and $\tmMeta_{{\mathrm{2}}}$:
    \[
       \judg{ \Gamma }{ \tmMeta_{{\mathrm{1}}} }{ \tyAc_{{\mathrm{1}}} }  \textand  \judg{ \Gamma }{ \tmMeta_{{\mathrm{2}}} }{ \tyAc_{{\mathrm{2}}} } 
    \]
    It follows $ \mathin{ \substmeta  \ottsym{(}  \tmMeta_{{\mathrm{1}}}  \ottsym{)} }{  \lrexp{ \tyAc_{{\mathrm{1}}} }  } $ and $ \mathin{ \substmeta  \ottsym{(}  \tmMeta_{{\mathrm{2}}}  \ottsym{)} }{  \lrexp{ \tyAc_{{\mathrm{2}}} }  } $.
    Our goal is then a consequence of \Lref{lem:apdx:logsem:pair}.
  \end{case}

  \begin{case}{$\tmMeta  \ottsym{=}  \ottsym{\{}  \tmMeta_{\ottmv{i}}  \ottsym{\mbox{$\mid$}}   \mathin{ \mathit{i} }{ \mathit{I} }   \ottsym{\}}$, $\tyAc  \ottsym{=}   \tyBigJoin{  \mathin{ \mathit{i} }{ \mathit{I} }  }{  \tySng{ \tyAc_{\ottmv{i}} }  } $ where
    $ \mathforall{  \mathin{ \mathit{i} }{ \mathit{I} }  }{  \judg{ \Gamma }{ \tmMeta_{\ottmv{i}} }{ \tyAc_{\ottmv{i}} }  } $}
    We must show $ \mathin{ \ottsym{\{}  \substmeta  \ottsym{(}  \tmMeta_{\ottmv{i}}  \ottsym{)}  \ottsym{\mbox{$\mid$}}   \mathin{ \mathit{i} }{ \mathit{I} }   \ottsym{\}} }{  \lrexp{  \tyBigJoin{  \mathin{ \mathit{i} }{ \mathit{I} }  }{  \tySng{ \tyAc_{\ottmv{i}} }  }  }  } $.
    Let $ \mathin{ \mathit{i} }{ \mathit{I} } $.
    By \Lref{lem:apdx:logsem:set-mono}, it suffices to show
    $ \mathin{ \substmeta  \ottsym{(}  \tmMeta_{\ottmv{i}}  \ottsym{)} }{  \lrexp{ \tyAc_{\ottmv{i}} }  } $.
    This is an immediate consequence of the induction hypothesis, $ \judglr{ \Gamma }{ \tmMeta_{\ottmv{i}} }{ \tyAc_{\ottmv{i}} } $.
  \end{case}

  \begin{case}{$\tmMeta  \ottsym{=}   \valLam{ \mathit{x} }{ \tmMeta' } $, $\tyAc  \ottsym{=}   \tyFun{  \mathin{ \mathit{i} }{ \mathit{I} }  }{ \tyAv_{\ottmv{i}} }{ \tyAc_{\ottmv{i}} } $ where
    $ \mathforall{  \mathin{ \mathit{i} }{ \mathit{I} }  }{  \judg{ \Gamma  \ottsym{,}  \mathit{x}  \ottsym{:}  \tyAv_{\ottmv{i}} }{ \tmMeta }{ \tyAc_{\ottmv{i}} }  } $}
    We must show $ \mathin{  \valLam{ \mathit{x} }{ \substmeta  \ottsym{(}  \tmMeta'  \ottsym{)} }  }{  \lrexp{  \tyFun{  \mathin{ \mathit{i} }{ \mathit{I} }  }{ \tyAv_{\ottmv{i}} }{ \tyAc_{\ottmv{i}} }  }  }  \supseteq  \lrval{  \tyFun{  \mathin{ \mathit{i} }{ \mathit{I} }  }{ \tyAv_{\ottmv{i}} }{ \tyAc_{\ottmv{i}} }  } $.
    Consider arbitrary $ \mathit{J}  \subseteq  \mathit{I} $ and $ \mathin{ \valMeta }{  \lrval{  \tyBigJoin{  \mathin{ \mathit{j} }{ \mathit{J} }  }{ \tyAv_{\ottmv{j}} }  }  } $.
    The value predicate requires us to show
    $ \mathin{  \subst{ \substmeta  \ottsym{(}  \tmMeta'  \ottsym{)} }{ \valMeta }{ \mathit{x} }  }{  \lrexp{  \tyBigJoin{  \mathin{ \mathit{j} }{ \mathit{J} }  }{ \tyAc_{\ottmv{j}} }  }  } $.
    From Directedness we have $ \judg{ \Gamma  \ottsym{,}  \mathit{x}  \ottsym{:}   \tyBigJoin{  \mathin{ \mathit{j} }{ \mathit{J} }  }{ \tyAv_{\ottmv{j}} }  }{ \tmMeta' }{  \tyBigJoin{  \mathin{ \mathit{j} }{ \mathit{J} }  }{ \tyAc_{\ottmv{j}} }  } $.
    This gives us an induction hypothesis for $\tmMeta'$:
    \[  \judglr{ \Gamma  \ottsym{,}  \mathit{x}  \ottsym{:}   \tyBigJoin{  \mathin{ \mathit{j} }{ \mathit{J} }  }{ \tyAv_{\ottmv{j}} }  }{ \tmMeta' }{  \tyBigJoin{  \mathin{ \mathit{j} }{ \mathit{J} }  }{ \tyAc_{\ottmv{j}} }  }  \]
    It follows that $ \mathin{  \subst{ \substmeta  \ottsym{(}  \tmMeta'  \ottsym{)} }{ \valMeta }{ \mathit{x} }  }{  \lrexp{  \tyBigJoin{  \mathin{ \mathit{j} }{ \mathit{J} }  }{ \tyAc_{\ottmv{j}} }  }  } $.
  \end{case}

  \begin{case}{$\tmMeta  \ottsym{=}  \ottkw{let} \, \ottnt{s}  \ottsym{=}  \tmMeta_{{\mathrm{1}}} \, \ottkw{in} \, \tmMeta_{{\mathrm{2}}}$ where
      $ \judg{ \Gamma }{ \tmMeta_{{\mathrm{1}}} }{ \ottnt{s} } $ and $ \judg{ \Gamma }{ \tmMeta_{{\mathrm{2}}} }{ \tyAc } $}
    We must show $ \mathin{ \ottkw{let} \, \ottnt{s}  \ottsym{=}  \substmeta  \ottsym{(}  \tmMeta_{{\mathrm{1}}}  \ottsym{)} \, \ottkw{in} \, \substmeta  \ottsym{(}  \tmMeta_{{\mathrm{2}}}  \ottsym{)} }{  \lrexp{ \tyAc }  } $.
    We have induction hypotheses for $\tmMeta_{{\mathrm{1}}}$ and $\tmMeta_{{\mathrm{2}}}$:
    \[
       \judg{ \Gamma }{ \tmMeta_{{\mathrm{1}}} }{ \ottnt{s} }  \textand  \judg{ \Gamma }{ \tmMeta_{{\mathrm{2}}} }{ \tyAc } 
    \]
    It follows $ \mathin{ \substmeta  \ottsym{(}  \tmMeta_{{\mathrm{1}}}  \ottsym{)} }{  \lrexp{ \ottnt{s} }  } $ and $ \mathin{ \substmeta  \ottsym{(}  \tmMeta_{{\mathrm{2}}}  \ottsym{)} }{  \lrexp{ \tyAc }  } $.
    Let $ \mathin{ \ottnt{s'} }{  \lrval{ \ottnt{s} }  } $.
    By \Lref{lem:apdx:logsem:mon-bind}, it suffices to show:
    \[
       \mathin{ \ottkw{let} \, \ottnt{s}  \ottsym{=}  \ottnt{s'} \, \ottkw{in} \, \substmeta  \ottsym{(}  \tmMeta_{{\mathrm{2}}}  \ottsym{)} }{  \lrexp{ \tyAc }  } 
    \]
    From the definition of the value predicate we have $\ottnt{s}  \leq  \ottnt{s'}$.
    Thus, $\ottkw{let} \, \ottnt{s}  \ottsym{=}  \ottnt{s'} \, \ottkw{in} \, \substmeta  \ottsym{(}  \tmMeta_{{\mathrm{2}}}  \ottsym{)}$ so our goal follows from
    \Lref{lem:apdx:logsem:antired}.
  \end{case}

  \begin{case}{$\tmMeta  \ottsym{=}  \ottkw{let} \, \ottsym{(}  \mathit{x}  \ottsym{,}  \mathit{y}  \ottsym{)}  \ottsym{=}  \tmMeta_{{\mathrm{1}}} \, \ottkw{in} \, \tmMeta_{{\mathrm{2}}}$ where
    $ \judg{ \Gamma }{ \tmMeta_{{\mathrm{1}}} }{ \ottsym{(}  \tyAv  \ottsym{,}  \tyBv  \ottsym{)} } $ and
    $ \judg{ \Gamma  \ottsym{,}  \mathit{x}  \ottsym{:}  \tyAv  \ottsym{,}  \mathit{y}  \ottsym{:}  \tyBv }{ \tmMeta_{{\mathrm{2}}} }{ \tyAc } $}
    We must show $ \mathin{ \ottkw{let} \, \ottsym{(}  \mathit{x}  \ottsym{,}  \mathit{y}  \ottsym{)}  \ottsym{=}  \substmeta  \ottsym{(}  \tmMeta_{{\mathrm{1}}}  \ottsym{)} \, \ottkw{in} \, \substmeta  \ottsym{(}  \tmMeta_{{\mathrm{2}}}  \ottsym{)} }{  \lrexp{ \tyAc }  } $.
    We have induction hypotheses for $\tmMeta_{{\mathrm{1}}}$ and $\tmMeta_{{\mathrm{2}}}$:
    \[
       \judg{ \Gamma }{ \tmMeta_{{\mathrm{1}}} }{ \tyAc_{{\mathrm{1}}} }  \textand  \judg{ \Gamma  \ottsym{,}  \mathit{x}  \ottsym{:}  \tyAv  \ottsym{,}  \mathit{y}  \ottsym{:}  \tyBv }{ \tmMeta_{{\mathrm{2}}} }{ \tyAc } 
    \]
    It follows $ \mathin{ \substmeta  \ottsym{(}  \tmMeta_{{\mathrm{1}}}  \ottsym{)} }{  \lrexp{ \ottsym{(}  \tyAv  \ottsym{,}  \tyBv  \ottsym{)} }  } $.
    Let $ \mathin{  \tmPair{ \valMeta_{{\mathrm{1}}} }{ \valMeta'_{{\mathrm{1}}} }  }{  \lrval{ \ottsym{(}  \tyAv  \ottsym{,}  \tyBv  \ottsym{)} }  } $.
    By \Lref{lem:apdx:logsem:mon-bind}, it suffices to show:
    \[
       \mathin{ \ottkw{let} \, \ottsym{(}  \mathit{x}  \ottsym{,}  \mathit{y}  \ottsym{)}  \ottsym{=}   \tmPair{ \valMeta_{{\mathrm{1}}} }{ \valMeta'_{{\mathrm{1}}} }  \, \ottkw{in} \, \substmeta  \ottsym{(}  \tmMeta_{{\mathrm{2}}}  \ottsym{)} }{  \lrexp{ \tyAc }  } 
    \]
    From the definition of the value predicate, we have
    $ \mathin{ \valMeta_{{\mathrm{1}}} }{  \lrval{ \tyAv }  } $ and $ \mathin{ \valMeta'_{{\mathrm{1}}} }{  \lrval{ \tyBv }  } $.
    From our induction hypothesis for $\tmMeta_{{\mathrm{2}}}$ it follows
    $ \mathin{  \subst{  \subst{ \substmeta  \ottsym{(}  \tmMeta_{{\mathrm{2}}}  \ottsym{)} }{ \valMeta_{{\mathrm{1}}} }{ \mathit{x} }  }{ \valMeta'_{{\mathrm{1}}} }{ \mathit{y} }  }{  \lrexp{ \tyAc }  } $.
    Applying \Lref{lem:apdx:logsem:antired} completes this case.
  \end{case}

  \begin{case}{$\tmMeta  \ottsym{=}   \tmBigJoin{ \mathit{x} }{ \tmMeta' }{ \tmMeta'' } $, $\tyAc  \ottsym{=}   \tyBigJoin{  \mathin{ \mathit{i} }{ \mathit{I} }  }{ \tyAc_{\ottmv{i}} } $ where
    $ \judg{ \Gamma }{ \tmMeta' }{ \ottsym{\{}  \tyAv_{\ottmv{i}}  \ottsym{\mbox{$\mid$}}   \mathin{ \mathit{i} }{ \mathit{I} }   \ottsym{\}} } $ and $ \mathforall{  \mathin{ \mathit{i} }{ \mathit{I} }  }{  \judg{ \Gamma  \ottsym{,}  \mathit{x}  \ottsym{:}  \tyAv_{\ottmv{i}} }{ \tmMeta'' }{ \tyAc_{\ottmv{i}} }  } $}
    We must show $ \mathin{  \tmBigJoin{ \mathit{x} }{ \substmeta  \ottsym{(}  \tmMeta'  \ottsym{)} }{ \substmeta  \ottsym{(}  \tmMeta''  \ottsym{)} }  }{  \lrexp{  \tyBigJoin{  \mathin{ \mathit{i} }{ \mathit{I} }  }{ \tyAc_{\ottmv{i}} }  }  } $.
    By the induction hypothesis, $ \judglr{ \Gamma }{ \tmMeta' }{ \ottsym{\{}  \tyAv_{\ottmv{i}}  \ottsym{\mbox{$\mid$}}   \mathin{ \mathit{i} }{ \mathit{I} }   \ottsym{\}} } $.
    It follows $ \mathin{ \substmeta  \ottsym{(}  \tmMeta'  \ottsym{)} }{  \lrexp{ \ottsym{\{}  \tyAv_{\ottmv{i}}  \ottsym{\mbox{$\mid$}}   \mathin{ \mathit{i} }{ \mathit{I} }   \ottsym{\}} }  } $.
    Fix arbitrary $ \mathin{ \valMeta }{  \lrval{ \ottsym{\{}  \tyAv_{\ottmv{i}}  \ottsym{\mbox{$\mid$}}   \mathin{ \mathit{i} }{ \mathit{I} }   \ottsym{\}} }  } $.
    Applying Monadic Bind, it suffices to show
    $ \mathin{  \tmBigJoin{ \mathit{x} }{ \valMeta }{ \substmeta  \ottsym{(}  \tmMeta''  \ottsym{)} }  }{  \lrexp{  \tyBigJoin{  \mathin{ \mathit{i} }{ \mathit{I} }  }{ \tyAc_{\ottmv{i}} }  }  } $.

    From the definition of the value predicate, we know
    $\valMeta  \ottsym{=}  \ottsym{\{}  \valMeta_{\ottmv{j}}  \ottsym{\mbox{$\mid$}}   \mathin{ \mathit{j} }{ \mathit{J} }   \ottsym{\}}$ and have $ \mathin{ \mathit{f} }{  \mathit{I}  \to  \mathit{J}  } $ where:
    \begin{equation}
      \label{eqn:logsem:fund-prop:tmBigJoin}
       \mathforall{  \mathin{ \mathit{j} }{ \mathit{J} }  }{  \mathin{ \valMeta_{\ottmv{j}} }{  \lrval{  \tyBigJoin{  \mathin{ \mathit{i} }{  \mathinv{ \mathit{f} }   \ottsym{(}  \mathit{j}  \ottsym{)} }  }{ \tyAv_{\ottmv{i}} }  }  }  } 
    \end{equation}
    By \Lref{lem:apdx:logsem:antired} it is enough to show
    $ \mathin{  \subst{  \tmJoinManySet{  \mathin{ \mathit{j} }{ \mathit{J} }  }{ \substmeta  \ottsym{(}  \tmMeta''  \ottsym{)} }  }{ \valMeta_{\ottmv{j}} }{ \mathit{x} }  }{  \lrexp{  \tyBigJoin{  \mathin{ \mathit{i} }{ \mathit{I} }  }{ \tyAc_{\ottmv{i}} }  }  } $.
    Let $ \mathin{ \mathit{j} }{ \mathit{J} } $.
    \Lref{lem:apdx:logsem:join} further reduces our obligation to proving
    $ \mathin{  \subst{ \substmeta  \ottsym{(}  \tmMeta''  \ottsym{)} }{ \valMeta_{\ottmv{j}} }{ \mathit{x} }  }{  \lrexp{  \tyBigJoin{  \mathin{ \mathit{i} }{  \mathinv{ \mathit{f} }   \ottsym{(}  \mathit{j}  \ottsym{)} }  }{ \tyAc_{\ottmv{i}} }  }  } $.
    By Directedness,
    \[
      { \judg{ \Gamma  \ottsym{,}  \mathit{x}  \ottsym{:}   \tyBigJoin{  \mathin{ \mathit{i} }{  \mathinv{ \mathit{f} }   \ottsym{(}  \mathit{j}  \ottsym{)} }  }{ \tyAv_{\ottmv{i}} }  }{ \tmMeta'' }{  \tyBigJoin{  \mathin{ \mathit{i} }{  \mathinv{ \mathit{f} }   \ottsym{(}  \mathit{j}  \ottsym{)} }  }{ \tyAc_{\ottmv{j}} }  } }
    \]
    We then have an induction hypothesis for $\tmMeta''$:
    \[
       \judglr{ \Gamma  \ottsym{,}  \mathit{x}  \ottsym{:}   \tyBigJoin{  \mathin{ \mathit{i} }{  \mathinv{ \mathit{f} }   \ottsym{(}  \mathit{j}  \ottsym{)} }  }{ \tyAv_{\ottmv{i}} }  }{ \tmMeta'' }{  \tyBigJoin{  \mathin{ \mathit{i} }{  \mathinv{ \mathit{f} }   \ottsym{(}  \mathit{j}  \ottsym{)} }  }{ \tyAc_{\ottmv{j}} }  } 
    \]
    Note that from \eqref{eqn:logsem:fund-prop:tmBigJoin} we have
    $ \mathin{  \extmap{ \substmeta }{ \mathit{x} }{ \valMeta_{\ottmv{j}} }  }{  \lrenv{ \Gamma  \ottsym{,}  \mathit{x}  \ottsym{:}   \tyBigJoin{  \mathin{ \mathit{i} }{  \mathinv{ \mathit{f} }   \ottsym{(}  \mathit{j}  \ottsym{)} }  }{ \tyAv_{\ottmv{i}} }  }  } $.
    It follows
    $ \mathin{  \subst{ \substmeta  \ottsym{(}  \tmMeta''  \ottsym{)} }{ \valMeta_{\ottmv{j}} }{ \mathit{x} }  }{  \lrexp{  \tyBigJoin{  \mathin{ \mathit{i} }{  \mathinv{ \mathit{f} }   \ottsym{(}  \mathit{j}  \ottsym{)} }  }{ \tyAc_{\ottmv{i}} }  }  } $.
  \end{case}

  \begin{case}{$\tmMeta  \ottsym{=}   \tmApp{ \tmMeta' }{ \tmMeta'' } $ where $ \judg{ \Gamma }{ \tmMeta' }{  \tyFunOne{ \tyAv }{ \tyAc }  } $ and $ \judg{ \Gamma }{ \tmMeta'' }{ \tyAv } $}
    We must show $ \mathin{  \tmApp{ \substmeta  \ottsym{(}  \tmMeta'  \ottsym{)} }{ \substmeta  \ottsym{(}  \tmMeta''  \ottsym{)} }  }{  \lrexp{ \tyAc }  } $.
    We have induction hypotheses for $\tmMeta'$ and $\tmMeta''$:
    \[
       \judglr{ \Gamma }{ \tmMeta' }{  \tyFunOne{ \tyAv }{ \tyAc }  }  \ \textand\  \judglr{ \Gamma }{ \tmMeta'' }{ \tyAv } 
    \]
    It follows $ \mathin{ \substmeta  \ottsym{(}  \tmMeta'  \ottsym{)} }{  \lrexp{  \tyFunOne{ \tyAv }{ \tyAc }  }  } $ and $ \mathin{ \substmeta  \ottsym{(}  \tmMeta''  \ottsym{)} }{  \lrexp{ \tyAv }  } $.
    Fix arbitrary $ \mathin{ \valMeta' }{  \lrval{  \tyFunOne{ \tyAv }{ \tyAc }  }  } $ and $ \mathin{ \valMeta'' }{  \lrval{ \tyAv }  } $.
    By Monadic Bind, it suffices to show
    $ \mathin{  \tmApp{ \valMeta' }{ \valMeta'' }  }{  \lrexp{ \tyAc }  } $.
    From the definition of the value predicate, we can see that
    $\valMeta'  \ottsym{=}   \valLam{ \mathit{x} }{ \tmMetaAlt } $ and $ \mathin{  \subst{ \tmMetaAlt }{ \valMeta'' }{ \mathit{x} }  }{  \lrexp{ \tyAc }  } $.
    Since $ \step{  \tmApp{ \valMeta' }{ \valMeta'' }  }{  \subst{ \tmMetaAlt }{ \valMeta'' }{ \mathit{x} }  } $, applying \Lref{lem:apdx:logsem:antired}
    completes the proof.
  \end{case}

  \begin{case}{$\tmMeta  \ottsym{=}  \ottkw{let} \, \ottsym{(}  \mathit{x}  \ottsym{,}  \mathit{y}  \ottsym{)}  \ottsym{=}  \tmMeta_{{\mathrm{1}}} \, \ottkw{in} \, \tmMeta_{{\mathrm{2}}}$ and $\tyAc  \ottsym{=}   \tmTopC $ where $ \judg{ \Gamma }{ \tmMeta_{{\mathrm{1}}} }{  \tmTopC  } $}
    We must show $ \mathin{ \ottkw{let} \, \ottsym{(}  \mathit{x}  \ottsym{,}  \mathit{y}  \ottsym{)}  \ottsym{=}  \substmeta  \ottsym{(}  \tmMeta_{{\mathrm{1}}}  \ottsym{)} \, \ottkw{in} \, \substmeta  \ottsym{(}  \tmMeta_{{\mathrm{2}}}  \ottsym{)} }{  \lrexp{  \tmTopC  }  } $.
    By the induction hypothesis, $ \judglr{ \Gamma }{ \tmMeta_{{\mathrm{1}}} }{  \tmTopC  } $.
    It follows $ \mathin{ \substmeta  \ottsym{(}  \tmMeta_{{\mathrm{1}}}  \ottsym{)} }{  \lrexp{  \tmTopC  }  } $ so $\substmeta  \ottsym{(}  \tmMeta_{{\mathrm{1}}}  \ottsym{)}  \stepsym^{*}   \tmTopC $.
    Then by the operational semantics $\ottkw{let} \, \ottsym{(}  \mathit{x}  \ottsym{,}  \mathit{y}  \ottsym{)}  \ottsym{=}  \substmeta  \ottsym{(}  \tmMeta_{{\mathrm{1}}}  \ottsym{)} \, \ottkw{in} \, \substmeta  \ottsym{(}  \tmMeta_{{\mathrm{2}}}  \ottsym{)}  \stepsym^{*}   \tmTopC $.
  \end{case}

  \begin{case}{$\tmMeta  \ottsym{=}  \ottkw{let} \, \ottnt{s}  \ottsym{=}  \tmMeta_{{\mathrm{1}}} \, \ottkw{in} \, \tmMeta_{{\mathrm{2}}}$ and $\tyAc  \ottsym{=}   \tmTopC $ where $ \judg{ \Gamma }{ \tmMeta_{{\mathrm{1}}} }{  \tmTopC  } $}
    We must show $ \mathin{ \ottkw{let} \, \ottnt{s}  \ottsym{=}  \substmeta  \ottsym{(}  \tmMeta_{{\mathrm{1}}}  \ottsym{)} \, \ottkw{in} \, \substmeta  \ottsym{(}  \tmMeta_{{\mathrm{2}}}  \ottsym{)} }{  \lrexp{  \tmTopC  }  } $.
    By the induction hypothesis, $ \judglr{ \Gamma }{ \tmMeta_{{\mathrm{1}}} }{  \tmTopC  } $.
    It follows $ \mathin{ \substmeta  \ottsym{(}  \tmMeta_{{\mathrm{1}}}  \ottsym{)} }{  \lrexp{  \tmTopC  }  } $ so $\substmeta  \ottsym{(}  \tmMeta_{{\mathrm{1}}}  \ottsym{)}  \stepsym^{*}   \tmTopC $.
    Then by the operational semantics $\ottkw{let} \, \ottnt{s}  \ottsym{=}  \substmeta  \ottsym{(}  \tmMeta_{{\mathrm{1}}}  \ottsym{)} \, \ottkw{in} \, \substmeta  \ottsym{(}  \tmMeta_{{\mathrm{2}}}  \ottsym{)}  \stepsym^{*}   \tmTopC $.
  \end{case}

  \begin{case}{$\tmMeta  \ottsym{=}   \tmApp{ \tmMeta_{{\mathrm{1}}} }{ \tmMeta_{{\mathrm{2}}} } $ and $\tyAc  \ottsym{=}   \tmTopC $ where $ \judg{ \Gamma }{ \tmMeta_{{\mathrm{1}}} }{  \tmTopC  } $}
    We must show $ \mathin{  \tmApp{ \substmeta  \ottsym{(}  \tmMeta_{{\mathrm{1}}}  \ottsym{)} }{ \substmeta  \ottsym{(}  \tmMeta_{{\mathrm{2}}}  \ottsym{)} }  }{  \lrexp{  \tmTopC  }  } $.
    By the induction hypothesis, $ \judglr{ \Gamma }{ \tmMeta_{{\mathrm{1}}} }{  \tmTopC  } $.
    It follows $ \mathin{ \substmeta  \ottsym{(}  \tmMeta_{{\mathrm{1}}}  \ottsym{)} }{  \lrexp{  \tmTopC  }  } $ so $\substmeta  \ottsym{(}  \tmMeta_{{\mathrm{1}}}  \ottsym{)}  \stepsym^{*}   \tmTopC $.
    Then by the operational semantics $ \tmApp{ \substmeta  \ottsym{(}  \tmMeta_{{\mathrm{1}}}  \ottsym{)} }{ \substmeta  \ottsym{(}  \tmMeta_{{\mathrm{2}}}  \ottsym{)} }   \stepsym^{*}   \tmTopC $.
  \end{case}

  \begin{case}{$\tmMeta  \ottsym{=}   \tmApp{ \tmMeta_{{\mathrm{1}}} }{ \tmMeta_{{\mathrm{2}}} } $ and $\tyAc  \ottsym{=}   \tmTopC $ where $ \judg{ \Gamma }{ \tmMeta_{{\mathrm{1}}} }{ \tyAv } $ and $ \judg{ \Gamma }{ \tmMeta_{{\mathrm{2}}} }{  \tmTopC  } $}
    We must show $ \mathin{  \tmApp{ \substmeta  \ottsym{(}  \tmMeta_{{\mathrm{1}}}  \ottsym{)} }{ \substmeta  \ottsym{(}  \tmMeta_{{\mathrm{2}}}  \ottsym{)} }  }{  \lrexp{  \tmTopC  }  } $.
    We have induction hypotheses for $\tmMeta_{{\mathrm{1}}}$ and $\tmMeta_{{\mathrm{2}}}$:
    \[
       \judglr{ \Gamma }{ \tmMeta_{{\mathrm{1}}} }{ \tyAv }  \ \textand\  \judglr{ \Gamma }{ \tmMeta_{{\mathrm{2}}} }{  \tmTopC  } 
    \]
    It follows $ \mathin{ \substmeta  \ottsym{(}  \tmMeta_{{\mathrm{1}}}  \ottsym{)} }{  \lrexp{ \tyAv }  } $ and $ \mathin{ \substmeta  \ottsym{(}  \tmMeta_{{\mathrm{2}}}  \ottsym{)} }{  \lrexp{  \tmTopC  }  } $.
    Thus, $\substmeta  \ottsym{(}  \tmMeta_{{\mathrm{2}}}  \ottsym{)}  \stepsym^{*}   \tmTopC $.
    Let $ \mathin{ \valMeta_{{\mathrm{1}}} }{  \lrval{ \tyAv }  } $.
    By \Lref{lem:apdx:logsem:mon-bind}, it suffices to show:
    \[  \mathin{  \tmApp{ \valMeta_{{\mathrm{1}}} }{ \substmeta  \ottsym{(}  \tmMeta_{{\mathrm{2}}}  \ottsym{)} }  }{  \lrexp{  \tmTopC  }  }  \]
    By the operational semantics, $ \tmApp{ \valMeta }{ \substmeta  \ottsym{(}  \tmMeta_{{\mathrm{2}}}  \ottsym{)} }   \stepsym^{*}   \tmTopC $.
  \end{case}

  \begin{case}{$\tmMeta  \ottsym{=}   \tmBigJoin{ \mathit{x} }{ \tmMeta_{{\mathrm{1}}} }{ \tmMeta_{{\mathrm{2}}} } $ where $ \judg{ \Gamma }{ \tmMeta_{{\mathrm{1}}} }{  \tmTopC  } $}
    We must show $ \mathin{  \tmBigJoin{ \mathit{x} }{ \substmeta  \ottsym{(}  \tmMeta_{{\mathrm{1}}}  \ottsym{)} }{ \substmeta  \ottsym{(}  \tmMeta_{{\mathrm{2}}}  \ottsym{)} }  }{  \lrexp{  \tmTopC  }  } $.
    By the induction hypothesis, $ \judglr{ \Gamma }{ \tmMeta_{{\mathrm{1}}} }{  \tmTopC  } $.
    It follows $ \mathin{ \substmeta  \ottsym{(}  \tmMeta_{{\mathrm{1}}}  \ottsym{)} }{  \lrexp{  \tmTopC  }  } $ so $\substmeta  \ottsym{(}  \tmMeta_{{\mathrm{1}}}  \ottsym{)}  \stepsym^{*}   \tmTopC $.
    Then by the operational semantics $ \tmBigJoin{ \mathit{x} }{ \substmeta  \ottsym{(}  \tmMeta_{{\mathrm{1}}}  \ottsym{)} }{ \substmeta  \ottsym{(}  \tmMeta_{{\mathrm{2}}}  \ottsym{)} }   \stepsym^{*}   \tmTopC $.
  \end{case}
\end{proof}

\begin{lemma}[Adequacy]
  \label{lem:apdx:logsem:adequacy}

  If $\logapx{\valMeta}{\tmMeta}$ then $ \returnsany{ \tmMeta } $.
\end{lemma}
\begin{proof}
  Suppose $\logapx{\valMeta}{\tmMeta}$.
  We must show $\tmMeta  \stepsym^{*}   \tmTopC $ or $\mathexists{\valMeta'}{\tmMeta  \stepsym^{*}  \valMeta'}$.
  It is immediate from the definition of formula assignment (specifically the rule
  \rulename{TBotV}) that there exists some value formula $\tyAv$ such that
  $ \judg{  \envEmpty  }{ \valMeta }{ \tyAv } $. 
  By assumption, we have $ \judg{  \envEmpty  }{ \tmMeta }{ \tyAv } $ and thus $ \judglr{  \envEmpty  }{ \tmMeta }{ \tyAv } $ thanks
  to the Fundamental Property (\Lref{lem:apdx:log-sem:fund-prop}).
  The definition of $ \lrexp{ \tyAv } $ then gives us a result $\resMeta$
  such that $\tmMeta  \stepsym^{*}  \resMeta$ and $ \mathin{ \resMeta }{  \lrres{ \tyAv }  } $.
  Examining the definition of $ \lrres{ \tyAv } $ reveals that $\resMeta$ must
  either be a value or $ \tmTopC $.
\end{proof}

\subsection{Results}

\begin{theorem}[Monotonicity]
  For any context $\ottnt{C}$ and $\logapx{\tmMeta}{\tmMeta'}$,
  we have $\logapx{\ottnt{C}  \ottsym{[}  \tmMeta  \ottsym{]}}{\ottnt{C}  \ottsym{[}  \tmMeta'  \ottsym{]}}$.
\end{theorem}
\begin{proof}
  Immediate from \Lref{lem:apdx:log-sem:compos}.
\end{proof}

\begin{lemma}[Soundness]
  If $\tmMeta  \stepsym^{*}  \tmMeta'$ then $\logapx{\tmMeta'}{\tmMeta}$.
\end{lemma}
\begin{proof}
  Induction on $\tmMeta  \stepsym^{*}  \tmMeta'$, applying Subject Expansion
  (\Lref{lem:apdx:log-sem:subj-exp}) at each step.
\end{proof}

\begin{theorem}
  If $\logapx{\tmMeta_{{\mathrm{1}}}}{\tmMeta_{{\mathrm{2}}}}$
  then $\ctxapx{\tmMeta_{{\mathrm{1}}}}{\tmMeta_{{\mathrm{2}}}}$.
\end{theorem}
\begin{proof}
  Consider a context $\ottnt{C}$ such that $ \returns{ \ottnt{C}  \ottsym{[}  \tmMeta_{{\mathrm{1}}}  \ottsym{]} }{ \resMeta } $ where $\resMeta \neq  \tmBotC $.
  We must show $ \returnsany{ \ottnt{C}  \ottsym{[}  \tmMeta_{{\mathrm{2}}}  \ottsym{]} } $.

  We deduce the following:
  \[
  \begin{array}{rcll}
     \valBotV  &\logapxsym & \resMeta     &\text{Straightforward from the formula assignment rules.} \\
    &\logapxsym & \ottnt{C}  \ottsym{[}  \tmMeta_{{\mathrm{1}}}  \ottsym{]} &\text{Soundness} \\
    &\logapxsym & \ottnt{C}  \ottsym{[}  \tmMeta_{{\mathrm{2}}}  \ottsym{]} &\text{Monotonicity} \\
  \end{array}
  \]
  Therefore we may apply Adequacy to complete the proof.
\end{proof}

\section{Proofs: Domain Theory}
\label{sec:apdx:domains}
\subsection{Preliminaries}
A \emph{preorder} $(X, \sqsubseteq)$ is a set $X$ together with a binary relation
$\sqsubseteq$ over $X$ that is reflexive and transitive.
Let $Y$ be a subset of $X$.
An \emph{upper bound} of $Y$ is an element of $X$ which is greater than every
element of $Y$.
The \emph{least upper bound} or \emph{join} of $Y$, if it exists, is an upper
bound of $Y$ that is less than every other upper bound.
It is denoted $\bigsqcup Y$.

Consider a preorder $(x, \sqsubseteq)$ and a subset $Y \subseteq X$.
The subset $Y$ is \emph{downward closed} iff
\[
  \mathforall{x \in X, y \in Y}{x \sqsubseteq y \Rightarrow x \in Y}
\]
Moreover, we say that $Y$ is \emph{directed} when we have:
\[
  \mathforall{y_1, y_2 \in Y}{\mathexists{y \in Y}{y_1, y_2 \sqsubseteq y}}
\]
We call $Y$ an \emph{ideal} of $X$ iff it is non-empty,
downward closed, and directed.
The set of ideals of $X$ is written $\ideals{X}$.
Given an element $x \in X$, the set
$\prinideal{x} = \mathsetcomp{y \in X}{y \sqsubseteq x}$ is the \emph{principal ideal} of $x$.

A \emph{partial order} is a preorder $(X, \sqsubseteq)$ which is antisymmetric, that is:
\[
\mathforall{x_1, x_2 \in X}{x_1 \sqsubseteq x_2 \textand x_2 \sqsubseteq x_1 \Rightarrow x_1 = x_2}
\]
The partial order is \emph{directed complete} iff every non-empty directed
subset has a least upper bound (or \emph{join}).
It is \emph{bounded complete} iff every non-empty subset with any upper bound has
a least upper bound.

An element $k \in X$ is \emph{compact} (elsewhere sometimes called \emph{finite}) iff
for all directed subsets $Y \subseteq X$ we have:
\[ k \sqsubseteq \bigsqcup Y \Rightarrow \mathexists{y \in Y}{k \sqsubseteq y} \]
The set of the compact elements of $X$ is written $\compact{X}$.
Given an element of a partial order $x \in X$, the set of compact elements below $x$ is written:
\[
  \downclosk{x} = \mathsetcomp{k \in \compact{X}}{k \sqsubseteq x}
\]
The partially ordered set $X$ is \emph{algebraic} iff every element is
the least upper bound of the set of compact elements beneath it.
That is, for all $x \in X$ we have $x = \bigsqcup (\downclosk{x})$.

A \emph{domain} (elsewhere called a \emph{Scott predomain}) is a partial order
that is directed complete, bounded complete, and algebraic.
A preorder $(B, \sqsubseteq)$ is a \emph{finitary basis} iff $B$ is countable and
every non-empty finite subset with an upper bound has a least upper bound.
The \emph{ideal completion} of a basis,  $\ideals{B}$, forms a
domain in which the compact elements are precisely the principal ideals of elements of $B$.

Given finitary bases $(A, \sqsubseteq_A)$ and $(B, \sqsubseteq_B)$,
an \emph{approximable mapping} from $A$ to $B$ is a binary relation
$R \subseteq A \times B$ such that:
\begin{enumerate}
\item \emph{Totality.} $\mathforall{a \in A}{\mathexists{b \in B}{a \R b}}$
\item \emph{Downward Closure.} If $a \R b$ and $b' \sqsubseteq_{B} b$ then $a \R b'$.
\item \emph{Weakening.} If $a \R b $ and $a \sqsubseteq_{A} a'$ then $a' \R b$.
\item \emph{Directedness.} If $a \R b_1$ and $a \R b_2 $ then $\mathexists{b \in B}{b_1, b_2 \sqsubseteq_B b \textand a \R b}$.
\end{enumerate}
We write $\apxmap{A}{B}$ to refer to the set of approximable mappings from $A$ to $B$.
It forms a partial order under the subset relation.

Given two partial orders $(X, \sqsubseteq_X)$ and $(Y, \sqsubseteq_Y)$,
a function $f : X \to Y$ is \emph{monotone} iff
\[ \mathforall{x_1, x_2 \in X}{x_1 \sqsubseteq_X x_2 \Rightarrow f(x_1) \sqsubseteq_Y f(x_2)} \]
Two partial orders are \emph{isomorphic} if and only if there exists a monotone
bijection between them.

Given two domains $(D, \sqsubseteq_D)$ and $(E, \sqsubseteq_E)$,
a function $f : D \to E$ is \emph{continuous} iff
for all directed subsets $X \subseteq D$, we have
$\bigsqcup f(X) = f(\bigsqcup X)$.
Every continuous function is monotone.
The space of continuous functions is written $\contfun{D}{E}$.

\begin{proposition}
  \label{prop:apdx:densem:apxmap-iso}
  Let $A$ and $B$ be finitary bases.
  Then we have an isomorphism of partial orders:
  \[ \apxmap{A}{B} \cong \contfun{\ideals{A}}{\ideals{B}} \]
\end{proposition}
\begin{proof}
  See Theorem~2.6 from \citet{cartwright16}.
\end{proof}

The domain constructors $(\mathemdash)_\bot$ and $(\mathemdash)_\top$ insert a new
least and greatest element into a domain respectively.
The binary constructor $\mathemdash + \mathemdash$ performs disjoint union,
while $\mathemdash \times \mathemdash$ represents cartesian product.
These operations can be performed on bases or on domains.
With respect to ideal completion, they behave as follows.
\begin{proposition}
  \label{lem:apdx:densem:ideal-comp-ctors}
  \leavevmode
  \begin{enumerate}
  \item \( \ideals{B_\bot} \cong \ideals{B}_\bot \)
  \item \( \ideals{B_\top} \cong \ideals{B}_\top \)
  \item \( \ideals{A + B} \cong \ideals{A} + \ideals{B} \)
  \item \( \ideals{A \times B} \cong \ideals{A} \times \ideals{B} \)
  \end{enumerate}
\end{proposition}

For denoting sets, we also make use of the \emph{Hoare powerdomain}.
\begin{definition}[Hoare Powerdomain]
  Given a domain $D$, the Hoare powerdomain $\powdomh{D}$ is defined below and
  forms a domain ordered by subset inclusion.
  \[
  \powdomh{D} = \mathsetcomp{X \subseteq \compact{D}}{\text{$X$ is downward closed}}
  \]
\end{definition}

\subsection{Domain Equation}
We would like to show that $D = \ideals{\synvtype}$ is a solution to the domain
equation below.
\begin{equation}
  \label{eqn:apdx:densem:dom-eqn}
  \domval  \cong (\ideals{\symset} + \domval \times \domval + \powdomh{\domval} + (\contfun{\domval}{\domval_{\bot\top}}))_\botv
\end{equation}

\begin{definition}
We define the following sets of formulae, called \emph{components}.
\begin{itemize}
  \item $\synvtypepair = \mathsetcomp{\ottsym{(}  \tyAv_{{\mathrm{1}}}  \ottsym{,}  \tyAv_{{\mathrm{2}}}  \ottsym{)}}{\tyAv_{{\mathrm{1}}} \in \synvtype \textand \tyAv_{{\mathrm{2}}} \in \synvtype}$
  \item $\synvtypeset = \mathsetcomp{\ottsym{\{}  \tyAv_{\ottmv{i}}  \ottsym{\mbox{$\mid$}}   \mathin{ \mathit{i} }{ \mathit{I} }   \ottsym{\}}}{\mathforall{ \mathin{ \mathit{i} }{ \mathit{I} } }{\tyAv_{\ottmv{i}} \in \synvtype}}$
  \item $\synvtypefun = \mathsetcomp{ \tyFun{  \mathin{ \mathit{i} }{ \mathit{I} }  }{ \tyAv_{\ottmv{i}} }{ \tyAc_{\ottmv{i}} } }{\mathforall{ \mathin{ \mathit{i} }{ \mathit{I} } }{\tyAv_{\ottmv{i}} \in \synvtype \textand \tyAc_{\ottmv{i}} \in \synctype }}$
\end{itemize}
\end{definition}

\begin{lemma}
  \label{lem:apdx:densem:synvtype-disj}
  \( \synvtype \cong (\symset + \synvtypepair + \synvtypeset + \synvtypefun)_\botv \)
\end{lemma}
\begin{proof}
  A direct consequence of the definition of formulae.
\end{proof}

\begin{lemma}
  \label{lem:apdx:densem:iso-pair}
  \(
  \ideals{\synvtypepair} \cong \ideals{\synvtype} \times \ideals{\synvtype}
  \)
\end{lemma}
\begin{proof}
  It is not hard to see $\synvtypepair \cong \synvtype \times \synvtype$.
  By \Pref{lem:apdx:densem:ideal-comp-ctors} it follows
  $\ideals{\synvtypepair} \cong \ideals{\synvtype} \times \ideals{\synvtype}$.
\end{proof}
\begin{lemma}
  \label{lem:apdx:densem:iso-set}
  \(
  \ideals{\synvtypeset} \cong \powdomh{\ideals{\synvtype}}
  \)
\end{lemma}
\begin{proof}
  Let $X \in \ideals{\synvtypeset}$.
  Define $f(X) = \mathsetcomp{\prinideal{\tyAv}}{\ottsym{\{}  \tyAv  \ottsym{\}} \in X}$.
  We can see that $f(X)$ is a downward closed set of compact elements since $X$ is
  downward closed and $f(X)$ is defined as a set of principal ideals.

  \paragraph{Monotonicity}
  Consider the ideals $X$ and $Y$ such that $X \subseteq Y$.
  We must show $f(X) \subseteq f(Y)$.
  Let $\prinideal{\tyAv}$ be an element of $f(X)$ were $\ottsym{\{}  \tyAv  \ottsym{\}} \in X$.
  To prove $\prinideal{\tyAv} \in f(Y)$, it suffices to show $\ottsym{\{}  \tyAv  \ottsym{\}} \in Y$.
  This follows from the assumption that $X \subseteq Y$.

  \paragraph{Injectivity}
  Consider the ideals $X, Y \in \ideals{\synvtypeset}$ such that $f(X) = f(Y)$.
  For contradiction, suppose $X \neq Y$.
  Without loss of generality, we assume there exists some $\tyAv \in X$ such
  that $\tyAv \notin Y$.
  From the definition of $\synvtypeset$, the formula $\tyAv$ has the form
  $\ottsym{\{}  \tyAv_{\ottmv{i}}  \ottsym{\mbox{$\mid$}}   \mathin{ \mathit{i} }{ \mathit{I} }   \ottsym{\}}$.
  From the fact that $Y$ does not contain $\tyAv$ and the properties of
  ideals, we can deduce that there exists some $ \mathin{ \mathit{i} }{ \mathit{I} } $ such that
  $\ottsym{\{}  \tyAv_{\ottmv{i}}  \ottsym{\}} \in X$ but $\ottsym{\{}  \tyAv_{\ottmv{i}}  \ottsym{\}} \notin Y$.
  From the definition of $f$, it follows $\prinideal{\tyAv_{\ottmv{i}}} \in f(X)$ but
  $\prinideal{\tyAv_{\ottmv{i}}} \notin f(Y)$.
  This contradicts the assumption that $f(X) = f(Y)$.

  \paragraph{Surjectivity}
  Consider arbitrary $Y \in \powdomh{\ideals{\synvtype}}$.
  Let $X = \mathsetcomp{\ottsym{\{}  \tyAv_{\ottmv{i}}  \ottsym{\mbox{$\mid$}}   \mathin{ \mathit{i} }{ \mathit{I} }   \ottsym{\}}}{\mathforall{ \mathin{ \mathit{i} }{ \mathit{I} } }{\prinideal{\tyAv_{\ottmv{i}}} \in Y}}$.
  Using the fact that $Y$ is downward closed, it is straightforward to check that
  $X$ is an ideal.
  Note that we have $\ottsym{\{}  \tyAv  \ottsym{\}} \in X$ iff $\prinideal{\tyAv} \in Y$;
  it follows $f(X) = Y$.
\end{proof}
\begin{lemma}
  \label{lem:apdx:densem:iso-fun}
  \(
  \ideals{\synvtypefun} \cong \contfun{\ideals{\synvtype}}{\ideals{\synvtype}_{\bot\top}}
  \)
\end{lemma}
\begin{proof}
  First, note $\ideals{\synctype} \cong \ideals{\synvtype_{\bot\top}} \cong \ideals{\synvtype}_{\bot\top}$.
  Using this fact and \Pref{prop:apdx:densem:apxmap-iso}, it suffices to show
  \(
  \ideals{\synvtypefun} \cong \apxmap{\synvtype}{\synctype}
  \).
  Let $X \in \ideals{\synvtypefun}$.
  Define $f(X) = \mathsetcomp{(\tyAv, \tyAc)}{ \tyFunOne{ \tyAv }{ \tyAc }  \in X}$.
  It is easy to verify that most of the properties of
  approximable mappings for $f(X)$ follow from the properties of ideals that $X$
  has.
  Totality is less obvious: it requires recognizing that the empty function
  formula (which is included in $X$) is equivalent to $ \tyFunOne{ \tyAv }{  \tmBotC  } $ for all
  $\tyAv$. \todo{That could be a separate lemma}
  As a consequence we have
  $\mathforall{\tyAv}{(\tyAv,  \tmBotC ) \in f(X)}$.
  It remains to show that $f$ is monotone, injective, and surjective.

  \paragraph{Monotonicity}
  Consider the ideals $X$ and $Y$ such that $X \subseteq Y$.
  It is easy to see that the set $f(X)$ is included within $f(Y)$.

  \paragraph{Injectivity}
  Consider the ideals $X, Y \in \ideals{\synvtypefun}$ such that $f(X) = f(Y)$.
  For contradiction, suppose $X \neq Y$.
  Without loss of generality, we assume there exists some $\tyAv \in X$ such
  that $\tyAv \notin Y$.
  From the definition of $\synvtypefun$, the formula $\tyAv$ has the form
  $ \tyFun{  \mathin{ \mathit{i} }{ \mathit{I} }  }{ \tyAv_{\ottmv{i}} }{ \tyAc_{\ottmv{i}} } $.
  From the fact that $Y$ does not contain $\tyAv$ and the properties of ideals, we can deduce that there exists some $ \mathin{ \mathit{i} }{ \mathit{I} } $ such that $ \tyFunOne{ \tyAv_{\ottmv{i}} }{ \tyAc_{\ottmv{i}} }  \notin Y$.
  It follows that $(\tyAv_{\ottmv{i}}, \tyAc_{\ottmv{i}})$ is in $f(X)$ but not in $f(Y)$.
  This contradicts our assumption $f(X) = f(Y)$.

  \paragraph{Surjectivity}
  Let $R \in \apxmap{\synvtype}{\synctype}$.
  Let $X = \mathsetcomp{ \tyFun{  \mathin{ \mathit{i} }{ \mathit{I} }  }{ \tyAv_{\ottmv{i}} }{ \tyAc_{\ottmv{i}} } }{\mathforall{ \mathin{ \mathit{i} }{ \mathit{I} } }{\tyAv_{\ottmv{i}} \R \tyAc_{\ottmv{i}}}}$.
  We first check that $X$ is an ideal.
  \begin{itemize}
  \item \emph{Non-empty.} The trivial 0-case function formula is included in $X$ by definition.
  \item \emph{Downward closed.}
    Suppose $\tyAv  \ottsym{=}   \tyFun{  \mathin{ \mathit{i} }{ \mathit{I} }  }{ \tyAv_{\ottmv{i}} }{ \tyAc_{\ottmv{i}} }  \in X$ and $\tyAv'  \ottsym{=}   \tyFun{  \mathin{ \mathit{j} }{ \mathit{J} }  }{ \tyAv'_{\ottmv{i}} }{ \tyAc'_{\ottmv{i}} } $
    and $ \tyapx{ \tyAv' }{ \tyAv } $.
    Fix $ \mathin{ \mathit{j} }{ \mathit{J} } $.
    To show $\tyAv' \in X$, it is enough to prove $\tyAv'_{\ottmv{j}} \R \tyAc'_{\ottmv{j}}$.

    Note that by \Lref{lem:apdx:log-sem:subtyping-inv}, there exists $ \mathit{I'}  \subseteq  \mathit{I} $ such that
    $ \tyapx{  \tyBigJoin{  \mathin{ \mathit{i} }{ \mathit{I'} }  }{ \tyAv_{\ottmv{i}} }  }{ \tyAv'_{\ottmv{j}} } $ and $ \tyapx{ \tyAc'_{\ottmv{j}} }{  \tyBigJoin{  \mathin{ \mathit{i} }{ \mathit{I'} }  }{ \tyAc_{\ottmv{i}} }  } $.
    Since $\tyAv \in X$ we also have $ \tyFun{  \mathin{ \mathit{i} }{ \mathit{I'} }  }{ \tyAv_{\ottmv{i}} }{ \tyAc_{\ottmv{i}} }  \in X$.
    By the properties of ideals, $\ottsym{(}   \tyBigJoin{  \mathin{ \mathit{i} }{ \mathit{I'} }  }{ \tyAv_{\ottmv{i}} }   \ottsym{)} \R \ottsym{(}   \tyBigJoin{  \mathin{ \mathit{i} }{ \mathit{I'} }  }{ \tyAc_{\ottmv{i}} }   \ottsym{)}$.
    Again using the properties of ideals, it follows $\tyAv'_{\ottmv{j}} \R \tyAc'_{\ottmv{j}}$.
  \item \emph{Directed.}
    Suppose $ \tyFun{  \mathin{ \mathit{i} }{ \mathit{I} }  }{ \tyAv_{\ottmv{i}} }{ \tyAc_{\ottmv{i}} }  \in X$ and $ \tyFun{  \mathin{ \mathit{i} }{ \mathit{I'} }  }{ \tyAv_{\ottmv{i}} }{ \tyAc_{\ottmv{i}} }  \in X$.
    We must show $ \tyFun{  \mathin{ \mathit{i} }{  \mathit{I}  \cup  \mathit{I'}  }  }{ \tyAv_{\ottmv{i}} }{ \tyAc_{\ottmv{i}} }  \in X$.
    This is clear from the definition of $X$.
  \end{itemize}
  Now note we have $ \tyFunOne{ \tyAv }{ \tyAc }  \in X$ iff $\tyAv \R \tyAc$;
  it follows $f(X) = R$.
\end{proof}

\begin{theorem}
  \label{lem:apdx:densem:iso}
  $D = \ideals{\synvtype}$ is a solution to the domain equation
  \eqref{eqn:apdx:densem:dom-eqn}.
  That is,
\[
\ideals{\synvtype}  \cong (\ideals{\symset} + \ideals{\synvtype} \times \ideals{\synvtype} + \powdomh{\ideals{\synvtype}} + (\contfun{\ideals{\synvtype}}{\ideals{\synvtype}_{\bot\top}}))_\botv
\]
\end{theorem}
\begin{proof}
  We first derive the following.
  \[
  \begin{array}{lcll}
    \ideals{\synvtype} &\cong & \ideals{(\symset + \synvtypepair + \synvtypeset + \synvtypefun)_\botv} & \text{\Lref{lem:apdx:densem:synvtype-disj}} \\
  &\cong & (\ideals{\symset} + \ideals{\synvtypepair} + \ideals{\synvtypeset} + \ideals{\synvtypefun})_\botv
    & \text{\Pref{lem:apdx:densem:ideal-comp-ctors}}
  \end{array}
  \]
  Lemmas~\ref{lem:apdx:densem:iso-pair}-\ref{lem:apdx:densem:iso-fun} then complete the proof.
\end{proof}

\fi

\end{document}